\documentclass[a4paper,11pt]{article}

\def\biblio{\bibliography{bibliography}}

\pdfoutput=1 

\usepackage{jheppub}

\usepackage[usenames,dvipsnames]{xcolor}
\usepackage[T1]{fontenc} 

\usepackage{amsmath}
\usepackage{amsthm}
\usepackage{amssymb}
\usepackage{float}
\usepackage{braket}
\usepackage{graphicx}
\usepackage[utf8]{inputenc}
\usepackage[T1]{fontenc}
\usepackage{lmodern}
\usepackage{setspace}
\usepackage{tensor}
\usepackage{tikz-cd,tikz} 
\usetikzlibrary{decorations.pathreplacing}
\usepackage{color,colortbl,hhline}
\usepackage{hyperref}

\usepackage{cleveref}
\hypersetup{
    colorlinks,
    citecolor=black,
    filecolor=black,
    linkcolor=black,
    urlcolor=black
}
\usepackage{epic}
\usepackage{pict2e}
\usepackage{bbm}
\usepackage{bm}
\usepackage{dashrule}
\definecolor{lightergray}{rgb}{.9,.9,.9}
\definecolor{lighterblue}{rgb}{.9,.9,1}
\definecolor{lightergreen}{rgb}{.9,1,.9}
\definecolor{lightgray}{rgb}{.75,.75,.75}

\usepackage{subcaption}
\usepackage{ytableau}
\usepackage{youngtab}

\makeatletter
\newcommand{\rref}[1]{\protected@edef\@currentlabel{#1}}
\makeatother

\DeclareMathAlphabet{\mathpzc}{OT1}{pzc}{m}{it}

\DeclareMathOperator{\Id}{Id}
\DeclareMathOperator{\ii}{i}

\theoremstyle{plain}
\newtheorem{theorem}{Theorem}[section]
\newtheorem{lemma}[theorem]{Lemma}

\newtheorem*{theorem*}{Theorem}

 \numberwithin{equation}{section}

\let\emptyset\varnothing

\newcommand{\ZZ}{\mathbb{Z}}

\newcommand{\CC}{\mathbb{C}}

\newcommand{\gl}{\mathfrak{gl}}

\newcommand{\es}{{\emptyset}}
\newcommand{\fullset}{{\bar\es}}

\newcommand{\groupG}{\ensuremath{\mathsf{G}}}
\newcommand{\groupB}{\ensuremath{\mathsf{B}}}
\newcommand{\groupP}{\ensuremath{\mathsf{P}}}
\newcommand{\groupN}{\ensuremath{\mathsf{N}}}

\newcommand{\Gl}{{\mathsf{GL}}}
\newcommand{\Sl}{{\mathsf{SL}}}
\newcommand{\SO}{{\mathsf{SO}}}
\newcommand{\gm}{{\mathsf{m}}}
\newcommand{\gn}{{\mathsf{n}}}

\renewcommand{\sl}{{\mathfrak{sl}}}
\newcommand{\so}{{\mathfrak{so}}}

\newcommand{\glmn}{{\mathfrak{\gl}_{\gm|\gn}}}

\newcommand{\Y}{{\mathcal{Y}}}

\newcommand{\CF}{{\mathcal{F}}}

\newcommand{\Tau}{{\mathcal{T}}}
\newcommand{\aast}{a_\ast}

\newcommand{\Ac}{\mathsf{A}}
\newcommand{\Bc}{\mathsf{B}}
\newcommand{\bc}{\mathsf{b}}

\newcommand{\algg}{{\mathfrak{g}}}

\newcommand{\be}{\begin{eqnarray}}
\newcommand{\ee}{\end{eqnarray}}

\newcommand{\eg}{{\it e.g. }}
\newcommand{\ie}{{\it i.e. }}
\newcommand{\cf}{{\it cf. }}
\newcommand{\etc}{{\it etc}}
\newcommand{\via}{{\it via }}
\newcommand{\rhs}{{r.h.s. }}
\newcommand{\lhs}{{l.h.s. }}
\newcommand{\wrt}{{w.r.t. }}
\newcommand{\apriori}{{\it a-priori }}
\newcommand{\Plucker}{{Pl{\"u}cker }}
\newcommand{\Backlund}{{B{\"a}cklund }}
\newcommand{\Schrodinger}{{Schr{\"o}dinger }}
\def \ba {\begin{aligned}}
\def \ea {\end{aligned}}

\newcommand{\Qs}{q} 
\newcommand{\Q}{Q} 
\newcommand{\Ts}{t} 
\newcommand{\T}{T} 
\newcommand{\q}{\ensuremath{\mathsf{q}}} 
\newcommand{\spm}{z} 
\newcommand{\spa}{u} 
\newcommand{\rank}{r} 
\newcommand{\CD}{\mathcal{D}} 
\newcommand{\orthb}{\varepsilon} 

\newcommand{\evp}{\ensuremath{\lambda}}
\newcommand{\syman}{\ensuremath{\q}}
\newcommand{\EAff}{E_{\alpha_{0}}}
\newcommand{\Cox}{\ensuremath{h}}
\newcommand{\WG}{\ensuremath{\mathcal W}}
\newcommand{\qV}{V}
\newcommand{\qS}{\psi}
\newcommand{\qC}{\eta}

\newcommand{\Vandermonde}{\Delta}

\newcommand{\dressing}{\sigma}
\newcommand{\metric}{g}

\newcommand{\vecf}[1]{#1}
\newcommand{\qZ}{\zeta}

\makeatletter
\def\@fpheader{\relax}
\makeatother
\makeatletter
\tikzset{
  edge node/.code={%
      \expandafter\def\expandafter\tikz@tonodes\expandafter{\tikz@tonodes #1}}}
\makeatother
\tikzset{
  subseteq/.style={
    draw=none,
    edge node={node [sloped, allow upside down, auto=false]{$\subset$}}},
  Subseteq/.style={
    draw=none,
    every to/.append style={
      edge node={node [sloped, allow upside down, auto=false]{$\subset$}}}
  }
}

\title{Extended systems of Baxter Q-functions and fused flags I: simply-laced case}
\author{Simon Ekhammar$^{a}$,}
\author{Hongfei Shu$^{b}$,}
\author{Dmytro Volin$^{a,b}$}

\affiliation[a]{Department of Physics and Astronomy,\\
Uppsala University, Box 516, SE-751 20 Uppsala, Sweden}
\affiliation[b]{Nordita, KTH Royal Institute of Technology and Stockholm University,\\
Roslagstullsbacken 23, SE-106 91 Stockholm, Sweden}

\emailAdd{simon.ekhammar@physics.uu.se}
\emailAdd{hongfei.shu@su.se}
\emailAdd{dmytro.volin@physics.uu.se}

\abstract{
The spectrum of integrable models is often encoded in terms of commuting functions of a spectral parameter that satisfy functional relations. We propose to describe this commutative algebra in a covariant way by means of the extended Q-system that comprise Q-vectors in each of the fundamental representations of the (Langlands dual of) the underlying symmetry algebra. These Q-vectors turn out to parameterise a collection of complete flags which are fused with one another in a particular way. We show that the fused flag is a finite-difference oper in a particular gauge, explicit identification depends on a choice of a Coxeter element.

The paper considers the case of simple Lie algebras with a simply-laced Dynkin diagram. For the $A_r$ series, the construction coincides with already known results in the literature. We apply the proposed formalism to the case of the $D_r$ series and the exceptional algebras $E_r$, $r=6,7,8$. In particular, we solve Hirota bilinear equations in terms of Q-functions and give the explicit character solution of the extended Q-system in the $D_r$ case. We also show how to build up the extended Q-system of $D_r$ type starting either from vectors, by a procedure similar to the $A_r$ scenario which however constructs a fused flag of isotropic spaces, or from pure spinors, \via fused Fierz relations.

Finally, for the case of rational, trigonometric, and elliptic spin chains, we propose an explicit ansatz for the analytic structure of Q-functions of the extended Q-system. We conjecture that the extended Q-system constrained in such a way is always in bijection with the Bethe algebra of commuting transfer matrices of these models and moreover can be used to show that the Bethe algebra has a simple joint spectrum.
}

\begin{document} 
\maketitle

\newpage
\section{Introduction}
Existence of commuting charges is a landmark feature of integrable systems which is even used sometimes as their defining property. While this point of view works for \eg  Liouville integrability, it is not sufficient in other situations where one needs also a controlled way to construct the commuting charges. Take for instance a quantum system with a finite-dimensional Hilbert space. If the physical Hamiltionian has a non-degenerate spectrum, we can construct a maximal commutative algebra by simply taking powers of the Hamiltonian meaning that commuting charges are always at hand in principle, while integrability should be a more stringent property.

A good example of a controlled construction of the commuting charges is the quantum inverse scattering method \cite{springerlink:10.1007/BF01018718,KBI93} or an equivalent procedure to generate transfer matrices of integrable spin chains. When performed properly it should provide us with a maximal commutative subalgebra of Hilbert space endomorphisms which shall be called the Bethe algebra. Examples are the Bethe algebras in trigonometric XXZ-type models that are representations of $U_\q(\hat\algg)$ and rational XXX-type models that are representations of the Yangian $\Y(\algg)$ \cite{Kulish:1983md,Jimbo:1985zk,Drinfeld1988,Chari:1994pz,Faddeev:1996iy}. Being generated in the particular manner, the Bethe algebra has a specific internal structure and a natural objective is to comprehensively describe it.

In this paper, we describe the Bethe algebra in terms of a collection of Baxter Q-functions. Whereas this approach is already understood for the cases when $\algg=\sl_{\rank+1}$ it is less developed for other Lie algebras. We fill in  gaps in this knowledge by developing a concept of the extended Q-system. This paper shall focus on simply-laced cases only.

While we refer to the XXX/XXZ-type spin chains for concrete examples, our results are of universal nature. They do not depend on a particular physical model but only on the representation theory of $\hat\algg$ (more accurately, of its Langlands dual). For instance, we use ODE/IM correspondence to derive the main statements, and the employed differential equation does not correspond to the mentioned spin chains except in a certain limit.
\newline
\newline
Universal nature of the Q-system already manifests in the fact that equivalent questions and answers we face while studying it emerge in a broad variety of research directions. We will be able to exhibit only some of the connections, and for those that do not get enough of discussion in the main text, we refer to: \cite{maulik2018quantum,Nekrasov:2013xda} for Bethe/gauge correspondence and quantum geometry, \cite{Frenkel:2005pa,Chervov:2006xk} for geometric Langlands, \cite{Bazhanov:2001xm} for CFT, \cite{Lacroix:2018njs} for Gaudin models~\footnote{The cited works are not necessarily the first ones in the corresponding subject but contain a connection to our discussion and/or can be used as reviews.}.

\subsection{Concept of a Q-system}
As the first illustration, let us consider one of the most famous equations in the theory of quantum integrable systems -- Baxter's TQ-relation \cite{Baxter:1971cs,Baxter:1972hz}
\be
\label{eq:TQ}
\Ts\,\Qs=\phi^- \Qs^{++}+\phi^+ \Qs^{--}.
\ee
This relation appears in numerous studies. In some, like spin chain models,  $\Ts$ and $\Qs$ have the operatorial meaning of transfer matrices generating commuting charges and in some, like TBA \cite{Zamolodchikov:1989cf} and ODE/IM correspondence \cite{Dorey:1998pt,Bazhanov:1998wj}, $t$ and $q$ are convenient functions of a spectral parameter having a different physical origin. Eventually, it is the analytic properties of $\Ts,\Qs,\phi$ that decide which system the TQ-relation describes. Demanding concrete analytic properties often goes under the name of analytic Bethe Ansatz \cite{Reshetikhin1983}.

For XXZ-type models based on $U_\q(\hat\algg)$ (with $\algg=\sl_2$ for the above TQ-relation), $\Ts,\Qs,\phi$ are  polynomials in multiplicative spectral parameter $\spm$ while the spectral parameter shift is defined as $f^{\pm}(\spm):=f(\q^{\pm 1/2} \spm)$. For XXX-type models based on $\Y(\algg)$, $\Ts,\Qs,\phi$ are polynomials in an additive spectral parameter $\spa$, and then $f^{\pm}:=f(\spa\pm\frac{\hbar}{2})$, with $\hbar\in\CC^{\times}$. In this paper, explicit dependence on the spectral parameter is often of little relevance, and so we shall exploit the implicit notation $f^\pm$ without specifying how the shift is realised, we shall also use $f^{[n]}$ which means applying the shift $n$ times. In fact, one can go as far as to consider only a discrete set of points that are related to one another by translations $f^{[n]}$, $n\in\ZZ$. \label{page:singularities}For clarity of exposition though and for a comfortable treatment of potential singularities via analytic continuation we shall assume that the functions are holomorphic functions of the spectral parameter in a large enough simply-connected domain that allows applying the shift operation for as many times as is needed. If the space $\Sigma$ of the spectral parameter values is not simply-connected, this domain is meant to be in the universal cover of $\Sigma$.

It will be convenient to absorb the source term $\phi$ into the definition of Baxter Q-function and T-function
\be
\label{eq:Ga}
Q = \dressing\,q\,,\quad T= \dressing^{++}\dressing^{--}\,t\,,
\ee
where $\dressing^{+}\dressing^{-}=1/\phi$, and so Baxter relation becomes
\be
\label{eq:TQ2}
T\,Q=Q^{++}+Q^{--}\,.
\ee
Rescaling \eqref{eq:Ga} is often called gauge transformation however we warn the reader that the gauge transformations we are going to speak about later are different operations.

The function $Q$ is in many aspects a more fundamental object than $T$. One can readily see it from \eqref{eq:TQ2}: If $Q_1,Q_2$ are two independent solutions of \eqref{eq:TQ2} then $T$ can be written as the determinant
\be
\label{eq:TfromQ}
T=\left|\begin{matrix}Q_1^{++} & Q_2^{++} \\ Q_1^{--} & Q_2^{--}\end{matrix}\right|\,
\ee
provided that we normalised solutions to satisfy
\be
\label{eq:Plucker}
W(Q_{1},Q_{2})=1\quad\text{with}\quad W(f_{1},f_{2}):=\left|\begin{matrix}f_1^{+} & f_2^{+} \\ f_1^{-} & f_2^{-}\end{matrix}\right|\,.
\ee
On the example of Baxter TQ-relation we observe that there is not one but two Q-functions $Q_1,Q_2$ obeying certain relations. In higher-rank cases there will be many more Q-functions. A collection of these functions, together with relations they obey and symmetry transformations they enjoy shall be called a Q-system~\footnote{not to be confused with Q-systems that are a character limit  of T-systems \eg in \cite{Kuniba:2010ir}.}.

\subsection{State of the art for $\sl_{\rank +1}$ Q-systems}
Generalisation of the above discussion from $\sl_2$ to $\sl_{\rank +1}$ can be done in several conceptually different and yet deeply interrelated ways.

\paragraph{Quantum characters} One possibility is to perceive $T$ in \eqref{eq:TQ2} as a ``quantum character'', see \cite{KNIGHT1995187,1998math.....10055F} for precise definitions. Think about the Schur polynomial $\chi_{{\tiny\yng(1)}}=x+y$ for the defining representation of $\sl_2$ and \eqref{eq:TQ2} written as $T=\frac{Q^{++}}{Q}+\frac{Q^{--}}{Q}$ as its quantum generalisation. For higher ranks, the character of finite-dimensional representation of $\sl_{\rank+1}$ labelled by an integer partition $\lambda$ is given by
\be\label{eq:chichar}
\chi_\lambda=\sum_{\Tau}\prod_{(a,s)\in\lambda}x_{\Tau_{a,s}}\,,
\ee
where the sum runs over all semi-standard Young tableaux $\Tau$ of shape $\lambda$ and the product runs over all boxes of $\lambda$ parameterised using the Cartesian coordinates $(a,s)$. The quantum version of this character is the transfer matrix in representation $\lambda$ which can be computed as \cite{Bazhanov:1989yk}~\footnote{An overall shift of spectral parameter in the transfer matrix as well as numeration of $\Lambda_a$ is a matter of convention. We choose it to match the rest of the paper.}
\be\label{eq:Tchar}
\T_{\lambda}^{[\lambda_1-\lambda_1'-1]}=\sum_{\Tau}\prod_{(a,s)\in\lambda}\Lambda_{\rank+2-\Tau_{a,s}}^{[2(s-a)]}\,,
\ee
with $\lambda'$ being the transposed Young diagram and
\be\label{eq:Ladef}
\Lambda_{a}=\left(\frac{Q_{\leftarrow\vphantom{(} a}^{+}Q_{\leftarrow\vphantom{(} (a-1)}^{[-2]}}{Q_{\leftarrow\vphantom{(} a}^{-}Q_{\leftarrow\vphantom{(} (a-1)}^{}}\right)^{[-a+\frac{\rank+1}2]}\,,\quad a=1,\ldots, \rank+1\,,
\ee
where, in our normalisation choice that generalises \eqref{eq:Ga}, $Q_{\leftarrow\vphantom{(} 0}=Q_{\leftarrow\vphantom{(} (\rank+1)}=1$. By continuing the analogy, $\Lambda_a$ shall be called quantum eigenvalues \cite{Sklyanin:1992sm}.

By imposing that $\T_\lambda$ are non-singular functions (probably up to a well-controlled prefactor like $\dressing^{++}\dressing^{--}$), one should require that poles coming from denominators of \eqref{eq:Ladef} are cancelled out which results in the conventional nested Bethe Ansatz equations
\be\label{eq:NBAE}
\frac{Q_{\leftarrow\vphantom{(} (a-1)}^+Q_{\leftarrow\vphantom{(} (a+1)}^+Q_{\leftarrow\vphantom{(} \vphantom{(}a}^{[-2]}}{Q_{\leftarrow\vphantom{(} (a-1)}^-Q_{\leftarrow\vphantom{(} (a+1)}^-Q_{\leftarrow\vphantom{(} \vphantom{(}a}^{[+2]}}=-1\quad\text{at zeros of}\quad Q_{\leftarrow\vphantom{(} a}\,.
\ee
Here we come to the most simplistic way of introducing Baxter Q-functions by analytic Bethe Ansatz: these are functions of type $Q=\dressing\, q$, where $q$ is a polynomial (or, more generally, a non-singular function) at whose zeros (Bethe roots) the nested Bethe Ansatz equations \eqref{eq:NBAE} should be satisfied, and $\dressing$ provides source terms in these equations. At generic point (of a parameter space describing a model), solutions to the nested Bethe equations with non-coinciding Bethe roots are expected to correctly describe the spectrum of the model but at special points which are of relevance for applications one needs to be more careful, a more detailed discussion is present in Section~\ref{sec:AnalyticBetheAnsatz}.

Whereas existence of Baxter Q-functions and of the generating sequence \eqref{eq:Tchar} can be guessed from Bethe equations, their to-date derivation is far beyond pure guesswork, Q-functions are constructed \cite{Bazhanov:1996dr,Bazhanov:1998dq,Bazhanov:2001xm,Derkachov:2003qb, Yang:2005ce, Kojima:2008zza, Derkachov:2010qe, Bazhanov:2010jq, Kazakov:2010iu}  as Q-operators or, almost equivalently \cite{Boos:2017mqq},  as characters \cite{Hernandez:2011ama} of infinite-dimensional representations of $U_\q(\hat{\mathfrak{b}})$ (or of a shifted Yangian \cite{Brundan2004,Gerasimov:2005qz}), see also \cite{Frassek20} and references therein. By the construction, the Q-functions have  analytic properties that justify the analytic Bethe Ansatz requirements.

\paragraph{Quantum spectral curve} Another way for introducing Q-functions is to generalise \eqref{eq:TQ2} as a finite-difference linear equation. For $\sl_{\rank+1}$, this equation becomes of degree $\rank+1$ \cite{Krichever:1996qd}
\be\label{eq:KLWZ}
\sum_{a=0}^{\rank+1} (-1)^a  \CD^{-a}\,T_{(a)}\, \CD^{-a}\,  Q^{[\frac{\rank+1}2]}=0\,,
\ee
where $\CD f:=f^{+}$. For spin chains, $T_{(a)}$ have interpretation of transfer matrices in the $a$'th fundamental representation of $\sl_{\rank+1}$ and the above equation can be equivalently written as \cite{Talalaev:2004qi, Chervov:2007bb}
\be\label{eq:Tal}
\det (1-\CD^{-1}\mathcal{M}\CD^{-1})\,Q^{[\frac{\rank+1}2]}=0\,,
\ee
where $\mathcal{M}$ is the monodromy matrix of the model (which is the universal $R$-matrix in a particular representation) and $\det$ is a column-ordered determinant \cite{Maninbook}. The higher-rank Baxter relation written as \eqref{eq:Tal} clearly suggests to interpret it as a quantisation of the classical spectral curve $\det(\lambda-\mathcal{M})=0$ with $\mathcal{M}$ being the classical Lax matrix, in pretty much the same way as the \Schrodinger equation is a quantisation of $\frac{p^2}{2}+V-E=0$. This idea was one of the key ingredients for the Sklyanin's separation of variables program which he realised for $\sl_2$ \cite{Sklyanin:1984sb,Sklyanin:1991ss} and partially for $\sl_3$  \cite{Sklyanin:1992sm} cases. For rational XXX-type $\sl_{\rank+1}$ spin chains in arbitrary finite-dimensinal representation, an SoV basis which features $Q$ as the wave function was built in \cite{Ryan:2018fyo, Ryan:2020rfk}~\footnote{Based on the recipe of \cite{Maillet:2018bim}, an equivalent SoV basis, with $Q$ playing the role of the wave function, can be constructed for XXZ case as well, \cf \cite{Maillet:2018rto}. However, the proof of \cite{Ryan:2018fyo, Ryan:2020rfk} that this SoV basis also diagonalises the higher-rank version of Sklyanin's B-operator \cite{Smi01,Gromov:2016itr} cannot be that easily generalised and hence interpretation of separated variables as a quantisation of the classical dynamical divisor is not yet justified beyond the rational case.}.

Equation \eqref{eq:Tal} has $\rank+1$ independent solutions $Q_a$ and we normalise them to satisfy
\be\label{eq:Qcond}
W(Q_1,\ldots, Q_{\rank+1})=1\,,
\ee
where $W(f_1,\ldots, f_{k}):=\det\limits_{1\leq a,b\leq k}f_{a}^{[k+1-2b]}\,.$

The solutions can be also used to formulate a quantum analog of Weyl-Jacobi determinant character formula $\chi_{\lambda}=\frac{\det\limits_{1\leq a,b\leq \rank+1}x_a^{\lambda_b+1-b}}{\det\limits_{1\leq a,b\leq \rank+1}x_a^{1-b}}$ \cite{Krichever:1996qd}:
\be\label{eq:WJ}
\T_{\lambda}^{[\lambda_1-\lambda_1'+1-\frac{r+1}{2}]}=\det\limits_{1\leq a,b\leq \rank+1} Q_a^{[2(\lambda_b+1-b)]}\,.
\ee
Clearly, \eqref{eq:Qcond} is specialisation of \eqref{eq:WJ} to the trivial representation. 

Being solutions of a linear equation, the functions $Q_a$ are defined  ambiguously, up to linear $\sl_{\rank+1}$ transformations~\footnote{We want to keep normalisation \eqref{eq:Qcond} intact, that is why these are not $\gl_{\rank+1}$ transformations.}. As we are dealing with finite-difference equations, linear transformations with periodic ($f^{++}=f$) functions are also allowed which promotes this symmetry to the loop algebra $\hat\sl_{\rank+1}$~\footnote{potentially subject to an appropriate completion or restriction dictated by imposed analytic properties of Q-functions.}. However, one should not confuse it with $\hat\sl_{\rank+1}$ used to construct the quantum algebra $U_\q(\hat\sl_{\rank+1})$. To start with, algebra of symmetries of Baxter equation exists independently of which quantum algebra was used to construct the integrable model. Also, the quantum algebra has generators that commute with Hamiltonians, in particular with Baxter operators $Q_k$, whereas $(Q_1,\ldots, Q_{\rank+1})$ transform as a vector representation under action of  $\hat\sl_{\rank+1}$ symmetry of Baxter equation. Even in the XXZ case, we are nominally dealing with two different algebras, one is $\q$-deformed and the other one is not. In the latter one, the parameter $\q$ appears instead as a period of elements of the non-deformed $\hat\sl_{\rank+1}$, \ie $\hat\sl_{\rank+1}\simeq \sl_{\rank+1}\otimes\CC[t,t^{-1}]$ with $t=z^{\frac{2\pi i}{\log \q}}$. Although the two mentioned $\hat\sl_{\rank+1}$ algebras are conceptually different, they are nevertheless related: they are Langlands dual of one another. To get a better feeling about this statement we remark that the underlying zero-level algebras share the same Weyl group (the permutation group $\mathsf{S}_{\rank+1}$). Identification goes beyond a formal isomorphism: On the level of the quantum algebra, the action of the Weyl group changes the representation in which Baxter operator is computed (explicitly this can be seen in the constructions of \cite{Bazhanov:2001xm,Bazhanov:2010jq}) which can be literally mapped to taking a different solution of Baxter equation.

\paragraph{Weyl transform}  
Let us take another look on \eqref{eq:Plucker}. It relates two Q-functions. One of them, say $Q_1$, solves conventional Bethe equations \eqref{eq:NBAE} and so we identify $Q_1\equiv Q_{\leftarrow\vphantom{(} 1}$. On the other hand, by applying exactly the same logic as in derivation of \eqref{eq:NBAE} we see that $Q_2$ also solves Bethe equations of the same form. To summarise, starting from $Q_1$ which satisfies \eqref{eq:NBAE}, we use \eqref{eq:Plucker} to compute $Q_2$ which satisfies an equivalent of \eqref{eq:NBAE}. The preference of one set of the equations over another may exist (if \eg $Q_1$ is a polynomial of lower degree than $Q_2$) but definitely it is not meaningful for as long as we are mostly ignoring explicit analytic structure of Q-functions. 

\label{bosonicdualitypage}
 This generalises to the higher-rank case as follows: Starting from the functions $Q_{\leftarrow\vphantom{(} (a-1)}$, $Q_{\leftarrow\vphantom{(} a}$, $Q_{\leftarrow\vphantom{(} (a+1)}$, one introduces a new Q-function  $\bar Q_{\leftarrow\vphantom{(} a}$ as the one satisfying
\be\label{eq:Plucker2}
W(Q_{\leftarrow\vphantom{(} a},\bar Q_{\leftarrow\vphantom{(} a})=Q_{\leftarrow\vphantom{(} (a-1)}Q_{\leftarrow\vphantom{(} (a+1)}\,.
\ee
Equation \eqref{eq:NBAE} can be seen as a consequence of \eqref{eq:Plucker2} using the following argument \cite{Voros_2000}: shift \eqref{eq:Plucker2} in two different directions, $Q_{\leftarrow\vphantom{(} a}^{[\pm 2]}\bar Q_{\leftarrow\vphantom{(} a}-\bar Q_{\leftarrow\vphantom{(} a}^{[\pm 2]} Q_{\leftarrow\vphantom{(} a}=\pm Q_{\leftarrow\vphantom{(} (a-1)}^{\pm}Q_{\leftarrow\vphantom{(} (a+1)}^{\pm}$\,, evaluate at zeros of $Q_{\leftarrow\vphantom{(} a}$ thus cancelling one term assuming non-singularity of $Q$-functions, and divide the shifted equation in one direction by the shifted equation in the other direction. $\bar Q_{\leftarrow\vphantom{(} a}$ and $Q_{\leftarrow\vphantom{(} a}$ enter symmetrically (up to a sign) the \Plucker\!\!-type relation \eqref{eq:Plucker2} and hence we can derive using the same procedure Bethe equations \eqref{eq:NBAE} of exactly the same form but with $Q_{\leftarrow\vphantom{(} a}$ replaced everywhere  with $\bar Q_{\leftarrow\vphantom{(} a}$ ~\footnote{This includes \eqref{eq:NBAE} at zeros of $Q_{\leftarrow\vphantom{(} (a\pm 1)}$, where $\frac{Q_{\leftarrow\vphantom{(} a}^+}{Q_{\leftarrow\vphantom{(} a}^-}$ is replaced with $\frac{\bar Q_{\leftarrow\vphantom{(} a}^+}{\bar Q_{\leftarrow\vphantom{(} a}^-}$.}.

The transformation from $Q_{\leftarrow\vphantom{(} a}$ to $\bar Q_{\leftarrow\vphantom{(} a}$ appeared numerous times in the literature and has several different names, we are aware that it was called: beyond equator \cite{Pronko:1998xa}, reproduction procedure \cite{mukhin2002populations}, bosonic (as opposed to fermionic) duality of Bethe equations \cite{Gromov:2007ky}, and \Backlund\!\!-type transformation \cite{Frenkel:2020iqq}. We shall refer to it under yet another name ``Weyl transform'' of Q-functions in attempt to settle a name that reflects the group-theoretical meaning of what is happening. Indeed, for the $\sl_2$ case we can readily notice that the transform permutes the two solutions of Baxter equation. The Weyl symmetry interpretation for higher rank cases shall become clear as we proceed ~\footnote{Since $\bar Q_{\leftarrow\vphantom{(} k}+\alpha\, Q_{\leftarrow\vphantom{(} k}$ is a solution of \eqref{eq:Plucker2} for any periodic function $\alpha$, the transform gets true meaning of the Weyl group action only after this symmetry is taken under control. For instance, in the case of spin chains with twisted boundary conditions, large-$z$ asymptotics of Q-functions and a prescription for Borel resummation of $1/\log z$ expansion fix $\alpha$ \cite{Kazakov:2015efa}. In the case of twist-less rational spin chains, $\alpha$'s should be constants and all options to choose them are organised into a complete flag \cite{Mukhin2003,MV05}: In its Bruhat decomposition $[B\sigma]$, $\sigma$ corresponds to a composition of Weyl transforms and $B$ controls the ambiguity in choosing $\alpha$'s, see more on page~\pageref{pg:miura}.}. 

\paragraph{Miura transform}
Introduce now a suggestive notation \cite{Tsuboi:2009ud}
\begin{align}
\label{eq:BosWeyl}
Q_{1,2,\ldots, a-1,a} &:=Q_{\leftarrow\vphantom{(} a} \,,
\nonumber\\
Q_{1,2,\ldots, a-1,a+1} &:=\bar Q_{\leftarrow\vphantom{(} a}
\end{align}
which alludes to the orthonormal basis $\varepsilon_a$  of the $\sl_{\rank+1}$ root lattice (simple roots are $\alpha_a=\varepsilon_a-\varepsilon_{a+1}$) and to the transition from $Q_{\leftarrow\vphantom{(} a}$ to $\bar Q_{\leftarrow\vphantom{(} a}$ corresponding to an action of the Weyl reflection that permutes $\varepsilon_a$ and $\varepsilon_{a+1}$.

By performing in total $\rank+1 \choose 2$ Weyl transforms in a special way one can generate $\rank+1 \choose 2$ new Q-functions that contain, among others, $Q_1,\ldots, Q_{\rank+1}$ solving Baxter equation \eqref{eq:KLWZ}, see \eg (69)-(71) of \cite{Marboe:2017dmb} for an illustration. Thus we say that $Q_1,\ldots, Q_{\rank+1}$ can be derived from $Q_{\leftarrow\vphantom{(} 1}$, $\ldots$, $Q_{\leftarrow\vphantom{(} \rank}$ and $Q_{\leftarrow\vphantom{(}(\rank+1)}=1$. The reverse procedure is neatly organised into the following determinants
\be
\label{eq:QkQ}
Q_{\leftarrow\vphantom{(} a}=W(Q_1,\ldots, Q_a)\,,\quad a=1,\ldots, \rank+1\,.
\ee
These remarkable relations appeared numerously in the literature  under different disguises, and in particular they are a special case of Theorem 3.2 in \cite{Tsuboi:2009ud}. To better understand the meaning of \eqref{eq:QkQ}, it is instructive to rewrite Baxter equation \eqref{eq:KLWZ} in a factorised form
\be\label{eq:Miura20}
(1-\Lambda_{\rank+1}\CD^{-2})\ldots(1-\Lambda_2\CD^{-2})(1-\Lambda_{1}\CD^{-2})Q_a^{[\frac{\rank+1}{2}]}=0\,
\ee
which is also known as Miura transform \cite{1996CMaPh.178..237F}. The conditions specifying the factorisation are the following ones
\be
\label{eq:varconstant}
\nonumber
&0&=(1-\Lambda_{1}\CD^{-2})Q_1^{[\frac{\rank+1}{2}]}\,,
\\
&0&=(1-\Lambda_{2}\CD^{-2})(1-\Lambda_{1}\CD^{-2})Q_2^{[\frac{\rank+1}{2}]}\,,
\\
&&\nonumber\ldots
\\
&0&=(1-\Lambda_{a}\CD^{-2})\ldots(1-\Lambda_{1}\CD^{-2})Q_a^{[\frac{\rank+1}{2}]}\,,
\nonumber
\\
&&\nonumber\ldots\,.
\ee
The factorisation procedure hence reduces the symmetry algebra $\hat\sl_{\rank +1}$ to $\hat{\mathfrak b}$, where $\mathfrak b$ is the Borel subalgebra of $\sl_{\rank +1}$.

The reader is welcome to verify that  $\Lambda_a$ are precisely the ones given by \eqref{eq:Ladef} and then
\be
(1-\Lambda_{a}\CD^{-2})\frac{Q_{\leftarrow\vphantom{(} a}^{[1-a+\frac{\rank+1}2]}}{Q_{\leftarrow\vphantom{(} (a-1)}^{[-a+\frac{\rank+1}{2}]}}=0\,.
\ee
Alternatively, if $Q_1,\ldots, Q_{\rank-a}$ are solutions of Baxter equation which is an equation of degree $\rank+1$, the functions $\tilde Q_b=\frac{W(Q_1,\ldots, Q_{\rank-a},Q_b)}{W(Q_1,\ldots,Q_{\rank-a})^{-}}$, for $b=\rank-a+1,\ldots,\rank+1$ are solutions of the degree $a+1$ equation
\be
(1-\Lambda_{\rank+1}\CD^{-2})\ldots(1-\Lambda_{\rank+1-a}\CD^{-2})\tilde Q_b^{[a-r+\frac{\rank+1}{2}]}=0\,.
\ee
This can be viewed as a \Backlund flow from $\sl_{\rank+1}$ to $\sl_{a+1}$ Q-systems \cite{Krichever:1996qd}, but, in simplest possible terms, it is just the method of variation of constants \cite{Lagrange}.

\paragraph{Extended Q-system on the Weyl orbit} Relation \eqref{eq:QkQ} suggests an immediate generalisation. For any multi-index $A=a_1\ldots a_k$ of no more than $\rank+1$ distinct entries, one can define a Q-function $Q_A$
\be\label{eq:QA}
Q_{A}=W(Q_{a_1},\ldots, Q_{a_k})\,.
\ee
A collection of $2^{\rank+1}$ such Q-functions (with $Q_{\emptyset}=Q_{\fullset}=1$) shall be called the extended Q-system or the Q-system on the Weyl orbit (these two names will become distinct for other Lie algebras) because $Q_{A}$ with $|A|=a$ constitute the orbit of $Q_{\leftarrow\vphantom{(} a}$ under the action of Weyl transforms in the sense of \eqref{eq:BosWeyl}. They generalise \eqref{eq:Plucker2} to  \cite{Pronko:1999gh,Bazhanov:2001xm,Tsuboi:2009ud}
\be
\label{eq:Plucker5}
W(Q_{Aa},Q_{Ab})=Q_{A}Q_{Aab}\,.
\ee
The Q-functions $Q_A$ with $|A|=a$ can be viewed as components of an $a$-form thus transforming under $a$'th fundamental representation of $\hat\sl_{r+1}$. We see that the Weyl group acting on $Q_A$ gets promoted to the full symmetry group $\hat\sl_{r+1}$ of Baxter equation. 

Although the Q-functions $Q_A$ are definitely not functionally independent, the gained covariance has its own benefits. To illustrate some of them, let us also introduce contra-variant Hodge-dual Q-functions \cite{Krichever:1996qd,Kojima:2008zza,Gromov:2014caa}
\be\label{eq:Hodge}
Q^{A}:=\frac 1{|\bar A|!}\epsilon^{A\bar A}Q_{\bar A}\,.
\ee
Both $Q_a$ and $Q^a$ were recently used simultaneously for computation of scalar products \cite{Gromov:2019wmz}. One of the reasons for which this computation was possible is that $Q^a$ satisfy  the ``conjugate'' Baxter equation \cite{Kuniba:2001ub}
\be\label{eq:KLWZ2}
Q^{b\,[-\frac{\rank+1}2]}\sum_{a=0}^{\rank+1} (-1)^a {\overleftarrow\CD}^{-a}\, T_{a,1} \, {\overleftarrow\CD}^{-a}  =0\,,
\ee
where $f \overleftarrow\CD=f^{-}$.

Furthermore, one can form singlets from Q-functions and their Hodge duals which, by inspection, provide us with a compact bilinear formula for transfer matrices $T_{\lambda}$ with $\lambda=(s^a)$ being a Young diagram of rectangular shape (\ie a Kirillov-Reshetikhin module \cite{Kirillov:1990}):
\be\label{eq:TQQ}
T_{a,s}=\frac 1{a!} \sum\limits_{|A|=a}Q_{A}^{[s+{\frac{\rank+1}{2}}]}(Q^A)^{[-s-{\frac{\rank+1}{2}}]}\,.
\ee

Supersymmetric version of the extended Q-system \cite{Tsuboi:2009ud} was, with further elaboration, instrumental in  solution of the AdS$_5$/CFT$_4$ spectral problem: First, generalisation of \eqref{eq:TQQ} allowed to solve \cite{Gromov:2010km} the T-system on T-hook \cite{Gromov:2009tv} and then to exploit this solution to formulate a finite set of nonlinear integral equations \cite{Gromov:2011cx}. Then the integral equations were simplified further into the AdS/CFT quantum spectral curve \cite{Gromov:2013pga,Gromov:2014caa} -- a $\mathfrak{psl}_{4|4}$ extended Q-system supplemented with a certain Riemann-Hilbert problem fixing the analytic properties of the Q-functions (the latter can be viewed as an analytic Bethe Ansatz). Further analysis of this curve (using the extended Q-system and not the nested Bethe equations!) allowed getting explicit solutions, up to numbers thus providing exact results for spectrum of planar $\mathcal{N}$=4 SYM, see \eg \cite{Gromov:2015wca,Marboe:2017dmb,Marboe:2018ugv} and reviews \cite{Gromov:2017blm, Levkovich-Maslyuk:2019awk}.

\paragraph{Fused flags and opers} The extended Q-system has also a natural geometric interpretation. Recall that $Q_A$ with $|A|=a$ are components of an exterior $a$-form in $\CC^{\rank+1}$ which we shall denote as $Q_{(a)}$. Based on \eqref{eq:QA}, this form is not arbitrary but such that it defines an $a$ dimensional hyperplane $\CC^a\subset \CC^{\rank+1}$ which we shall also denote as $Q_{(a)}$. The embedding $Q_{(a)}\subset \CC^{\rank+1}$ naturally depends on the spectral parameter.

Determinant relation \eqref{eq:QA} also informs us that the hyperplanes $Q_{(a)}$ are embedded into one another in a special way:
\be
\label{eq:fus}
\begin{tikzcd}
  \ldots \drar[Subseteq]&        & Q_{(a)}^{[n+2]} \drar[Subseteq] &       & \ldots
\\
& Q_{(a-1)}^{[n+1]}  \drar[Subseteq]\urar[Subseteq] &             & Q_{(a+1)}^{[n+1]} \drar[Subseteq]\urar[Subseteq] &
\\
  \ldots \drar[Subseteq]\urar[Subseteq]&         & Q_{(a)}^{[n]}  \drar[Subseteq]\urar[Subseteq]      &           & \ldots
\\
&Q_{(a-1)}^{[n-1]} \urar[Subseteq]  &            & Q_{(a+1)}^{[n-1]} \urar[Subseteq] &
\end{tikzcd}
\ee
The embeddings are non-degenerate meaning that $Q_{(a-1)}^{[n-1]}$ and $Q_{(a+1)}^{[n+1]}$ span $Q_{(a)}^{[n]}$. This agrees with the normalisation condition \eqref{eq:Qcond}. Enforcing the normalisation condition \eqref{eq:Qcond} potentially introduces poles into the Q-functions and non-degeneracy is allowed to fail at such poles (which, in concrete examples, is a discrete or even a finite set of points).

A chain of embeddings, for instance 
\be
0\subset Q_{(1)}\subset Q_{(2)}^+\subset Q_{(3)}^{[2]}\subset\ldots\subset Q_{(\rank)}^{[\rank-1]}\subset\CC^{\rank+1}\,,
\ee
defines a maximal flag of $\CC^{\rank+1}$. Given this observation, we shall call a collection $Q_{(a)}(z)$ obeying \eqref{eq:fus} a {\it fused flag}. ``Fusion'' alludes here to the fusion procedure \cite{Kulish:1981gi,Kulish86,Cherednik86,Zabrodin:1996vm} used to construct higher representations and which involves shifting spectral parameter by an integer.

The concept of a fused flag for $\sl_{r+1}$ system was described in \cite{Kazakov:2015efa} though this name was not used there (it is a new name that we propose in this paper).  A very similar geometric construction appeared later and independently from \cite{Kazakov:2015efa} in \cite{Koroteev:2018jht}, from study of this work we concluded that the fused flag is gauge-equivalent to a finite-difference oper \cite{STS1,STS2,Koroteev:2018jht}. We postpone a detailed  discussion of this correspondence until Section~\ref{sec:opers}. Finite-difference opers in turn generalise differential opers which go back (at least) to the work of Drinfeld and Sokolov \cite{Drinfeld1985}, see \cite{Frenkel:2003qx,Frenkel:2004qy,2005math......1398B} and references therein. 

\subsection{State of the art for Q-systems based on other simple Lie algebras}
For arbitrary simple Lie algebras, Bethe equations are known \cite{Ogievetsky:1986hu,Ogievetsky:1987vv}. By taking the most simplistic point of view on Q-functions as $Q=\dressing\,q$, where zeros of polynomial $q$ are Bethe roots, Bethe equations can be written as
\be\label{eq:BAEOW}
\prod_{b} \frac{Q_{(b),1}^{[+A_{ab}]}}{Q_{(b),1}^{[-A_{ab}]}}=-1\quad\text{at all the zeros of}\quad q_{(a),1}\,,
\ee
where the product runs over the nodes of the Dynkin diagram and $A_{ab}$ is the symmetrised Cartan matrix. The functions $Q_{(a),1}$ are analogs of $Q_{\leftarrow\vphantom{(} a}$, they shall be called Q-functions on the Dynkin diagram. The choice of notation $Q_{(a),1}$ is done for future convenience.

The functions $Q_{(a),1}$ are by now \cite{Frenkel:2013uda} well-established as characters of prefundamental representations, moreover corresponding generalisations of TQ-relations \eqref{eq:TQ2} and  quantum eigenvalues $\Lambda_a$ \eqref{eq:Ladef} were lifted  to the Grothendieck ring of $U_{\q}(\hat{\mathfrak{b}})$ representation category, where $\hat{\mathfrak{b}}$ is the Borel sublagebra of $\hat\algg$. 

A sensible complication compared to the $\sl_{\rank+1}$ case is that evaluation maps (from $U_{\q}(\hat\algg)$ to $U_{\q}(\algg)$ or from $\Y(\algg)$ to $\algg$)  are not universally available, and hence finite-dimensional irreps of a quantum algebra are generically different, as vector spaces, from the corresponding irreps of the Lie algebra \cite{Chari:1994pz}. Therefore it is unlikely that expressions for transfer matrices would always straightforwardly generalise Lie algebra characters, as \eqref{eq:Tchar} did with \eqref{eq:chichar}. Explicit expressions in terms of quantum eigenvalues are known for Kirillov-Reshetikhin representations of classical series Lie algebras \cite{Reshetikhin1983,Reshetikhin:1986vd,Kuniba_1995}, and for some transfer matrices for other algebras, including twisted affine ones \cite{Reshetikhin:1987}.

To construct an analog of Baxter equation in the form \eqref{eq:Tal} that uses a mondromy matrix, one could expect \cite{Chervov:2006xk} that we should take the same operator $(1-\CD^{-}\mathcal{M}\CD^{-})$ as is featured in (quantum) Knizhnik-Zamolodchikov equations \cite{Knizhnik:1984nr,Frenkel:1991gx}, and then Baxter Q-functions are equal to or reconstructable from solutions of the equation $\det (1-\CD^{-}\mathcal{M}\CD^{-})Q=0$, with certain prescription of how to compute $\det$. It seems that an explicit realisation of this idea was not done, at least we are not aware of one. 

If one asks about a generalisation of \eqref{eq:Miura20}, it was given for classical Lie algebras in \cite{Kuniba:2001ub,Tsuboi:2002py}. This generalisation also has a complication: the factorised operator acting on $Q$ has a denominator meaning that expansion in powers of the shift operator $\CD$ does not terminate producing formally an infinite-degree equation. A similar situation is present in supersymmetric $\glmn$ systems \cite{Tsuboi:1997iq,Kazakov:2007fy}. On should be careful in identifying solutions of such an equation. One can use for instance the variation of constant logic \eqref{eq:varconstant} while fixing $\Lambda_a$ independently from $Q$, using the fact that the factorised operator is a generating functional for transfer matrices, \cf\eqref{eq:Miura20} vs \eqref{eq:KLWZ}. 

We see that Baxter equation is more subtle and also more technically challenging to use for algebras beyond the $\sl_{\rank+1}$ case. However, our main interest is not this equation but rather its solutions and their combintaions, {\ie}a Q-system. It is possible to study this system using different means, namelly we shall use ideas supplied by ODE/IM correspondence. In the form directly relevant for us, these ideas were first given by Sun \cite{Sun:2012xw} and then in the works of Masoero, Raimondo, and Valeri \cite{Masoero:2015lga,Masoero:2015rcz}. There they introduced the concept of the QQ-system based on the following observation which we explain on the example of simply-laced Lie algebras. As will be reviewed in detail in Section~\ref{sec:mainfeatures}, there exist spectral-parameter-dependent vectors $Q_{(a)}$ in the $a$'th fundamental representation of the Lie algebra $\algg$, and $a$ running through all nodes of the Dynkin diagram such that
\be
\label{eq:QQi}
(Q^+_{(a)}\wedge Q^-_{(a)})_{{\rm L}(\omega_{\rm max})} = \left(\bigotimes\limits_{b, C_{ab}=-1} Q_{(b)}\right)_{{\rm L}(\omega_{\rm max})}\,,
\ee
where $(\ldots)_{{\rm L}(\omega_{\rm max})}$ means restriction to the irreducible representation with highest weight $\omega_{\rm max}=\sum\limits_{b, C_{ab}=-1}\omega_b$, and $\omega_b$ are the fundamental weights. 

Choose $Q_{(a),1}$ -- the highest-weight component of $Q_{(a)}$ and $Q_{(a),2}$ -- the component of $Q_{(a)}$ corresponding to the first descendent~\footnote{The descendent is unique for any fundamental representation and it has the weight $\omega_{a}-\alpha_a$, where $\alpha_a$ is the $a$'th simple root.}. Then, by specialising \eqref{eq:QQi} to the highest-weight component, one gets the QQ-relation
\be\label{eq:Plucker3}
W(Q_{(a),1},Q_{(a),2})=\prod\limits_{b,C_{ab}=-1}Q_{(b),1}\,.
\ee
So, clearly $Q_{(a),2}$ plays the role of $\bar Q_{\leftarrow\vphantom{(} k}$ in \eqref{eq:Plucker2}, and the QQ-system can be defined as the collection of $Q_{(a),i}$, for all $a$ and $i=1,2$, satisfying \eqref{eq:Plucker3}.  One easily gets Bethe equations \eqref{eq:BAEOW} from \eqref{eq:Plucker3} by the argument explained after \eqref{eq:Plucker2}.

Instead of restriction to the highest-weight component, one can consider also any other component of \eqref{eq:QQi} on the Weyl orbit of $\omega_{\rm max}$ thus producing equations 
\be
\label{eq:Plucker4}
W(Q_{(a),\sigma(1)},Q_{(a),\sigma(2)})=\pm\prod_{b,C_{ab}=-1}Q_{(b),\sigma(1)}\,,
\ee
where $\sigma$ is an element of the Weyl group (a precise meaning of how it acts on indices shall be clarified later). These equations are analogs of \eqref{eq:Plucker5} for the $\sl_{\rank+1}$ case, and in the context of ODE/IM they were explicitly mentioned in \cite{Masoero:2018rel} ~\footnote{We provide a slightly stronger statement than that of \cite{Masoero:2018rel} for what concerns normalisations. We shall demonstrate that it is always possible to normalise bases of fundamental irreps that equality \eqref{eq:Plucker4} holds up to a sign simultaneously for all $a$ and $\sigma$, and we provide a way to control the sign as well.}.

A result equivalent to \eqref{eq:Plucker4} was obtained well before its appearance in ODE/IM: an equivalent of the QQ-system \eqref{eq:Plucker3} was considered  in the paper by Mukhin and Varchenko \cite{MV05}, and  all possible Weyl transforms in the sense as we defined them on page~\pageref{bosonicdualitypage} were performed there to arrive to a collection of Q-functions. With an appropriate adjustment of notations, this collection satisfies \eqref{eq:Plucker4}.

Restriction to simply-laced Lie algebras in the above discussion can be waved \cite{Masoero:2015rcz}:  in order to get Bethe equations for a Lie algebra $\algg$, one builds a Q-system using a linear problem based on the affine Lie algebra ${}^L\hat\algg$, where $L$ is the Langlands dual. For simply-laced cases, ${}^L\hat\algg=\hat\algg$. This is no longer the case for non-simply-laced cases, however the main complication is not in the absence of invariance but in the fact that the affine algebra ${}^L\hat\algg$ is twisted. Although \cite{Masoero:2015rcz} explains how to tackle this complication, this requires an extra layer of notations, we hence decided to focus on simply-laced cases only in this paper. The non-simply-laced cases are planned for the sequel \cite{ESV-in-preparation}.

A few months ago, a paper by Frenkel, Koroteev, Sage, and Zeitlin appeared \cite{Frenkel:2020iqq} where the notion of the finite-difference oper for arbitrary simple Lie algebras was introduced. The authors of this paper linked the oper construction to the QQ-system of \cite{Masoero:2015lga,Masoero:2015rcz} and hence (in the simply-laced case) to Bethe equations.
\newline
\newline
All of the above-mentioned works about extension of Q-systems beyond Q-functions on the Dynkin diagram were focused on the analytic part of the story. Construction of Q-functions as explicit operators on the whole Weyl orbit was recently realised, though only for fermionic and vector fundamental representations, for $D$-type Yangians \cite{Frassek:2020nki}.

\subsection{The goal, results, and structure of the paper}\label{sec:14}
Although many aspects of Q-system extension for arbitrary Lie algebras were developed in the works \cite{MV05,Sun:2012xw,Masoero:2015lga,Masoero:2015rcz,Frenkel:2020iqq}, the authors of these works were quite focused on reproducing the Q-functions on the Dynkin diagram $Q_{(a),1}$ and their descendents $Q_{(a),2}$, probably with the intention to get to the conventional nested Bethe equations. While  a subset of Weyl transforms is performed in \cite{Frenkel:2020iqq} and the full Weyl orbit Q-system is present in \cite{MV05,Masoero:2018rel}, these observations were not used towards some further concrete advantage.

Furthermore, in contrast to the $A_r$ case, Q-functions on the Weyl orbit are not the only Q-functions that may appear. Indeed, fundamental irreps of ${}^L\hat\algg$ for $\algg\neq A_r$ contain weight subspaces which are not in the Weyl orbit of the highest-weight vector. And, as already clear from \eqref{eq:QQi}, we need to include these subspaces into the discussion to fully benefit from covariance of the Q-system under action of ${}^L\hat\algg$.

The main goal of our paper is to launch a more systematic study of the full extended Q-system which we define as a collection of all components of the vectors $Q_{(a)}$ that satisfy \eqref{eq:QQi} and certain other relations to be introduced later~\footnote{The set of relations we use form an over-determined set. We prove its consistency. For all cases except $E_8$, we can derive all the relations from \eqref{eq:QQi} assuming that parameters of a model are in general position and conjecture that it can be done for $E_8$ as well. Note that \eqref{eq:QQi} is not a minimal set of relations either.}. The extended Q-system enjoys covariance with respect to ${}^L\hat\algg$ action and we expect that, similarly to the advances of the $\sl_{\rank+1}$ case and its supersymmetrisation, such a covariant description will lead to numerous insights in studies of integrable systems and beyond. In this paper we assemble first few results in this direction.

    First,  the extended Q-system enjoys a variety of relations which we call projection properties. They can be interpreted as \Plucker relations defining a fused flag -- a new structure generalising \eqref{eq:fus} that we shall introduce. The fused flag can be identified with a finite-difference oper in a particular gauge, this identification depends on a choice of a Coxeter element while the fused flag itself does not.
    
    Second, while our intuition is strongly based on the corresponding linear problem under the ODE/IM correspondence, the obtained relations are pertinent to the Lie algebra alone. We show that all the obtained relations between Q-functions can be universally satisfied admitting $\rank$ functions as a functional freedom. 
    
    Third, an explicit parameterisation, similar to \eqref{eq:TQQ}, of T-functions for Kirillov-Reshetikhin modules in terms of Q-functions of the extended Q-system is given. This in turn yields solution of the corresponding Hirota equations and hence of the Y-systems appearing in the context of Thermodynamic Bethe Ansatz studies. For $D_\rank$ series, we also provide an explicit character solution of the Q-system which, by substitution to the ansatz for T-functions produces characters of the corresponding $D_\rank$-representations.

The paper is organised as follows: In Section~\ref{sec:MotivationFromODEIM} we review, with some updates, findings of \cite{Sun:2012xw,Masoero:2015lga} and use them to study  $\sl_4\simeq \so_6$ extended Q-system as the simplest concrete example. In Section~\ref{sec:Extended Q-system} we give a general definition of the extended Q-system, show its universality, introduce the notion of a fused flag, show that the extended Q-system is a fused flag, and, finally, link the fused flag to the notion of opers. In Section~\ref{sec:Applications} we solve Hirota equations and comment on character solution of the Q-system and analytic Bethe Ansatz. In Sections~\ref{sec:DAlgebras} and \ref{sec:ExceptionalAlgebras}, we give explicit realisations of the mentioned general ideas in the cases of $D_n$ and exceptional series, respectively. 

\section{ODE/IM and Q-functions}\label{sec:MotivationFromODEIM}
Throughout the paper $\algg$ denotes a simply-laced simple Lie algebra over $\CC$ of rank $\rank$; $\mathfrak{h},\mathfrak{b},\mathfrak{n}$ are, respectively, its Cartan, maximal solvable, and maximal nilpotent subalgebras such that $[\mathfrak{h},\mathfrak{b}]=\mathfrak{n}$. The corresponding simply-connected Lie groups of $\algg,\mathfrak{b},\mathfrak{n}$ are $\groupG,\groupB,\groupN$. The Lie group associated to $\mathfrak{h}$ is the maximal torus  $\mathsf{T}$. $\alpha\in \Phi$ are roots of the algebra, $\alpha_a$, $a=1,\ldots,\rank$ are simple roots, the set of simple roots shall be denoted $\Delta$. $\WG$ is the Weyl group of the root system. The degree of a Coxeter element is the Coxeter number $h$. We shall use a Chevalley basis, with $E_\alpha$ associated to roots, $h_{a}=\alpha_{a}^\vee$ ~\footnote{Capitalised letters for Cartan generators $H_a$ are reserved for an orthogonal basis to be introduced later.}, and $[h_a,E_{\pm \alpha_b}]=C_{ab} E_{\pm \alpha_b}$. The fundamental weights $\omega_a$ are introduced by $\omega_a(h_b)=\delta_{ab}$. The nilpotent subalgebra $\mathfrak{n}$ shall be considered as spanned by the raising operators $E_{\alpha}$, $\alpha>0$. 

As we are dealing with the simply-laced case, we shall not distinguish between the Cartan matrix $C_{ab}$ and the symmetrised Cartan matrix $A_{ab}=(\alpha_a,\alpha_b)$. Also, as the Langlands dual is isomorphic to the algebra itself, we will write $\hat\algg$ instead of ${}^L\hat\algg$, although  the Q-system is actually a representation of ${}^L\hat\algg$. Likewise, we shall not distinguish between Coxeter and dual Coxeter numbers.

\subsection{Main features of the linear problem}
\label{sec:mainfeatures}
Our main intuition is coming from the results of \cite{Sun:2012xw,Masoero:2015rcz} obtained in the context of the ODE/IM correspondence. In this subsection we mostly summarise certain of their findings. The new bits are the introduction of $S^*$-solutions and a more detailed prescription of normalisation conventions.

\paragraph{Linear problem}
Consider the following linear problem:
\be
\label{eq:dA}
{\cal L}_\algg(x,\spm,\evp)\Psi=\left(\frac{d}{dx}+A_\algg\right)\Psi=0,
\ee
where $A_\algg$ is the $\algg$-valued matrix defined by
\begin{equation}
   A_{\algg}=\sum_{a=1}^{\rank}E_{\alpha_a}+
   (x^{\Cox M}-\spm)\EAff\,,\quad \EAff=\evp\, E_{-\theta}\,,
\end{equation}
with $\theta$ being the longest root and $M>0$. In a Chevalley basis, $E_{-\theta}$ is defined up to a sign. Our convention to fix this sign will be specified later. Equation \eqref{eq:dA} is understood as a parallel transport equation  with $\Psi$ being a vector transforming in a representation of $\algg$

We note that $E_{\alpha_a}$ for $a=0,1,\ldots,\rank$ are generators of the untwisted affine Kac-Moody algebra $\hat\algg$. As we focus on representations which are finite-dimensional, the central charge of $\hat\algg$ is zero and hence $\hat\algg$ is isomorphic to the loop algebra $\algg\otimes\CC[t,t^{-1}]$ and $\EAff=t\, E_{-\theta}$, moreover, all representations are of evaluation type where  $t$ assumes a numerical value denoted here as $\evp$.

The linear problem describes the equations of KdV-type and is an example of a $\null^L\hat{\algg}$-oper \cite{Drinfeld1985,feigin2011,Frenkel:2016gxg}, it also describes a conformal limit of modified affine Toda equations \cite{Lukyanov:2010rn,Ito:2013aea}. In the relevant for us context of ODE/IM correspondence, it was used in \cite{Sun:2012xw}. For $\algg=\sl_2$ and in the fundamental representation, \eqref{eq:dA} is a first-order matrix ODE obtained from the \Schrodinger equation for a particle in the homogeneous potential $x^{hM}$ with $\spm$ playing the role of energy. It is from the study of this \Schrodinger equation \cite{voros-quartic} the ODE/IM correspondence eventually emerged.

\paragraph{Symanzik rotation}
Let $\rho^\vee\in\mathfrak{h}$ be the co-Weyl vector. Its defining feature is $[\rho^{\vee},E_{\alpha_a}]=E_{\alpha_a}$ for $a=1,\ldots,\rank$, and then it follows that  $[\rho^{\vee},E_{\alpha_0}]=(1-h)E_{\alpha_0}$. Using these properties, one can verify that the linear problem enjoys a ``Renorm-Group'' equation \cite{sibuya,Suzuki:1999hu,Sun:2012xw}~\footnote{In \cite{Dorey:2007zx}, this is called Symanzik rescaling.}
\be
\label{eq:Symanzik}
\syman^{-\frac{k}{h\,M}\,\rho^\vee}{\cal L}_\algg(\syman^{\frac{k}{h\,M}}x,\syman^{k}\spm,\evp)\syman^{\frac{k}{h\,M}\,\rho^\vee}=\syman^{-\frac{k}{h M}}{\cal L}_\algg(x,\spm,e^{2\pi\ii\,k}\evp)\,,
\ee
where~\footnote{From this definition of $\q$ it may appear that $\q$ is a root of unity. However, $M$ is not restricted to be an integer as we need to unambiguously define $\Psi$ only in a sector of the  complex plane of $x$ that contains the relevant Stokes sectors. To avoid (inessential) issues with definition of $\Psi$, the reader can also think that $M$ is a large enough integer so that $\q^{n}\neq 1$ for all $n\in\mathbb{Z}$ that are encountered in practice.} $\syman=e^{2\pi\,\ii\,\frac{M}{(M+1)}}$ and $k\in\CC$. Then ``the RG flow'' of solution is given by
\be
\label{eq:SR1}
\Psi^{[2k]}(x,\spm)=\syman^{-\frac{k}{h\,M}\,\rho^\vee}\Psi(\syman^{\frac{k}{h\,M}} x,\syman^{k}\spm)\,,
\ee
where by $\Psi^{[2k]}$ we denote a solution of \eqref{eq:dA} with the rescaled coupling constant: $\evp\to e^{2\pi\ii\,k}\evp$.

For $k\in\mathbb{Z}$, transformation \eqref{eq:Symanzik} is a symmetry of the equation and the Symanzik rotation \eqref{eq:SR1} is a way to generate its new solutions. For consistency with other parts of the paper, we have chosen a convention that $\Psi^{[2]}$ corresponds to the minimal non-trivial Symanzik rotation which is a symmetry of the equation. In the following, ``Symanzik rotation'' will typically refer to $\Psi^{[\pm 2]}$ and ``half of the Symanzik rotation'' -- to $\Psi^{\pm}\equiv \Psi^{[\pm 1]}$.

In the following and without loss of generality we set $\lambda=1$.
\paragraph{WKB analysis}
To analyse large-$x$ behaviour of the solutions of \eqref{eq:dA}, one should be a bit careful as the term $x^{h\,M}$ which is naively dominant at large $x$ is multiplied by a nilpotent operator. To rectify this issue, one performs a gauge transformaiton
\be 
{\cal L}\to \tilde{{\cal L}}=p^{\rho^{\vee}}{\cal L}p^{-\rho^{\vee}}
\ee
with $p=(x^{h\,M}-\spm)^{\frac{1}{\Cox}}$. Then, using the action variable $S=\int^x p(x')dx'$, the gauge-transformed linear operator reads
\be
\tilde{{\cal L}}={p}\left[\frac{d}{dS}+\Lambda+\cdots\right]\,,
\ee
where dots stand for the terms suppressed at large $x$ ~\footnote{We always assume that $M$ is large enough, $M>\frac{1}{h-1}$ would suffice for suppression of the dotted terms.}, and
\be
\label{eq:Ladef2}
\Lambda=\sum_{a=1}^{r}E_{\alpha_{a}}+ E_{\alpha_{0}}\,.
\ee
The remaining WKB analysis is straightforward. Let $\mathsf{U}_\mu$ be an eigenvector of $\Lambda$ with an eigenvalue $\mu$. Then there exists a solution of \eqref{eq:dA} whose large-$x$ behaviour is
\be
\label{eq:WKB0}
\Psi=e^{-{\mu}\int^x p(x')dx'}p^{-{\rho^{\vee}}}{\mathsf{U}_{\mu}}+\cdots=e^{-{\mu}\frac{x^{M+1}}{M+1}}x^{-{M}\rho^{\vee}}{\mathsf{U}_{\mu}}+\cdots\,.
\ee

\paragraph{Stokes phenomena}  We shall say that \eqref{eq:WKB0} is considered in the direction $k$, $k\in\mathbb{R}$, if $x=\syman^{\frac{k}{h\,M}}|x|$ with $|x|\gg 1$. Hence $k$ has meaning of a phase in units of the Symanzik angle. 
If $k=k_0$ is such that $\mu\, e^{\frac{2\pi\,\ii}{h} k_0}$ is real and positive then it is a direction of the fastest descent of \eqref{eq:WKB0}. 

There always exists a solution with asymptotics \eqref{eq:WKB0} valid in a direction of the fastest descent. Moreover, if $\mu\, e^{\frac{2\pi\,\ii}{h} k_0}$ is larger than $\Re(\mu'\, e^{\frac{2\pi\,\ii}{h} k_0})$ for $\mu'$ -- any other eigenvalue  of $\Lambda$ then this solution is defined uniquely up to a normalisation and it shall be called the Stokes solution or $S$-solution with the eigenvalue $\mu$ \cite{wasow}. If $\Psi$ is such a solution then $\Psi^{[2h]}$ is another one, with $k_0\to k_0-h$, and so to avoid ambiguity we take $k_0$ to be the one with the smallest absolute value~\footnote{$\Psi$ and $\Psi^{[2n]}$ do not generically coincide as solutions, for any $n\in\mathbb{Z}$. The requirement that they coincide up to a rescaling for certain $n$ is a quantisation condition on ``energy levels'' $\spm$ in the sense of a quantum mechanical problem. We do not impose it here.}. The Stokes solution is the smallest (the fastest decreasing) solution among all solutions of \eqref{eq:dA} for certain range of directions (Stokes sector) $k\in k_0+[-\epsilon,\epsilon]$, where $\epsilon>0$, often $\epsilon=\frac 12$;  also the leading large-$x$ asymptotic behaviour of the $S$-solution is given by \eqref{eq:WKB0} in the applicability cone $k\in k_0+[-\frac \Cox 2-\epsilon,+\frac \Cox 2+\epsilon]$.

Even if $\Re(\mu'\, e^{\frac{2\pi\,\ii}{\Cox} k_0})\geq \mu\, e^{\frac{2\pi\,\ii}{\Cox} k_0}$ for some $\mu'$, it is possible to define a unique solution with asymptotics \eqref{eq:WKB0} if we demand that the cone of applicability of \eqref{eq:WKB0} is large enough: for each $\mu'$, it should contain exactly one connected domain of directions $k$ where $\Re(\mu'\, e^{\frac{2\pi\,\ii}{\Cox} k})< \Re(\mu\,e^{\frac{2\pi\,\ii}{\Cox} k})$ is realised. We shall call such a solution $S^*$-solution ~\footnote{For most of the discussion focusing only on $S$-solutions suffices. The explicit interesting examples featuring $S^*$-solutions appear in footnote~\ref{laref} and relation \eqref{eq:D5SS}.}. Its identification depends on a choice of the applicability cone. 

If $\Psi$ is an $S$- or $S^*$-solution with an eigenvalue $\mu$ then $\Psi^{[2]}$ is also a solution of the same class, with the rotated counter-clock-wise eigenvalue $e^{\frac{2\pi\ii}{h}}\mu$ and the applicability cone rotated by one unit of the Symanzik angle clock-wise.
\paragraph{$\Psi$-functions} Let $a$ be a node of the Dynkin diagram. As will be derived later, depending on $a$, for the $a$'th
fundamental representation, one of the two options is realised. The first option is that $k_0 = 0$ is the direction of the fastest descent for some $S$-solution. Then denote this solution by $\Psi_{(a)}(x, z)$. Its associated eigenvalue is real positive and it will be denoted as $\mu_{a}$. The
second option is that $k_0 = \pm \frac{1}{2}$ are the fastest descent directions for $S$-solutions. Denote then by $\Psi_{(a)}(x, z)$ such a function that $\Psi_{(a)}^\pm(x, z)$ are the $S$-solutions with the fastest descent along $k_0 = \mp\frac{1}{2}$. Their associated eigenvalues are of the form $\gamma^{\pm \frac{1}{2}}\mu_{a}$, where $\gamma=e^{\frac{2\pi {\rm i}}{h}}$ and $\mu_{a}$ is a real positive number.

The above two options are realised in alternation: if a node of the Dynkin diagram realises one option, all adjacent nodes realise the other. By choosing a sign of $E_{-\theta}$  we can enforce one chosen node $a$ to feature $\Psi_{(a)}(x, z)$ as the $S$-solution with $k_0=0$ being the direction of the fastest descent. It will
be our convention for the sign of $E_{-\theta}$ that it is the node of the vector representation for
classical series $A_n$, $n \geq 1$ and $D_n$, $n > 3$, and it is the node at the end of the longest leg for
$E_6, E_7, E_8$.

\paragraph{Q-functions} Baxter Q-vectors are defined as
\be
\label{eq:Baxvec}
Q_{(a)}(\spm)=z^{-\frac{\rho^{\vee}}{\Cox\,M}}\Psi_{(a)}(0,\spm)\,.
\ee
The definition is designed to have the property ~\footnote{More generally, for any $x_0,z_0$, the definition $Q_{(a)}(z)=\left(\frac{z}{z_0}\right)^{-\frac{\rho^{\vee}}{\Cox\,M}}\Psi_{(a)}(\left(\frac{z}{z_0}\right)^{\frac{1}{\Cox\,M}}x_0,z)$ can be used.
}
\be
\label{eq:Baxvec}
Q_{(a)}^{[n]}(\spm):=Q_{(a)}(\syman^{n/2}\spm)=z^{-\frac{\rho^{\vee}}{\Cox\,M}}\Psi_{(a)}^{[n]}(0,\spm)\,.
\ee
Baxter Q-functions are defined as the components $Q_{(a),i}$ of the expansion $Q_{(a)}=\sum\limits_i Q_{(a),i}{\bf e}_{(a),i}$ \wrt some basis. In general discussion, $i$ runs through the set $\{1,2,\ldots,\dim L(\omega_a)\}$. For explicit cases however, it can be convenient to use $i$ as an index from a more descriptive set in which case we do not write $(a)$ in the subscript, \cf \eqref{eq:QA}.

The basis elements ${\bf e}_{(a),i}$ should diagonalise Cartan generators. Let ${\bf e}_{(a),i}$ be of weight $\gamma_i$. One agrees that $\gamma_1=\omega_a$ is the highest weight of the irrep, and $\gamma_2=\omega_a-\alpha_a$ is the only leading descendent from the highest weight. For $i$ such that $\gamma_i$ is on the Weyl orbit of the highest weight, there is a natural notation to use: ${\bf e}_{(a),\sigma(i)}={\bf e}_{(a),j}$, where $\sigma$ is an element of the Weyl group such that $\sigma\gamma_i=\gamma_j$. 

We shall impose the following requirement that partially restricts normalisation of basis vectors: For each $\gamma_i$ on the Weyl orbit of the highest-weight vector, choose one concrete $\sigma_i\in\WG$ such that $\gamma_i=\sigma_i\gamma_1$. Then we require that 
\be
\label{eq:NProp}
{\bf e}_{(a),i}=s_{\sigma_i}{\bf e}_{(a),1}\,,
\ee
where $s_{\sigma_i}$ is the representative of the Weyl group element $\sigma_i$ defined as follows: for $\sigma_{a}$ being a reflection \wrt the simple root $\alpha_a$, this representative is (see \eg \cite{Fulton-Harris}, appendix D.4)
\be
\label{eq:sst}
s_a=e^{E_{\alpha_a}}e^{-E_{-\alpha_a}}e^{E_{\alpha_a}}\,.
\ee 
For element $\sigma$ of length $\ell$ and its minimal length representation $\sigma=\sigma_{a_1}\ldots\sigma_{a_\ell}$, the representative is $s_{\sigma}=s_{a_1}\ldots s_{a_\ell}$. This choice of the representatives enjoys the property (Proposition 3.1.2 of \cite{Rostami15})
\be
\label{eq:ss}
s_{\sigma}s_{\sigma'}=s_{\sigma\sigma'} \prod_{\beta}(-1)^{h_{\beta}}\,,
\ee
where the product runs over such  positive roots $\beta$ that $\sigma'\beta$  is a negative root and $\sigma\sigma'\beta$ is a positive root. Partial cases of \eqref{eq:ss} are: $s_{a}^2=(-1)^{h_{\alpha_a}}$ for every simple root $\alpha_{a}$, and $s_{\sigma}s_{\sigma'}=s_{\sigma\sigma'}$ if $\ell(\sigma)+\ell(\sigma')=\ell(\sigma\sigma')$.

 It then follows from \eqref{eq:ss} that ${\bf e}_{(a),\sigma(i)}=\pm s_{\sigma}{\bf e}_{(a),i}$ for any $\sigma\in\WG$ as long as the normalisation \eqref{eq:NProp} is chosen. This also gives us a concrete recipe to fix signs in \eqref{eq:Plucker4}.

\paragraph{$\Psi$- and QQ-systems}
One of the main results of \cite{Sun:2012xw,Masoero:2015lga} is the equality
\be\label{eq:Psi1}
\left(\Psi_{(a)}^+\wedge \Psi_{(a)}^-\right)_{L(\omega_{\rm max})} =\left(\bigotimes_{b,C_{ab}=-1}\Psi_{(b)}\right)_{L(\omega_{\rm max})}\,.
\ee
This is called $\Psi$-system~\footnote{See also \cite{Dorey:2006an} where $\Psi$-system was obtained on the level of pseudo-differential equations.}. The proof is the following: First, if needed, perform a half-Symanzik rotation of \eqref{eq:Psi1} to make both sides of the equation solving \eqref{eq:dA} in the irrep $L(\omega_{\rm max})$. Then, by analysing the large-$x$ asymptotics along the line of the fastest descent one  deduces that both \lhs and \rhs of \eqref{eq:Psi1} have the same growth rate which moreover coincides with the growth rate of the $S$-solution of \eqref{eq:dA} in the irrep $L(\omega_{\rm max})$ along this line. Hence both sides of \eqref{eq:Psi1} should be, up to normalisation,  this $S$-solution, and it is easy to check that the coefficient of proportionality is non-zero. The normalisations of $\Psi_{(a)}$ can be fixed to get an equality in \eqref{eq:Psi1} for all $a$. 

By evaluation \eqref{eq:Psi1} at $x=0$, one gets the relation \eqref{eq:QQi} between the Q-vectors and eventually the QQ-system defined by \eqref{eq:Plucker3}.

As discussed in the introduction, our goal is not only to focus on the QQ-system relations \eqref{eq:Plucker3} but to explore a variety of properties of the Q-vectors, especially focusing on their covariance with respect to action of $\algg$. We stress that for algebras different from $A_n$, the Q-vectors have components outside of the Weyl orbit of the highest weight, and hence our study goes beyond the Weyl orbit Q-system \eqref{eq:Plucker4}.

\subsection{Example of $\sl_4\simeq \so_6$ extended Q-systems}\label{sec:sl4so6}
We consider first an explicit example of the $A_3$ extended Q-system to illustrate types of relations that we would like to explore in this paper. Because it is also the $D_3$ system, we shall use both $\sl_4$ and $\so_6$ notations in parallel, with the goal of future generalisation to $D_n$ series. The $\so_6$ Q-system with spinor notations was featured for the first time in the context of the AdS$_4$/CFT$_3$ correspondence \cite{Bombardelli:2017vhk}.

There are three fundamental representations: ${\bf 4}$, ${\bf 6}$, ${\bf \bar 4}$, and we use the following notations for $Q$-vectors $Q_{(1),\alpha}\equiv \qS_\alpha$, $Q_{(2),i}\equiv \qV_{i}$, $Q_{(3),\dot{\alpha}}\equiv \qC_{\dot \alpha}$:
\be\label{eq:so6Convention}
\begin{tabular}{c|c|c|c}
& $\sl_4$ & $\so_6$ \\
\hline
    {\bf 4} &
    $Q_a$   & 
    $\qS_{\alpha}$  & $(\qS_1,\qS_2,\qS_3,\qS_4)=(Q_1,Q_2,Q_3,Q_4)$ 
    \\
    {\bf 6} & 
    $Q_{ab}$ &  
    $\qV_{i}=  \gamma_i^{\alpha\beta}\psi^-_{\alpha}\psi^+_{\beta}$     & \parbox[c][4em]{0.5\linewidth}{$V_1=Q_{12}, V_2=-Q_{13}, V_3=Q_{23}$
    \newline
    $V_{-3}=Q_{14}, V_{-2}=Q_{24}, V_{-1}=Q_{34}$
    }
    \\
    ${\bf \bar 4}$ &
    $Q^a=\frac 16\epsilon^{abcd}Q_{bcd}$ &
    $\qC_{\dot\alpha}$ &
    $(\qC_1,\qC_2,\qC_3,\qC_4)=-(Q_{123},Q_{124},Q_{134},Q_{234})$
\end{tabular}
\ee
Here we have introduced standard spinor notation using dotted and un-dotted indices for the representations $\bf 4$ resp $\bf \bar 4$ so that $\alpha,\dot\alpha=1,2,3,4$, the range of the $\sl_4$ fundamental indices is the same, $a=1,2,3,4$. For the vector indices we have $i = 1,2,3,-3,-2,-1$.

We can think about Dirac spinors as elements of $\otimes_{i=1}^{3} \mathbb{C}^{2}$. Introduce the basis $\ket{+} = \begin{pmatrix}1 \\ 0 \end{pmatrix},\ket{-} = \begin{pmatrix}0 \\ 1 \end{pmatrix}$ for $\mathbb{C}^{2}$, then we have a natural basis $\ket{\pm \pm \pm}$ for $\otimes_{i=1}^{3}\mathbb{C}^2$. We need the $2\times 2$ matrices
\be\label{eq:Basic22matrices}
    \sigma^z = \begin{pmatrix}1 & 0 \\ 
                             0 & -1
               \end{pmatrix}\,,
    \quad
    \sigma^- = \begin{pmatrix}0 & 0 \\ 
                              1 & 0
               \end{pmatrix}\,,
    \quad
    \sigma^+ = \begin{pmatrix}0 & 1 \\ 
                              0 & 0
               \end{pmatrix}\,.
\ee
From these matrices we build the $8\times 8$ $\Gamma$-matrices and the charge conjugation matrix $C$ 
\begin{subequations}
\be
    &\Gamma_{\pm i} = -\underbrace{\sigma^{z}\otimes \dots \otimes \sigma^{z}}_{i-1}\otimes \sigma^{\mp} \otimes  \underbrace{\mathbf{1}\otimes \dots \otimes \mathbf{1}}_{r-i}\,,  \quad i>0\,, \\
    &C = \prod_{i=1}^{r}(\Gamma_{i}+\Gamma_{-i}) = \sigma^{x}\otimes (\ii \sigma^{y}) \otimes \sigma^{x}\,,
\ee
\end{subequations}
which satisfy $\{\Gamma_{i},\Gamma_{j}\} = \delta_{i+j,0}\,$. We take the metric for the vectors to be $g_{ij} = \delta_{i+j,0}$. The basis for $\otimes_{i=1}^{3}\mathbb{C}^2$ and the spinor indices are related as follows
\begin{subequations}\label{eq:example}
\be
    \psi = \psi_{1}  \ket{++-}\,
    +
    \psi_{2}  \ket{+-+}\,
    +
    \psi_{3}  \ket{-++}\,
    +
    \psi_{4}  \ket{---}\,, \\
    \eta = \eta_{1}  \ket{+++} \,
    +
    \eta_{2}  \ket{+--}\,
    +
    \eta_{3}  \ket{-+-}\,
    +
    \eta_{4} \ket{--+}\,.
\ee
\end{subequations}
To emphasize that $\psi_{\alpha}$ and $\eta_{\dot \alpha}$ are Weyl spinors we write
\begin{align*}
    &\psi_{\alpha}(C\Gamma_{I})^{\alpha\beta}\psi_{\beta} = \psi_{\alpha}\gamma_{I}^{\alpha\beta}\psi_{\beta}\,,
    & 
    &\psi_{\alpha}(C\Gamma_{I})^{\alpha\dot\beta}\eta_{\dot\beta} = \psi_{ \alpha}\gamma_{I}^{ \alpha\dot\beta}\eta_{\dot\beta}\,, 
    &
    &\eta_{\dot \alpha}(C\Gamma_{I})^{\dot \alpha\dot \beta}\eta_{\dot \beta} = \eta_{\dot \alpha}\Bar{\gamma}_{I}^{\dot\alpha\dot\beta}\eta_{\dot \beta}\,, 
\end{align*}
with $I$ being a multi-index denoting antisymmetrization of the $\Gamma$-matrices, \eg $\Gamma_{ij} = \frac{1}{2}(\Gamma_{i} \Gamma_{j}-\Gamma_{j} \Gamma_{i})$.

Note that the Pfaffian of antisymmetric tensors is related to the inner product of the vectors
\be
    \frac{1}{4}\epsilon^{abcd}Q_{ab}Q_{cd} =  2(Q_{12}Q_{34}-Q_{13}Q_{24}+Q_{14}Q_{23})= V_i \metric^{ij}V_j\,.
\ee

Having introduced the notation, our next goal is to take tensor products of the fundamental irreps and use Q-vectors to construct functions of the spectral parameter in other representations of $A_3$; The irreps that we shall encounter while performing this fusion procedure are given in Table~\ref{tab:A3irreps}.
We shall employ the linear problem \eqref{eq:dA},  we focus on the example $V^{+}\otimes V^{[1-2n]}$ to illustrate its usage. This tensor product decomposes as ${\bf 6}\otimes {\bf 6}={\bf 20'}\oplus{\bf 15}\oplus {\bf 1}$ and we shall specify what happens when $V^{+}\otimes V^{[1-2n]}$ are projected onto the irreps. Consider then $\Psi_{(2)}^{+}\otimes \Psi_{(2)}^{[1-2n]}$ and study their Stokes behaviour:
\begin{table}[t]
    \centering
    \begin{tabular}{c|c|l}
        Dynkin labels & Dimension & Name ($\so_6$ point of view) \\
        \hline
         $[010]$ & ${\bf 6}$ & vector
        \\
        $[100]$, $[001]$  & ${\bf 4}$, ${\bf{\bar 4}}$ & [co]-spinor 
        \\
        $[101]$ & ${\bf 15}$ & adjoint
        \\
        $[110]$, $[011]$ & ${\bf 20}$, ${\bf\overline{20}}$ & 
        \\
        $[111]$ & ${\bf 64}$ & Weyl vector multiplet
        \\
        $[020]$ & ${\bf 20'}$  & symmetric traceless tensor
        \\
        $[200]$, $[002]$ & ${\bf 10}$, ${\bf\overline{10}}$ & [anti]-self-dual 3-form
    \end{tabular}
    \caption{List of certain $A_3$ irreps \protect \footnotemark}
    \label{tab:A3irreps}
\end{table}
\footnotetext{With exception of $[020],[200],[002]$, these are precisely all the minimal size irreps whose highest-weight subspace is invariant under a parabolic subgroup action.}
\be
\label{eq:tab315}
\begin{tabular}{c|c|c|c|c|c}
     $n$ & $\mu$ & WKB applicability & ${\bf 20'}$ & ${\bf 15}$ & ${\bf 1}$ \\
     \hline
     0 & $2(1+\ii)$ & $-\frac 12$+$[-2,2]_+$ & $S$ & 0 & 0
     \\
     1 & $2$ & $0+[-\frac 32,\frac 32]_+$ & $S^*$ & $S$ & 0 
     \\
     2 & $0$ & $\frac 12+[-1,1]_+$ & $S^*$ & $S^*$ & $1$ 
     \\
     3 & $2\ii$ & $1+[-\frac 12,\frac 12]_+$ & NS & NS & $T_{2,1}^{[-2]}$ 
     \\
     4 & $2(1+\ii)$ & $\frac 32+[-0,0]_+$ & NS & NS & $T_{2,2}^{[-3]}$ 
     \\
     $\geq 5$ &  & none & NS & NS &   $T_{2,n-2}^{[1-n]}$
\end{tabular}
\ee
For $n\leq 4$, the asymptotic behaviour at infinity is given by \eqref{eq:WKB0}, where 
\be
\mu=\gamma^{\frac 12}(\mu_{2}+\gamma^{-n}\mu_{2})\,,
\ee
with $\gamma=e^{\frac{2\pi\ii}{h}}=\ii$, and $\gamma^{\frac 12}\mu_{2}=1+\ii$ being the eigenvalue of $\Lambda_{\bf 6}$ for the $S$-solution $\Psi_{(2)}^+$, \cf Fig~\ref{fig:EgLambda1}.
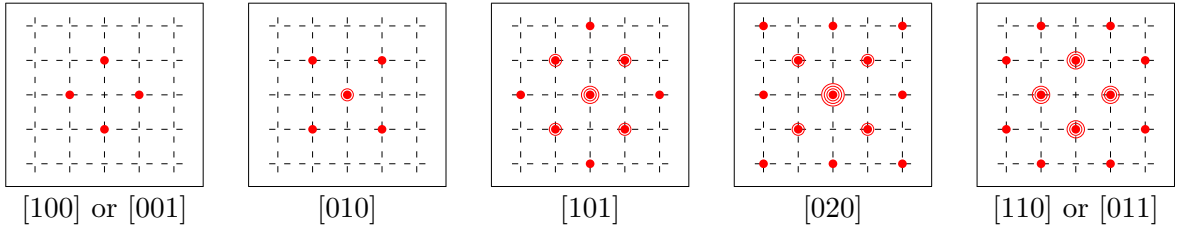
\begin{figure}[t]
\begin{center}
\newcommand\scalefactor{0.65}
\hskip -3em
\mbox{
\begin{tabular}{ccccc}
\scalebox{\scalefactor}{
\fbox{
\begin{picture}(100,100)(-50,-50)
\color{black}
\thinlines
\multiput(0,40)(0,-20){5}{\dashline{4}(-45,0)(45,0)}
\multiput(40,0)(-20,0){5}{\dashline{4}(0,-45)(0,45)}
\color{red}
\put(20,0){\circle*{5}}
\put(-20,0){\circle*{5}}
\put(0,20){\circle*{5}}
\put(0,-20){\circle*{5}}
\end{picture}
}
}
&
\scalebox{\scalefactor}{
\fbox{
\begin{picture}(100,100)(-50,-50)
\color{black}
\thinlines
\multiput(0,40)(0,-20){5}{\dashline{4}(-45,0)(45,0)}
\multiput(40,0)(-20,0){5}{\dashline{4}(0,-45)(0,45)}
\color{red}
\put(20,20){\circle*{5}}
\put(20,-20){\circle*{5}}
\put(-20,20){\circle*{5}}
\put(-20,-20){\circle*{5}}
\put(0,0){\circle*{5}}
\put(0,0){\circle{7}}
\end{picture}
}
}
&
\scalebox{\scalefactor}{
\fbox{
\begin{picture}(100,100)(-50,-50)
\color{black}
\thinlines
\multiput(0,40)(0,-20){5}{\dashline{4}(-45,0)(45,0)}
\multiput(40,0)(-20,0){5}{\dashline{4}(0,-45)(0,45)}
\color{red}
\put(20,20){\circle*{5}}
\put(20,-20){\circle*{5}}
\put(-20,20){\circle*{5}}
\put(-20,-20){\circle*{5}}
\put(20,20){\circle{7}}
\put(20,-20){\circle{7}}
\put(-20,20){\circle{7}}
\put(-20,-20){\circle{7}}
\put(0,0){\circle*{5}}
\put(0,0){\circle{7}}
\put(0,0){\circle{10}}
\put(40,0){\circle*{5}}
\put(-40,0){\circle*{5}}
\put(0,40){\circle*{5}}
\put(0,-40){\circle*{5}}
\end{picture}
}
}
&
\scalebox{\scalefactor}{
\fbox{
\begin{picture}(100,100)(-50,-50)
\color{black}
\thinlines
\multiput(0,40)(0,-20){5}{\dashline{4}(-45,0)(45,0)}
\multiput(40,0)(-20,0){5}{\dashline{4}(0,-45)(0,45)}
\color{red}
\put(20,20){\circle*{5}}
\put(20,-20){\circle*{5}}
\put(-20,20){\circle*{5}}
\put(-20,-20){\circle*{5}}
\put(20,20){\circle{7}}
\put(20,-20){\circle{7}}
\put(-20,20){\circle{7}}
\put(-20,-20){\circle{7}}
\put(40,40){\circle*{5}}
\put(40,-40){\circle*{5}}
\put(-40,40){\circle*{5}}
\put(-40,-40){\circle*{5}}
\put(0,0){\circle*{5}}
\put(0,0){\circle{7}}
\put(0,0){\circle{10}}
\put(0,0){\circle{13}}
\put(40,0){\circle*{5}}
\put(-40,0){\circle*{5}}
\put(0,40){\circle*{5}}
\put(0,-40){\circle*{5}}
\end{picture}
}
}
&
\scalebox{\scalefactor}{
\fbox{
\begin{picture}(100,100)(-50,-50)
\color{black}
\thinlines
\multiput(0,40)(0,-20){5}{\dashline{4}(-45,0)(45,0)}
\multiput(40,0)(-20,0){5}{\dashline{4}(0,-45)(0,45)}
\color{red}
\put(20,40){\circle*{5}}
\put(40,20){\circle*{5}}
\put(20,-40){\circle*{5}}
\put(40,-20){\circle*{5}}
\put(-20,40){\circle*{5}}
\put(-40,20){\circle*{5}}
\put(-20,-40){\circle*{5}}
\put(-40,-20){\circle*{5}}
\put(20,0){\circle*{5}}
\put(-20,0){\circle*{5}}
\put(0,20){\circle*{5}}
\put(0,-20){\circle*{5}}
\put(20,0){\circle{7}}
\put(-20,0){\circle{7}}
\put(0,20){\circle{7}}
\put(0,-20){\circle{7}}
\put(20,0){\circle{10}}
\put(-20,0){\circle{10}}
\put(0,20){\circle{10}}
\put(0,-20){\circle{10}}
\end{picture}
}
}
\\
$[100]$ or $[001]$
&
$[010]$
&
$[101]$
&
$[020]$
&
$[110]$ or $[011]$
\end{tabular}
}
\caption{\label{fig:EgLambda1} Eigenvalues of $\Lambda$ for irreps of $A_3$.}
\end{center}
\end{figure}
This WKB approximation is valid in a certain cone of applicability, these cones are listed for various $n$ in the table where the notation $[a,b]_+$ means ``at least in the range $[a,b]$''. The listed ranges are obtained by intersection of the applicability cones for $\Psi_{(2)}^{[k]}$ which are $-\frac k2+[-2,2]_+$.

Now we compare $\mu$ with the eigenvalues of $\Lambda_{\bf 20'}, \Lambda_{\bf 15}$. The cases when these eigenvalues match with an eigenvalue of $\Lambda_L$ for $L={\bf 20'}$ or ${\bf 15}$ and the applicability cones allow deciding that $(\Psi_{(2)}^{+}\otimes \Psi_{(2)}^{[1-2n]})_{L}$ is an $S$ or $S^*$-solution are marked in the table $S,S^*$. In these cases we can unambiguously, up to normalisation, identify $(\Psi_{(2)}^{+}\otimes \Psi_{(2)}^{[1-2n]})_{L}$ with a concrete solution $\Psi$ of \eqref{eq:dA} and hence one also knows that
\be
\label{eq:VPsi}
(V^{+}\otimes V^{[1-2n]})_{L}\propto \Psi(x=0)\,.
\ee
The coefficient of proportionality is fixed by analysing the prefactors of the large-$x$ asymptotics of $(\Psi_{(2)}^{+}\otimes \Psi_{(2)}^{[1-2n]})_{L}$. 

The identification \eqref{eq:VPsi} becomes valuable if we can realise $\Psi(x=0)$ as a projection from some other tensor product, for instance $n=1$ and $L={\bf 15}$ is precisely $V^{+}\wedge V^{-}$ which is the \lhs of \eqref{eq:QQi}. Equalities between projections of different tensor products shall be called fusion relations.

The cases when one cannot identify $(\Psi_{(2)}^{+}\otimes \Psi_{(2)}^{[1-2n]})_{L}$ with an $S$ or $S^*$-solution are denoted as NS. For instance, the issue with $n=3,4$ is that the applicability cone is inappropriate, notably it does not feature the fastest descent line for the corresponding $\mu$.
\newline
\newline
The case of the trivial representation ${\bf 1}$ is a bit special. Equation \eqref{eq:dA} has then the unique solution which is constant in $x$. Hence projection of $(\Psi_{(2)}^{+}\otimes \Psi_{(2)}^{[1-2n]})_{L}$ on the trivial representation will be always a constant in $x$. However, it can have a non-trivial dependence on $\spm$, and only in the case when the WKB analysis can be applied can we conclude  the value of $(V^{+}\otimes V^{[1-2n]})_{{\bf 1}}$. This is the $n=3$ case in the table which in components can be written as
\be
\label{eq:qe1}
\frac 14\epsilon^{abcd}Q_{ab}^{[2]}Q_{cd}^{[-2]}=1\,,
\hphantom{abcdefgabcdefg}
V_i^{[2]}g^{ij}V_{j}^{[-2]}=1\,.
\ee

It may also happen that $\mu$ {\it is not} an eigenvalue of $\Lambda_L$ in an irrep $L$. Then, if the WKB approximation is valid along the direction of the fastest descent and, for this direction $\mu\, x^{M+1}>\Re(\mu'\, x^{M+1})$ for all the eigenvalues $\mu'$ of $\Lambda_L$ then $\Psi_{(2)}^{+}\otimes \Psi_{(2)}^{[1-2n]}$ is sub-dominant compared to any solution of \eqref{eq:dA} in the representation $L$ which is only possible if $(\Psi_{(2)}^{+}\otimes \Psi_{(2)}^{[1-2n]}))_L=0$ and hence $(V^{+}\otimes V^{[1-2n]})_L=0$. We call this a projection property.

There are three instances of the projective property in the ${\bf 6}\otimes {\bf 6}$ example. One of them, $(V\otimes V)_{\bf 15}=V\wedge V=0$, is obvious while the other two are more interesting:
\begin{subequations}
\label{eq:Prboth}
\begin{align}
\label{eq:Pr1}
\epsilon^{abcd}Q_{ab}Q_{cd}&=0,
&
V_i g^{ij}V_{j}&=0\,,
\\
\label{eq:Pr2}
\epsilon^{abcd}Q_{ab}^{+}Q_{cd}^{-}&=0,
&
V_i^{+}g^{ij}V_{j}^{-}&=0\,.
\end{align}
\end{subequations}
Equations \eqref{eq:Pr1}  are certain relations between \Plucker coordinates of $\sl_{N=4}$ and $\so_{M=6}$ flags. For the $\sl_N$ case, \eqref{eq:Pr1} is the famous \Plucker quadric telling us that the two-form $Q_{(2)}$ identifies a plane embedded into $\mathbb{C}^N$. For the $\so_M$ flag, this is an assertion that all lines embedded into $\mathbb{C}^M$ are null. Equations \eqref{eq:Pr2} are relations that are featured by the fused flag. 

Finally we note that certain combinations of Q-functions are of physical significance even if they cannot be studied using the WKB analysis of \eqref{eq:dA}. For instance, as indicated in \eqref{eq:tab315}, singlets constructed from the vector representation should be interpreted, in appropriate explicit systems, as transfer matrices in the symmetric powers of the vector representation:
\be
\frac 14\epsilon^{abcd}Q_{ab}^{[s+2]}Q_{cd}^{[-s-2]}=
V_i^{[s+2]}g^{ij}V_{j}^{[-s-2]}=T_{2,s}\,.
\ee
The above study of ${\bf 6}\otimes{\bf 6}$ features all the properties we wanted to demonstrate. In the same way other tensor products can be studied. Below we summarise most of the interesting relations featured in $A_3$ which can be grouped into three classes:

\paragraph{Fusion relations} The main example is the QQ-relations \eqref{eq:QQi}: The WKB analysis of ${\bf 4}\otimes {\bf 4}={\bf 6}\oplus{\bf 10}$ implies $\Psi_{(1)}^+\wedge\Psi_{(1)}^-=\Psi_{(2)}$, ${\bf{\bar 4}}\otimes {\bf{\bar 4}}={\bf 6}\oplus{\bf\overline{10}}$ implies $\Psi_{(3)}^+\wedge\Psi_{(3)}^-=\Psi_{(2)}$, and the above-discussed ${\bf 6}\otimes{\bf 6}$ compared with ${\bf 4}\otimes {\bf\bar 4}={\bf 15}\oplus{\bf 1}$ implies $\Psi_{(2)}^+\wedge\Psi_{(2)}^-=\left(\Psi_{(1)}\Psi_{(3)}\right)_{\bf 15}$. The corresponding QQ-relations written in components are
\begin{subequations}
\label{eq:323}
\begin{align}
W(Q_a,Q_b) &=Q_{ab}\,, &  \qV_{i} &=\qS^{-}_{\alpha} \gamma_{i}^{\alpha\beta}\qS^{+}_{\beta}\,,\\
W(Q^a,Q^b) &=\frac 12\epsilon^{abcd}Q_{cd}\,, & \qV_{i} &=\qC^{-}_{\dot\alpha}\bar\gamma_{i}^{\dot\alpha\dot\beta}\qC^{+}_{\dot\beta}\,,
\end{align}
\be
\label{eq:QQadj}
&&W(Q_{ab},Q_{cd})=-\frac{1}{2}(Q_{a}Q_{bcd}-Q_{b}Q_{acd}-Q_{c}Q_{dab}+Q_{d}Q_{cab})\,,
\\
&&W(V_{i},V_{j})=\qS_{\alpha}\gamma_{ij}^{\alpha\dot \beta}\qC_{\dot\beta}\,.
\ee
\end{subequations}
In the $\sl_4$ interpretation, these are mostly the relations \eqref{eq:Plucker5}, with exception of the three equations featuring pairwise non-equal $a,b,c,d$ in \eqref{eq:QQadj}. These three cases correspond to the projection of the adjoint representation ${\bf 15}$ to the zero-weight space which is not on the Weyl orbit of the highest weight.

In the $\so_6$ interpretation, equations \eqref{eq:323} are instances of fused Fierz relations. Other examples of such a relations are
\begin{subequations}
\begin{align}
\gamma_{(-),ijk}^{\alpha\beta}\qS^-_{\alpha} \qS^+_{\beta} &= V_{(-),ijk}
& \text{{\bf 10} of {\bf 4}}\otimes {\bf 4}\quad &= \quad \text{{\bf 10} of {\bf 6}}\otimes {\bf 6}\otimes {\bf 6}\,,
\\
\bar \gamma^{\dot \alpha \dot \beta }_{(+),ijk}\qC^-_{\dot \alpha} \qC^+_{\dot \beta} &=V_{(+),ijk}
& \text{{\bf 10} of }{\bf{\bar 4}}\otimes {\bf{\bar 4}}\quad &= \quad \text{{\bf 10} of {\bf 6}}\otimes {\bf 6}\otimes {\bf 6}\,,
\end{align}
\end{subequations}
where $V_{(\pm)}$ are the self-/antiself-dual projections of the 3-form $V_{(3)}=V^{[2]}\wedge V\wedge V^{[-2]}$. In components we define the Hodge dual as $\star V_{I} = \epsilon_{I}{}^{J}V_{J}$ and pick the Levi-Civita symbol to satisfy $\epsilon_{123-1-2-3}=1$. The two most important examples are $\star V_{123} = V_{123}, \star V_{12-3} = -V_{12-3}$.

\paragraph{Projection relations} In addition to \eqref{eq:Prboth},  we list a couple of other projection relations
\begin{subequations}
\begin{align}
\label{eq:PuSp}
{\rm{For}}\ n&=0,\pm 2: & Q_a^{[n]}Q^a &=0 & \psi_{\alpha}^{[n]}C^{\alpha\dot\alpha}\eta_{\dot\alpha} &=0 & {\bf 4}\otimes{\bf{\bar 4}} &=\ldots+{\bf 1}+\ldots\,,
\\
{\rm{For}}\ n&=\pm 1: & \epsilon^{abcd}Q_{ab}^{[n]}Q_c &=0 & \qS_{\alpha}^{[n]}\gamma^{ \alpha  \beta}_i V^i &=0 & {\bf 6}\otimes{\bf{4}} &=\ldots+{\bf 20}+\ldots\,,
\\
{\rm{For}}\ n&=\pm 1: & Q_{ab}^{[n]}Q^b &=0 & \qC_{\dot \alpha}^{[n]} \bar\gamma^{\dot  \alpha \dot \beta}_i V^i &=0 & {\bf 6}\otimes{\bf{\bar 4}} &=\ldots+{\bf{\overline{20}}}+\ldots\,.
\end{align}
\end{subequations}
Together with \eqref{eq:Prboth}, these are all the relations establishing that $Q_{(1)}, Q_{(2)}, Q_{(3)}$ are \Plucker coordinates of a complete flag. Equation \eqref{eq:PuSp} for $n=0$ is also known as a pure spinor condition. The fact that we can choose various values of $n$ reflects that we are dealing with a fused flag, its general definition will be given in Section~\ref{sec:fusedflag}.

\paragraph{Quantisation relations} These are special instances of the fusion relations when the target projection representation is trivial, an example is \eqref{eq:qe1}. All of them can be shown to be equivalent to \eqref{eq:Qcond} which itself can be derived by performing the WKB analysis of ${\bf 4}\otimes{\bf 4}\otimes{\bf 4}\otimes{\bf 4}=\ldots+{\bf 1}+\ldots$. By applying analytic Bethe Ansatz, see Section~\ref{sec:AnalyticBetheAnsatz}, the quantisation relation becomes~\footnote{For at least rational spin chains in the defining representation of $\sl_4$.} the Wronskian Bethe equations in terminology of \cite{Chernyak:2020lgw}. They correctly describe the spectrum of the model for {\it any} values of the parameters \cite{2013arXiv1303.1578M,Chernyak:2020lgw} and in this sense are superior to  standard nested Bethe equations.
\subsection{Spectrum of $\Lambda$}
\label{sec:33}

As the previous subsection demonstrated, the spectrum of $\Lambda$ in various irreps contains valuable information for our study. We shall find this spectrum explicitly following closely \cite{FLO91,Masoero:2015lga}, the result is given by Lemma~\ref{thm:lemma21}. Then we explore other related properties of $\Lambda$, notably how its presentation and its eigenvectors are linked to a choice of a Coxeter element, these results shall be used in Section~\ref{sec:opers} where we relate fused flags and opers.

First, from Symanzik rotation (\ref{eq:Symanzik}), it is straightforward to see that
\be\label{eq:easy}
\gamma^{{\rm ad}\rho^{\vee}}\Lambda=\gamma\,\Lambda\,,\quad \gamma\equiv e^{\frac{2\pi\ii}{\Cox}}\,.
\ee
Hence, if $\mu$ is an eigenvalue of $\Lambda$ then $\gamma\,\mu$ is an eigenvalue as well. All the eigenvalues are therefore located on concentric circles, each such circle contains a multiple of $\Cox$ of the eigenvalues, see Fig.~\ref{fig:Lambda-spectrum}.
\begin{figure}[t]
\begin{minipage}[c]{0.5\hsize}
\begin{center}
\resizebox{60mm}{!}{\includegraphics{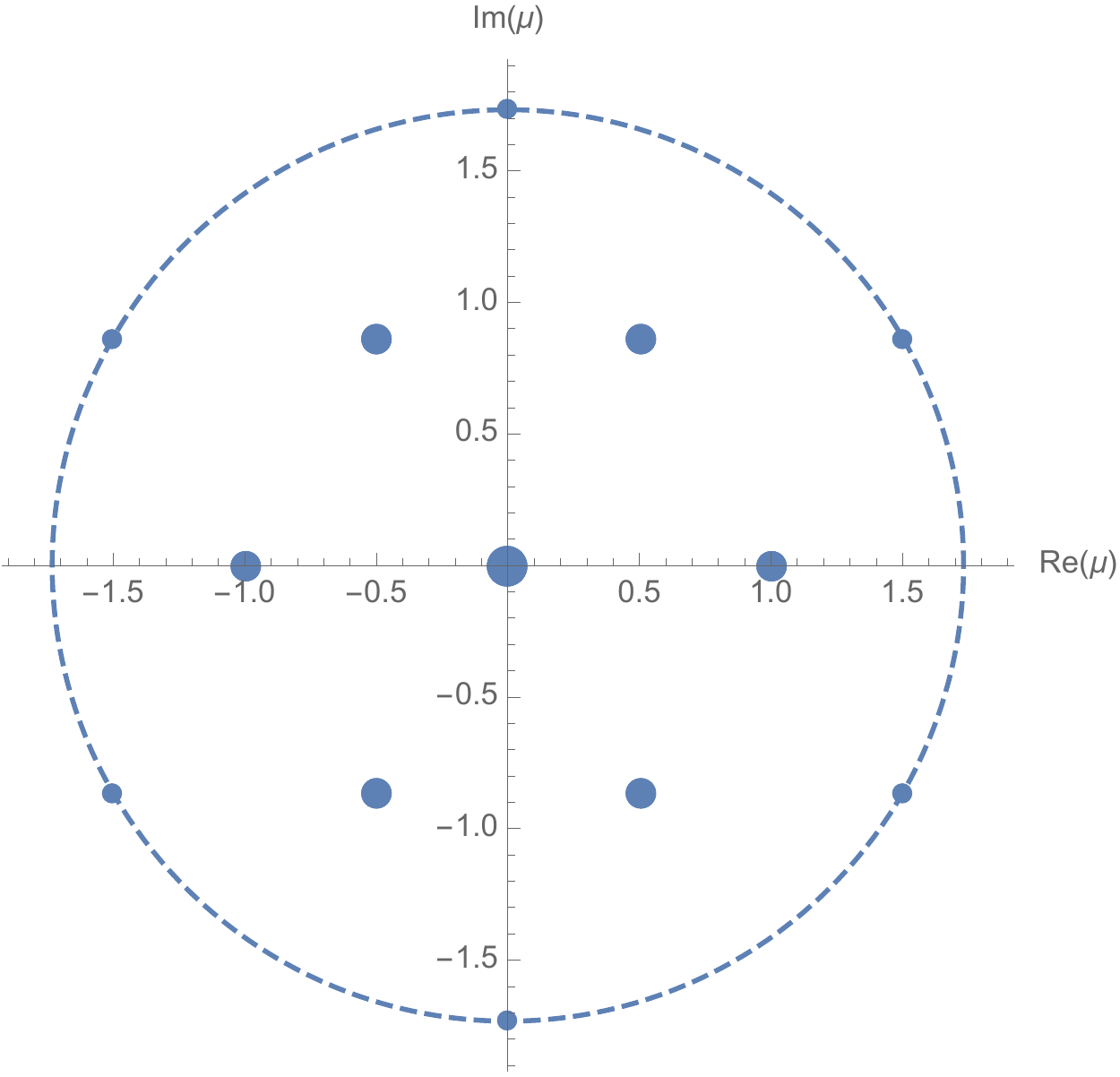}}
\end{center}
\end{minipage}
        \begin{minipage}[c]{0.5\hsize}
        \begin{center}
       \resizebox{60mm}{!}{ \includegraphics{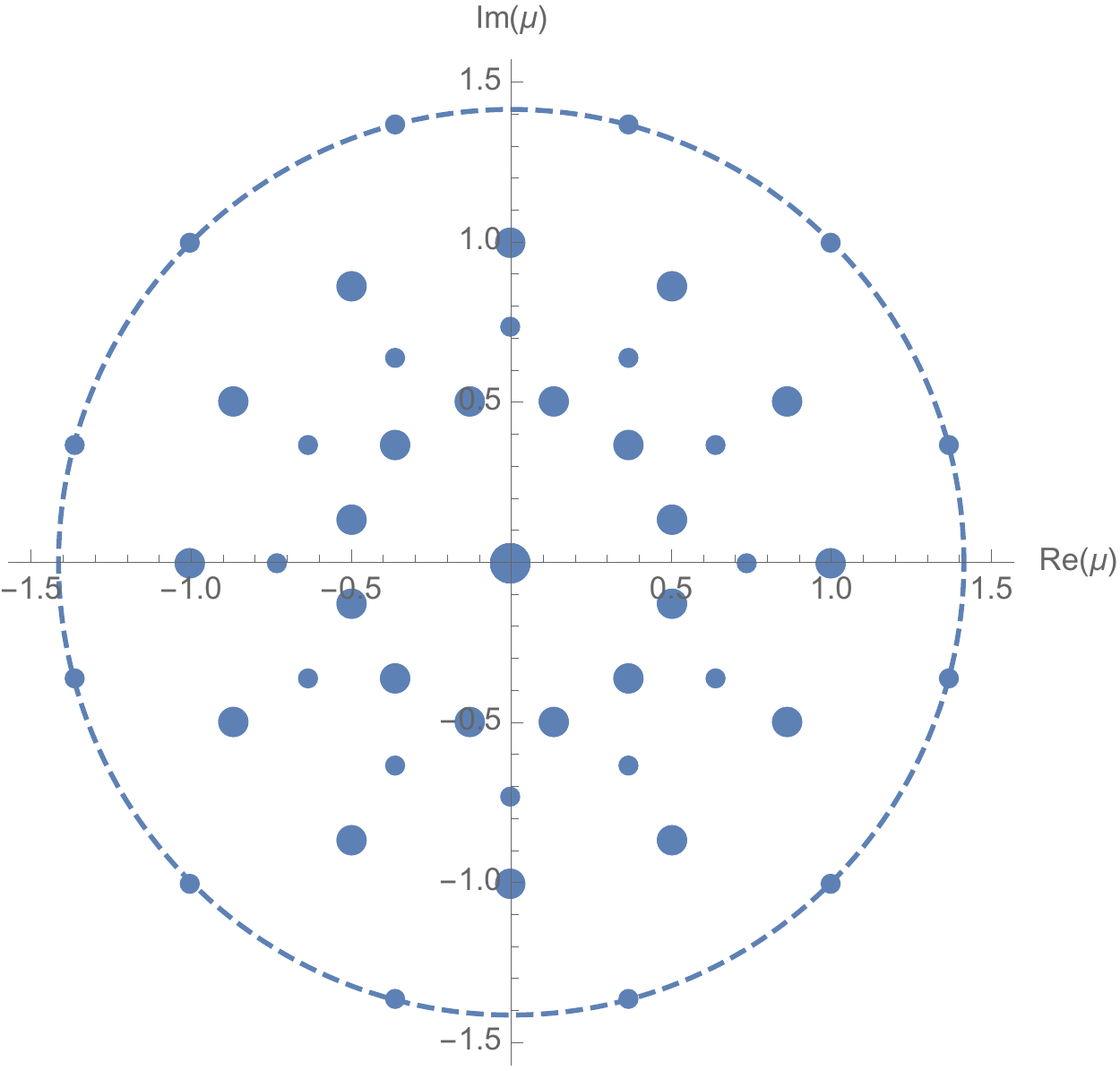}}
        \end{center}
\end{minipage}
\caption{Eigenvalues of $\Lambda$ for the adjoint representation of $D_4$ (left) and $E_6$ (right).}
\label{fig:Lambda-spectrum}
\end{figure}

$\Lambda$ is not ad-nilpotent, hence it can be viewed as an element of a Cartan subalgebra $\mathfrak{h}^\prime$. Choose a simple root system in $(\mathfrak{h}^\prime)^*$, and use the corresponding co-roots $h_a'\in \mathfrak{h}^\prime$, $a=1,\ldots,\rank$ as a basis in $\mathfrak{h}^\prime$ \cite{Kostant}. Our goal is to find the expansion of $\Lambda$ in this basis, $\Lambda=\sum_{a=1}^{\rank} c_a\,h_a'$. Then the spectrum of $\Lambda$ directly follows from the spectrum of $h_a'$.

Relation \eqref{eq:easy} implies that $\gamma^{{\rm ad}\rho^{\vee}}$ is in the normaliser of the maximal torus whose Lie algebra is $\mathfrak{h}^\prime$ and hence it realises action of the Weyl group on $\mathfrak{h}^\prime$. Moreover, it is an element of degree $h$ and hence it should be a Coxeter element. Recall how Coxeter elements are defined: Weyl reflection \wrt the $a$'th simple root~\footnote{In relation to $\mathfrak{h}'$, we will however omit primes for simplicity.} shall be denoted by $s_a$. Its action on $\mathfrak{h}^\prime$ is given by
\be\label{eq:327}
s_{a}(h_{b}^{\prime})=h_{b}^{\prime}-C_{ba}h_{a}^{\prime}\,.
\ee
A Coxeter element is defined as a product of all simple Weyl reflections $\gamma^{{\rm ad}\rho^{\vee}}=\prod\limits_{a=1}^{\rank}s_{a}$. A choice of an order in this product defines the Coxeter element which one considers. All Coxeter elements form one adjoint orbit of the Weyl group.

For a given Coxeter element, there exists a unique eigenvector with eigenvalue $\gamma=e^{\frac{2\pi\ii}{\Cox}}$, and so $\Lambda$ should be, up to normalisation, this eigenvector. $\Re(\Lambda)$ and $\Im(\Lambda)$ span the Coxeter plane -- the unique plane where the Coxeter element acts as the rotation by the angle $\frac{2\pi}{\Cox}$. Hence the spectrum of $\Lambda$ is the projection of the irrep weight space to the Coxeter plane.
\newline
\newline
First, we find $\Lambda$ explicitly for a particularly simple choice of order in $\prod\limits_{a=1}^{\rank}s_{a}$. Introduce a bipartition of the Dynkin diagram into ``even'' and ``odd'' nodes, such that no lines link even with even (or odd with odd). Our convention is that the fundamental representation of the smallest dimension is associated to an even node. Now consider a Coxeter element
$s_{{\rm odd}}s_{{\rm even}}$. Here $s_{{\rm odd}}$ is the product of the reflections of the odd simple roots, and $s_{{\rm even}}$ is the one of the even simple roots. Using \eqref{eq:327}, one finds \cite{BLM89}
\be\label{eq:328}
(s_{{\rm odd}}+s_{{\rm even}})h_{a}^{\prime}=\sum_{a=1}^{r}I_{ab}h_{b}^{\prime}\,,
\ee
where $I_{ab}$ is the incidence matrix of the Dynkin diagram.

Let $q$ be an eigenvector of the Coxeter element with some eigenvalue $\bar\gamma$, $s_{{\rm odd}}s_{{\rm even}}\,q=\bar\gamma\, q$. Parameterise it in the form $q=\bar\gamma^{1/2}q_{\rm odd}+q_{\rm even}$, where $q_{{\rm even}}=\sum\limits_{a\in{\rm even}}\bar\mu_{a}h_{a}^{\prime}$ and $q_{{\rm odd}}=\sum\limits_{a\in{\rm odd}}\bar\mu_{a}h_{a}^{\prime}$. Based on \eqref{eq:327}, one has
\be
\ba
s_{{\rm even}}(q_{{\rm even}})&=-q_{{\rm even}}\,, &
s_{{\rm odd}}(q_{{\rm odd}})&=-q_{{\rm odd}}\,,
\ea
\ee
while using \eqref{eq:328} we can write
\be
\ba
s_{{\rm even}}(q_{{\rm odd}})&=(1+\hat I)\,q_{{\rm odd}} \,, &
s_{{\rm odd}}(q_{{\rm even}})&=(1+\hat I)\,q_{{\rm even}}\,,
\ea
\ee
where $\hat I$ is an operator with matrix entries $I_{ab}$ in the basis $h_a^\prime$.

Using the mentioned properties, one derives from $s_{{\rm odd}}s_{{\rm even}}\,q=\bar\gamma\, q$ that
\be
\hat I(q_{\rm odd}-\bar\gamma^{1/2}q_{\rm even})=(\bar\gamma^{1/2}+\bar\gamma^{-1/2})(q_{\rm even}-\bar\gamma^{1/2}q_{\rm odd})\,.
\ee
Since the incidence matrix is of the graph with only links between nodes of a different type, the above relation can be projected to $\hat I\,{q_{\rm odd}}=(\bar\gamma^{1/2}+\bar\gamma^{-1/2})\,q_{\rm even}$, $\hat I\,{q_{\rm even}}=(\bar\gamma^{1/2}+\bar\gamma^{-1/2})\,q_{\rm odd}$ implying that $q_{\rm odd}\pm q_{\rm even}$ are eigenvectors of $\hat I$ with eigenvalues  $\pm (\bar\gamma^{1/2}+\bar\gamma^{-1/2})$.

As the logic can be reversed, we conclude that all the eigenvalues of $\hat I$ are of the form $\pm  (\bar\gamma^{1/2}+\bar\gamma^{-1/2})$, where $\bar\gamma$ is an eigenvalue of the Coxeter element.  Now we notice that $I_{ab}$ is a matrix of Perron-Frobenius type. Given the established bijection with the eigenvalues of the Coxeter element, the maximal eigenvalue of $I_{ab}$ is identified to be $\gamma^{1/2}+\gamma^{-1/2}$ for $\gamma=e^{\frac{2\pi i}{\Cox}}$. The corresponding eigenvector allows us then to construct $\Lambda$:

\begin{lemma}
\label{thm:lemma21}
Let $(\mu_1, \mu_2,\cdots, \mu_{r})$ be the Perron-Frobenius eigenvector of $I_{ab}$,
\be
\label{eq:PF}
\sum_{b=1}^{\rank}I_{ab}\,\mu_b = (\gamma^{1/2}+\gamma^{-1/2})\,\mu_a\,.
\ee
Then, for a choice of Cartan subalgebra and simple roots such that $\gamma^{{\rm ad}\rho^{\vee}}$ is the Coxeter element $s_{\rm odd}s_{\rm even}$,
\be
\label{eq:tLambda}
\Lambda=\sum_{a=1}^{r}\gamma^{\tilde{p}_a/2}\mu_{a}h_{a}^{\prime}\,,
\ee
where $\tilde p_a=0$ for even Dynkin nodes $a$ and $\tilde p_a=1$ for odd Dynkin nodes $a$.
\qed
\end{lemma}

Now we would like to understand how the expansion $\Lambda=\sum\limits_{a=1}^{\rank} c_a\,h_a'$ looks like for a different choice of a Coxeter element.

First, let us design a way to label different Coxeter elements.  We define Coxeter height function \label{Coxhf} as a  function $p:\{1,\ldots,\rank\}\to \mathbb{Z}$ satisfying the property $p_a-p_b=\pm 1$ if $I_{ab}\neq 0$. Coxeter height functions that differ only by a translation, $p_a\to p_a+n$, $n\in \mathbb{Z}$, shall be considered as equivalent. In view of the equivalence we will always assume that $p_a$ is even if $a$ is an even node.

\begin{lemma}
Coxeter height functions (up to the equivalence) are in bijection with distinct Coxeter elements by the following rule: For $a,a'$ being two adjacent nodes of the Dynkin graph, $s_a$ is to the left of $s_{a'}$ in the product $\prod\limits_{a=1}^{\rank} s_a$ defining a Coxeter element if and only if $p_a>p_{a'}$. 
\end{lemma}
\begin{proof}
Think about the product $\prod\limits_{a=1}^{\rank} s_a$ with a certain order of elements as a word $s_{a_1}\ldots s_{a_\rank}$. We are allowed to exchange two neighbouring letters in the word $\ldots s_a s_{a'} \ldots \simeq \ldots s_{a'} s_{a} \ldots$ if $a,a'$ are not adjacent nodes of the Dynkin graph. Indeed, Weyl reflections $s_a$ and $s_{a'}$ commute in such a case. Two words shall be called equivalent if they can be obtained from one another by a sequence of these exchanges. All Coxeter elements belonging to the same equivalence class coincide. By induction in $\rank$, one shows that the equivalence classes of Coxeter height functions are in bijection with the equivalence classes of words.

Different equivalence classes should define different Coxeter elements because the corresponding Coxeter elements have different eigenvectors with the eigenvalue $\gamma$, as is demonstrated by Lemma~\ref{thm:23} below. 
\end{proof}

\begin{lemma}
\label{thm:23}
If $s[p]$ is the Coxeter element with the height function $p$ then
\be
\label{eq:anyp}
s[p]\Lambda[p]=\gamma\, \Lambda[p]\,,\quad\text{where}\quad \Lambda[p]=\sum_{a=1}^{r}\gamma^{\frac{p_{a}}{2}}\mu_{a}h_{a}^{\prime}\,.
\ee
\end{lemma}
\begin{proof}
Without loss of generality assume that the minimal value of $p$ is either $0$ or $1$, depending on whether it is realised at, respectively, an even or an odd node. Suppose now that $\aast$ is such a node of the Dynkin diagram that $p_{\aast}$ is the maximal value of $p$. Then $p_a=p_{\aast}-1$ for all nodes $a$ adjacent to $\aast$.  On the one hand, we note that $s_{\aast}$ is to the left of all $s_a$, $I_{a\aast}\neq 0$, in $s[p]$. Then, by using commutativity with the other elementary reflections, we can write $s[p]=s_{\aast}s'$ for some $s'$ and hence
\be
s_{\aast}s[p_1,\ldots, p_{\aast},\ldots,p_\rank]s_{\aast}=s's_{\aast}=s[p_1,\ldots, p_{\aast}-2,\ldots,p_\rank]\,.
\ee
On the other hand, by acting with $s_{a_\ast}$ on $\Lambda[p]$, we find
\be
s_{a_{\ast}}\Lambda[p_{1},\cdots, p_{a_{\ast}},\cdots,p_{r}]&=&\gamma^{\frac{p_{a_{\ast}}}{2}}h_{a}^{\prime}\big(-\mu_{a_{\ast}}+\sum_{a=1}^{r}\gamma^{-\frac{1}{2}}\mu_{a}I_{aa_{\ast}}\big)+\sum_{a\neq \aast }\gamma^{\frac{p_{a}}{2}}\mu_{a}h_{a}^{\prime}
\nonumber\\
&=&\Lambda[p_{1},\cdots, p_{a_{\ast}}-2,\cdots,p_{r}]\,,
\ee
where \eqref{eq:PF} was used.

We thus could decrease $p_{a_\ast}$ to $p_{a_{\ast}}-2$ by acting with the same reflection on both $s[p]$ and $\Lambda[p]$. We repeat this progress, decreasing at each step one $p_{a_\ast}$ at a node $a_\ast$ that has currently the maximal value of $p$. If there are several nodes with the maximal value, they are not adjacent and hence we can decrease their value in any order. We continue until we obtain $p=\tilde p$, $s[\tilde p]=s_{\rm odd}s_{\rm even}$. We already established that $\Lambda[\tilde p]$ is the desired eigenvector of $s[\tilde p]$, \cf \eqref{eq:tLambda}. It remains to reverse all the performed reflections $s_{a_\ast}$ to prove the same for any $p$ thus confirming \eqref{eq:anyp}.
\end{proof}

In conclusion, we managed to find (up to an inessential normalisation) the explicit form for the originally defined by \eqref{eq:Ladef2} $\Lambda$ in a reference frame where $\gamma^{{\rm ad}\rho^{\vee}}$ is the Coxeter element $s[p]$: $\Lambda=\Lambda[p]$. Now it is also obvious that the spectrum of $\Lambda[p]$ is the same for any choice of the Coxeter height function $p$.

Finally, let us also comment about interpretation of the eigenvectors of $\Lambda$. Denote by $\mathsf{U}_{(a)}^{[p_a]}$ the eigenvector of $\Lambda_{L(\omega_a)}$ with the eigenvalue $\gamma^{p_a/2}\mu_a$. If we are in a frame where $\Lambda=\Lambda[p]$, we conclude that $h_b'\mathsf{U}_{(a)}^{[p_a]}=\delta_{ab}$ and hence $\mathsf{U}_{(a)}^{[p_a]}$ gets meaning of the highest-weight vector in the $a$'th fundamental representation. Importantly, we can make this conclusion simultaneously for all $a$:

\begin{lemma}
For every Coxeter height function $p$, there exists a Cartan subalgebra $\mathfrak{h}^\prime$ and a choice of simple roots such that $ \mathsf{U}_{(1)}^{[p_1]},\ldots,\mathsf{U}_{(\rank)}^{[p_\rank]}$ are the highest-weight vectors in the corresponding fundamental representations of $\algg$.
\qed
\end{lemma}

\section{Extended Q-system}\label{sec:Extended Q-system}
In this section we explore, using general formalism of representation theory, various relations between the Q-functions of the extended Q-system and their geometric interpretation. Recall that the Q-functions of the extended Q-system are the components of the vectors $Q_{(a)}$ in the $a$'th fundamental representations of the Lie algebra.

\subsection{Relations between Q-functions}
Generalising from the $\so_6\simeq \sl_4$ example, we organise all possible relations into three categories: fusion, quantisation and projection. 

Let $A={a_1,\ldots,a_{|A|}}$ be an ordered set comprising $|A|$ elements from $\{1,2,\ldots,\rank\}$, possibly with repetitions. Choose also some integers $n_1,\ldots,n_{|A|}$  and construct the following function
\be
\Psi=\bigotimes_{i=1}^{|A|}\Psi_{(a_i)}^{[n_i]}
\ee
which is naturally a vector in the representation
\be
\label{eq:935}
L:=\bigotimes_{i=1}^{|A|} L(\omega_{a_i})=\bigoplus\limits_{\omega}L(\omega)\,;
\ee
we also noted the decomposition of $L$ into irreps, the sum $\bigoplus\limits_{\omega}$ may feature repetitions of $\omega$. The largest $\omega$ that appears in this sum is $\omega_{\rm max}=\sum_{i=1}^{|A|}\omega_{a_i}$.

We demand that $\Psi$ is a solution of \eqref{eq:dA} and hence restrict $n_i$ to be even if $a_i$ is an even node of the Dynkin diagram and odd if $a_i$ is an odd node.

The cone of applicability of $\Psi$ is $[\alpha,\beta]_+$ with $\alpha=-\frac{1}{2}\min(n_1,\ldots,n_{|A|})-\frac \Cox 2$, and $\beta=-\frac{1}{2}\max(n_1,\ldots,n_{|A|})+\frac \Cox 2$. In this cone, the large-$x$ approximation \eqref{eq:WKB0} of $\Psi$ follows from that of $\Psi_{a_i}^{[n_i]}$, and the associated eigenvalue is $\mu=\sum_{i=1}^{|A|}\gamma^{\frac{n_i}{2}}\mu_{(a_i)}$. For the statements below, it is important that this cone is non-empty, and also it must be large enough in certain cases.

\paragraph{Fusion relations}
If $\mu$ is an eigenvalue of $\Lambda_{L(\omega)}$ in some irrep $L(\omega)$ appearing in the decomposition \eqref{eq:935} then one can apply the WKB analysis for $\Psi$ restricted to this irrep: If for each $\mu'$ -- eigenvalue of $\Lambda_{L(\omega)}$ different from $\mu$ -- there exists a direction $k\in [\alpha,\beta]_+$ such that $\Re(\mu' e^{\frac{2\pi\,\ii}{\Cox} k})< \Re(\mu e^{\frac{2\pi\,\ii}{\Cox} k})$ then $\Psi_{L(\omega)}$ is a solution of \eqref{eq:dA} of $S^*$-type. If directions $k$ for each $\mu'$ can be made equal than this a solution of $S$-type. In either case, it is fixed uniquely by its large-$x$ asymptotics.

Specialising to the Q-system, one writes
\be
\left(\bigotimes_{i=1}^{|A|}Q_{(a_i)}^{[n_i]}\right)_{\Lambda_{L(\omega)}}=z^{-\frac{\rho^{\vee}}{h\,M}}\left(\Psi(0)\right)_{\Lambda_{L(\omega)}}\,.
\ee
If there is a different way to get the same $S$-/$S^*$-solution, \eg using a set $A'$ and associated integers $n_i'$, then, from the uniqueness of such a solution, we derive the fusion relation
\be
\label{eq:fr}
\left(\bigotimes_{i=1}^{|A|}Q_{(a_i)}^{[n_i]}\right)_{\Lambda_{L(\omega)}}\propto \left(\bigotimes_{i=1}^{|A'|}Q_{(a'_i)}^{[n'_i]}\right)_{\Lambda_{L(\omega)}}\,.
\ee
We stress that the coefficient of proportionality does not depend on the spectral parameter. It is just a number. Indeed, it can be fixed from the comparison of the large-$x$ asymptotics \eqref{eq:WKB0} which does not depend on the spectral parameter.

Typical examples of the fused relations are the QQ-relations \eqref{eq:QQi}, relations featuring Wronskians \eg \eqref{eq:QA}, and a fused version of Fierz identities.

\paragraph{Quantisation relations} This is a special instance of the fused relations, when $L(\omega)$ is the trivial representation. In this case we do not need two different ways to realise the same $S/S^*-$solution but instead we can write
\be
\label{eq:qr}
\left(\bigotimes_{i=1}^{|A|}Q_{(a_i)}^{[n_i]}\right)_{\Lambda_{L(0)}}\propto 1\,.
\ee
It is important that the cone of applicability is non-empty to fix the normalisation constant from the large-$x$ asymptotics and, in particular, to conclude that it does not depend on the spectral parameter.

We call this type of relations quantisation relations because they essentially constrain possible functional dependence of the Q-functions on the spectral parameter. And indeed, in the example of the rational $\sl_{\rank+1}$ case, we know that the quantisation condition \eqref{eq:Qcond} selects a finite set of polynomials $Q_a$ which are precisely all physical solutions, at least in the case of spin chains in the defining representations.

There is one particular quantisation relation which we would like to mention explicitly. Let $Q_{(a^*)}$ be the Q-vector in the contra-gradient representation~\footnote{obtained by minus transposition of the representation matrices.}  compared to the Q-vector $Q_{(a)}$. Then we can always, by the natural pairing, construct a singlet from these two Q-vectors. Abusing a bit terminology, we shall refer to it as a scalar product and denote it by $\langle\cdot,\cdot\rangle$.

The quantisation relation reads
\be
\label{eq:qrh}
\langle Q_{(a)}^{[h/2]},Q_{(a^*)}^{[-h/2]}\rangle=1\,.
\ee
Normalisaton of the Q-functions is fixed, up to symmetries, by demanding equality and not just proportionality in \eqref{eq:QQi}. Our convention is to scale the definition of the scalar product in a way to get identity on the \rhs of \eqref{eq:qrh}.

The proof of \eqref{eq:qrh} is simple using the WKB analysis: Eigenvalues of $\Lambda$ are the same for a representation and its contra-gradient, and $\gamma^{h/4}+\gamma^{-h/4}=0$. Hence  $\langle \Psi_{(a)}^{[h/2]},\Psi_{(a^*)}^{[-h/2]}\rangle$ has constant large-$x$ asymptotics, while its cone of applicability is non-zero.

\paragraph{Projection relations} Finally, it may happen that $\mu$ is not an eigenvalue of $\Lambda_{L(\omega)}$ for a particular $\omega$. In such a case, if for each $\mu'$ -- eigenvalue of  $\Lambda_{L(\omega)}$  -- there exists a direction $k\in [\alpha,\beta]_+$ such that $\Re(\mu' e^{\frac{2\pi\,\ii}{\Cox} k})< \Re(\mu e^{\frac{2\pi\,\ii}{\Cox} k})$ then $\Psi_{L(\omega)}=0$ implying
\be\label{eq:pr}
\left(\bigotimes_{i=1}^{|A|}Q_{(a_i)}^{[n_i]}\right)_{\Lambda_{L(\omega)}}=0\,.
\ee
The above equality is easy to verify if $n_i=0$ for the even Dynkin nodes and $n_i=\pm 1$ (same sign for all $i$) for the odd Dynkin nodes, and $\omega<\omega_{\rm max}$. Condition on $n_i$ can be further relaxed which is the subject of Section~\ref{sec:fusedflag}.

There is also an important projection property featured by the scalar product:
\be
\label{eq:prh}
\langle Q_{(a)}^{[n]},Q_{(a^*)}^{\vphantom{[n]}}\rangle=0\,,
\ee
for $n=h-2,h-4,\ldots,2-h$. Its derivation is based on the fact $\gamma^{n/2}+1\neq 0$ implying that the eigenvalue controlling the large-$x$ asymptotics of $\langle \Psi_{(a)}^{[n] },\Psi_{(a^*)}^{\vphantom{[n]}}\rangle$ is non-zero, and the cone of applicability contains the fastest descent line for the mentioned $n$.  This is a sub-dominant asymptotic behaviour compared to the constant in $x$ solution featured by the trivial representation.

The projection relations have a natural geometric interpretation of generalised \Plucker relations as we shall soon see.
\subsection{Universality of the Q-system}
The relations listed in the previous subsection are derived for very specific Q-functions that originate from solutions of the linear problem \eqref{eq:dA} according to \eqref{eq:Baxvec}. We will now show that all these relations can be systematically imposed on Q-functions which are not \apriori linked to some ODE/IM problem. So the relations should be actually based on representation theory of the Lie algebra alone, they are not an exclusive feature of \eqref{eq:dA}.

To be specific, recall the terminology that we use: Ensemble of Q-functions of type $Q_{(a),\sigma(1)}$ is said to be a Q-system on the Weyl orbit if these Q-functions satisfy \eqref{eq:Plucker4}. Ensemble of Q-vectors $Q_{(a)}$ is said to be an extended Q-system if they satisfy all the fusion, quantisation, and projection relations introduced in the subsection above. 

\begin{theorem}
\label{thm:uni}
For any generic enough choice of the functions $Q_{(a),1}$, $a=1,\ldots,\rank$, there exists the unique, up to symmetries, Q-system on the Weyl orbit containing $Q_{(a),1}$. 
\end{theorem}
\begin{theorem}
\label{thm:unie}
For any Q-system on the Weyl orbit, there exists the unique extended Q-system containing the Q-system on the Weyl orbit as its part.
\end{theorem}

Let us explain what ``up to symmetries'' mean. For  $g$ -- any $\groupG$-valued periodic function of the spectral parameter -- transformation $Q_{(a)}\to g\,Q_{(a)}$ is a symmetry of the Q-system. By the condition of Theorem~\ref{thm:uni}, $Q_{(a),1}$ are fixed which restricts the symmetry to the unipotent radical $\groupN=[\groupB,\groupB]$. Hence, when computing the Weyl orbit Q-system, we should look for solutions modulo the transformations $Q_{(a)}\to g\,Q_{(a)}$ with periodic functions $g$ that take values in $\groupN$. Once the Weyl orbit Q-system is fixed, there is no residual symmetry left; Theorem~\ref{thm:unie} implies that further extension to the full extended Q-system has no ambiguities.

\begin{proof}[Proof of Theorem \ref{thm:uni} ~\footnotemark] 
~\footnotetext{Certain technical aspects of the proof will be better clarified in the sequel of this paper where we plan to present them in a unifying setting covering also non-simply-laced cases.}
Let the total number of unknown Q-functions on the Weyl orbit be $\#_{\rm unkn}$.  Consider an explicit algorithm that selects and solves a subset of $\#_{\rm unkn}$ equations from \eqref{eq:Plucker4} to compute the unknown Q-functions. Namely, each equation \eqref{eq:Plucker4} relates four functions (five in the case of the bifurcation node of the Dynkin diagrram). In the algorithm, one considers an equation with all but one already computed Q-functions to compute the remaining one, and proceeds recursively. Existence of such a recursion to compute all the Q-functions is a consequence of the results in \cite{MV05}.

Most steps of the recursion are straightforward where we fix an unknown Q-function as
$$
Q=\text{A rational combination of already fixed Q-functions (probably with some shifts).}
$$
However, precisely on $\dim\groupN$ occasions one encounters equation
\be
W(Q_a,Q_b)=\prod_{c\in \{c_1,c_2,\ldots\}} Q_{c}\,,
\ee
where $Q_b$ is the unknown. It is solved as follows. We  fix some large enough  integer $R$ and write a solution as
\be\label{eq:sportsol}
Q_{b}^{[2n+p]}=Q_{a}^{[2n+p]}\left(\sum\limits_{-m\leq k<n}\left(\frac{\prod\limits_{c\in \{c_1,c_2,\ldots\}} Q_{c}^+}{Q_a^{[2]}Q_a}\right)^{[2k+p]}+\frac{Q_{b}^{[-2R+p]}}{Q_{a}^{[-2R+p]}}\right)\,,
\ee
where $n$ is an integer $n>-R$ and $p=0$ or $1$. The term $\frac{Q_{b}^{[-2R+p]}}{Q_{a}^{[-2R+p]}}$ should be viewed as an integration constant, we can set it to any value using the residual symmetry of the problem (this is ``up to symmetries'' part of the theorem).

We therefore see that all $Q_{(a),i}^{[n]}$~\footnote{Admissible values of $n$ are bounded from below due to existence of $R$ in \eqref{eq:sportsol}, but this is not an obstruction as any explicit equation features only finitely many values of $n$ and we can take $R$ large enough.} can be written as rational combinations of $Q_{(a),1}^{[m]}$ for a range of $a,m$. $Q_{(a),1}^{[m]}$ for each $a$ and $m$ shall be considered as independent variables that assume certain numerical values -- the input to the system of equations \eqref{eq:Plucker4}. ``Generic enough choice'' of $Q_{(a),1}$ means that the denominators in the encountered rational combinations do not vanish, \ie that $Q_{(a),1}^{[m]}$ take values in a Zariski-open set.

Consequently, any relation between Q-functions, for instance any of the yet unused equations from \eqref{eq:Plucker4}, becomes of type
\be
\label{eq:1171}
\text{Rational function of } Q_{(a),1}^{[m]} =0\,.
\ee
It suffices to show that this rational function vanishes on a dense set to conclude that it is identically zero. To this end consider a generalisation of \eqref{eq:dA}:
\begin{align}
\label{eq:dAg}
\left(\frac d{dx}+\sum_{i=1}^\rank f_i\left(\frac{x}{z^{1/h\,M}}\right)H_i+\sum_{i=1}^\rank g_i\left(\frac{x}{z^{1/h\,M}}\right)E_{\alpha_i}+h \left(\frac{x}{z^{1/h\,M}}\right) (x^{hM}-z)E_{\alpha_0})\right)\Psi=0\,.
\end{align}
This generalisation with dexterously chosen $f_i,g_i$, $h$ was used in \cite{Masoero:2018rel,Masoero:2019wqf} to describe excited states of the quantum $\hat{\algg}$-KdV model, see also \cite{Bazhanov:2003ni,Fioravanti:2004cz,kar72958}.

The generalised equation still enjoys symmetry \eqref{eq:Symanzik}. Hence, if its asymptotic large-$x$ behaviour, at least in the relevant directions, coincides with the one of the original equation \eqref{eq:dA}, all the equations satisfied by the  defined by \eqref{eq:Baxvec}  Q-vectors and as a consequence of the WKB analysis will hold. In particular \eqref{eq:Plucker4} will hold.  There are $2\rank+1$ functions $f_i,g_i$, $h$ which can be be used as a functional freedom to engineer Q-vectors. However, $\rank$ of them can be fixed using gauge transformations that do not spoil the structure of the equation and one can be absorbed by a reparameterisation of $x$. The actual functional freedom to non-trivially modify the Q-system are the remaining $\rank$ functions which is precisely what we need to vary $Q_{(a),1}$ and form a dense set.
\end{proof}
If Q-functions are holomorphic functions of the spectral parameter $\spm$ then there should exist such its values $\spm_0$ (forming a Zariski-open set) that Q-functions are generic if evaluated at any point of a vicinity of $\spm_0$. Afterwards, Q-functions can be analytically continued outside the mentioned vicinities and the analytic continuation may reveal poles or other singularities for instance branch cuts; a scenario with branch cuts is realised in AdS/CFT integrable systems \cite{Beisert:2004hm}. A typical requirement to impose is that Q-functions describing the spectrum of a physical model have singularities only of a special type and at special values of the spectral parameter, see  Section~\ref{sec:AnalyticBetheAnsatz} for explicit examples.

\begin{proof}[Proof of Theorem \ref{thm:unie}]
For the $A_\rank$ case, all the Q-functions are already on the Weyl orbits. For the other algebras, it suffices to present an algorithm to compute all the extended Q-functions from the Q-functions on the Weyl orbit. Then we can represent any fusion, quantisation, or projection relation in the form \eqref{eq:1171} and use the same argumentation based on \eqref{eq:dAg} to conclude that any such relation is identically satisfied.

For $D_\rank$ algebras, all the Q-functions of the vector and both spinor representations are on the Weyl orbit. The other Q-functions can be computed \via the Wronskian determinant \eqref{eq:VWdet}.

For $E_6$, all the Q-functions of the two $27$-dimensional  representations are on the Weyl orbit. Explicit ways to compute the other Q-functions are presented in Section~\ref{sec:E6}

For $E_7$, all the Q-functions of the  $56$-dimensional representation are on the Weyl orbit.  Explicit ways to compute the other Q-functions are presented in Section~\ref{sec:E7}.

For $E_8$, the smallest non-trivial representation is the adjoint representation. In this representation, the zero weight vectors (Cartan subalgebra) are not on the Weyl orbit, but the Cartan subalgebra Q-functions can be computed using \eqref{eq:fijk}, see the explanation that follows this equation. From Q-functions of the adjoint representation, all the other Q-functions are computable using \eqref{eq:821}. 
\end{proof}
Note that all the extended Q-functions are computed polynomially from the Weyl orbit Q-functions, no divisions are encountered. Hence Theorem~\ref{thm:unie} does not require a generic position assumption.

\subsection{Fused flag}
\label{sec:fusedflag}
Recall some basic facts about compact homogeneous spaces (see \eg \cite{Fulton-Harris}~\textsection{23.3}). These spaces are of the form $\groupG/\groupP$, where $\groupP$ is a parabolic subgroup. Parabolic subgroups can be defined as the ones containing a Borel subgroup $\groupB$. In the following $\groupB$ is assumed to be fixed. The set of all $\groupP$'s containing $\groupB$ is partially ordered by inclusion:
\be
\groupB\equiv \groupP_{\fullset}\subset \ldots \subset \groupP_{a_1a_2a_3} \subset \groupP_{a_1a_2} \subset \groupP_{a_1} \subset \groupP_{\emptyset}\equiv \groupG\,,
\ee
where the Lie algebra of $\groupP_{a_1\ldots a_k}$ is generated by the Cartan generators, the raising operators $E_{\alpha_a}$ for all simple roots $\alpha_a$, and by the lowering operators $E_{-\alpha_a}$,  such that $a\neq a_i$, $i=1,\ldots,k$. In particular, the proper maximal parabolic subgroups of $\groupG$ are $\groupP_a$, $a=1,\ldots,\rank$. If $\groupG=\Sl_n$, $\groupG/\groupP_a$ is the Grassmannian manifold ${\bf Gr}(a,n)$.

A concrete way to realise  $\groupG/\groupP$ is by considering a representation whose highest-weight eigenspace is invariant under action of $\groupP$. Then $\groupG/\groupP$ is the orbit of the highest-weight vector under action of $\groupG$ in the representation space considered projectively (\ie up to normalisations). In the case of $\groupP_a$, the minimal such representation is the $a$'th fundamental representation. Let vectors of this representation have components $V_{(a),i}$, for $i=1,2,\ldots,\dim L(\omega_a)$. We call $V_{(a),i}$ the extended \Plucker coordinates~\footnote{For $\Gl_n$, these are normal \Plucker coordinates. In the works \cite{1987RuMaS..42..133G, 1998math......7079F} the name ``generalised \Plucker coordinates'' refers to $V_{(a),i}$ with $i$ being only on the Weyl orbit of the highest-weight vector. This orbit is also important for us,  \cf \eqref{eq:Plucker4}. The generalised \Plucker coordinates are used to identify the Bruhat cell to which a given point of $\groupG/\groupB$ belongs to but, in contrast to the extended coordinates, they are not sufficient to identify the point uniquely.} if they are the coordinates of the $\groupG$-orbit of the highest-weight vector. They are projective coordinates %
\be
\label{eq:Pluco}
[V_{(a),1}:V_{(a),2}:\ldots:V_{(a),\dim L(\omega_a)}]
\ee
that define embedding of $\groupG/\groupP_a$ into $\mathbb{P}L(\omega_a)$.

Consider now the minimal parabolic subgroup which is the Borel subgroup $\groupB$ itself. In this case, the compact homogeneous space $\groupG/\groupB$ is called the complete flag manifold (in the following, simply flag manifold). To describe this space, one considers the orbit of the highest weight vector in $L(\rho)$, where $\rho=\sum_{a=1}^{\rank}\omega_a$ is the Weyl vector. It is also practical to embed this orbit into a bigger representation $L(\omega_1)\otimes L(\omega_2)\otimes\ldots \otimes L(\omega_\rank)$ because the latter is naturally parameterised by the products $\prod\limits_{a=1}^{\rank} V_{(a),i_a}$, for all tuples $i_1\ldots i_{\rank}$. When we are on the highest-weight orbit, these products are in (projective) one-to-one correspondence with the sets of \Plucker coordinates \eqref{eq:Pluco} and so we can use $V_{(a),i}$ for all $a$ and the corresponding all $i$ to parameterise flags. By the same logic, we can use components of $V_{(a)}$ for $a\in\{a_1,\ldots,a_k\}$ to parameterise partial flags -- points of $\groupG/\groupP_{a_1\ldots a_k}$.

Extended \Plucker coordinates satisfy (generalisations of) \Plucker relations that can be obtained as follows. Consider some set $A$ composed from (possibly repeating) numbers $1,2,\ldots,\rank$. Consider the decomposition into irreps of the following tensor product 
\be
\bigotimes\limits_{a\in A}L(\omega_a)=L\left(\omega_{\rm max}\equiv \sum_{a\in A}\omega_a\right)+\bigoplus\limits_{\omega<\omega_{\rm max}}L(\omega)\,.
\ee
Then, for $V_{(a)}$ being the \Plucker coordinates of a maximal flag, it must hold
\be
\label{eq:Plure}
\left(\bigotimes\limits_{a\in A}V_{a}\right)_{L(\omega)}=0\,\quad {\rm if}\quad \omega<\omega_{\rm max}\,.
\ee
Indeed, this is obviously true for the highest-weight vector and therefore is also true for any vector on the $\groupG$-orbit.

The \Plucker relations \eqref{eq:Plure} form an ideal in $\mathbb{C}[V_{(a),i}]$. By the Hilbert basis theorem, one needs only finitely many of them to generate all the rest. The flag manifold can be also identified as all such $V_{(a),i}$ for which \eqref{eq:Plure} hold.
\newline
\newline
\noindent {\it Fused flag} is defined as follows: Consider the embedding of the complete flag manifold  $\groupG/\groupB$ into the product of $\groupG/\groupP_a$:
\be
 \groupG/\groupB\subset \groupG/\groupP_1\times \groupG/\groupP_2\times\ldots \times \groupG/\groupP_{\rank}\,.
\ee
Then a fused flag is a set of maps~\footnote{By slightly abusing notation we identify the map with the corresponding \Plucker coordinates $Q_{(a)}$ that depend on the spectral parameter. In this definition of a fused flag, we do not request that they are Q-vectors of an extended Q-system.} $Q_{(a)}:\Sigma\to \groupG/\groupP_a$, where $\Sigma$ is the space of spectral parameter, such that
\be\label{eq:ffpr}
Q_{(1)}^{[p_1]}\times Q_{(2)}^{[p_2]}\times\ldots\times Q_{(\rank)}^{[p_\rank]}\in \groupG/\groupB\,
\ee
for any Coxeter height function $p$ defined on page~\pageref{Coxhf}. For instance \eqref{eq:ffpr} should hold for an alternating pattern $p_a=0$, where $a$ are the even nodes, $p_a=1$, where $a$ are the odd nodes; but also for \eg increasing patterns like $(p_1,p_2,p_3)=(0,1,2)$ for the $A_3$ case.

\begin{lemma}
\label{thm:lm}
The maps $Q_{(a)}:\Sigma\to \groupG/\groupP_a$ define a fused flag if and only if $(Q_{(a)},Q_{(a')}^\pm)\in G/P_{aa'}$ for all adjacent nodes $a,a'$ of the Dynkin diagram.
\end{lemma}
\begin{proof}
The statement is  proven by induction using Lemma~\ref{thm:ABc}
\end{proof}
\begin{lemma}
\label{thm:ABc}
Let $A,B$ be two non-intersecting sets of Dynkin diagram nodes and $c$ is a node not belonging to $A$ or $B$. Denote by $x$ a point in $\groupG/\groupP_{A}$, by $y$ a point in $\groupG/\groupP_{B}$, and by $z$ a point in $\groupG/\groupP_{c}$. If $(x,z)$ belongs to $\groupG/\groupP_{Ac}$ and $(y,z)$ belongs to $\groupG/\groupP_{Bc}$ then $(x,y,z)$ belongs to $\groupG/\groupP_{ABc}$.
\end{lemma}
\begin{proof}
Because all properties can be viewed as defined \via polynomial equations \eqref{eq:Plure}, it is enough to prove the statement for a dense set of points $x,y,z$. We have $(x,z)=g_1\cdot (x_0,z_0)$ and $(y,z)=g_2 \cdot (y_0,z_0)$, where $x_0,y_0,z_0$ are the points of the standard partial flags (corresponding to the highest-weight vectors) and $g_1,g_2$ are some group elements.  We know that $g_1\cdot z_0=g_2 \cdot z_0$ and hence $g_1^{-1}g_2\in \groupP_{c}$. A dense set of elements of $P_c$ can be represented as $\prod\limits_{a\neq c}e^{c_{a}E_{-\alpha_a}} b$, where $b\in\groupB$ and $c_{a}$ are complex numbers, and an order in the product is chosen such that $\alpha_a$ with $a\in A$ are to the right compared to $\alpha_a$ with $a\in B$. Then, for any $\pi_c$ belonging to this dense set, we can write $\pi_c=\pi_A \pi_B$ for some $\pi_{A}\in P_{A}$, $\pi_B\in P_B$. Hence $(x,y,z)=g(x_0,y_0,z_0)$ for $g=g_1\pi_A=g_2\pi_B^{-1}$.
\end{proof}
\ 
\newline
\newline
It is clear, by the direct pattern recognition, that the projection properties \eqref{eq:pr} are instances of the \Plucker relations \eqref{eq:Plure}. The next statement establishes that all the \Plucker relations are encoded into the Q-system:

\begin{theorem}
\label{thm:ff}
$Q_{(a)}$ -- the Q-vectors of an extended Q-system -- are \Plucker coordinates of a fused flag.
\end{theorem}
\begin{proof}
It is easy to establish using the WKB analysis that \eqref{eq:pr} holds in the case $n_i-n_{i'}=\pm 1$, same sign for all $i,i'$ such that $a_i$ is any even node and $a_{i'}$ is any odd node. This relation implies that $Q_{(a_i)},Q_{(a_{i'})}^\pm$ are \Plucker coordinates of $\groupG/\groupP_{a_ia_{i'}}$. Then use Lemma~\ref{thm:lm}.
\end{proof}
A fused flag shall be called non-degenerate if, for all $a$ and any $k$ and $n_1,
\ldots, n_k$, the \Plucker vectors $Q_{(a)}^{[n_1]}$, $Q_{(a)}^{[n_2]}$, ...\,, $Q_{(a)}^{[n_k]}$ span a vector space of the maximal possible dimension provided the fused flag condition is satisfied. In the explicit physical systems, the non-degeneracy holds for all but a discrete or even a finite set of spectral parameter values. These values of the spectral parameter are related to the inhomogeneities of the spin chain, they are a part of the input information about the system allowing to fix its spectrum. As was discussed on page~\pageref{page:singularities},  we should first consider a domain that avoids these singularities where all relations between Q-functions can be freely used and then approach the singularities using analytic continuation. In a sense, the spectral parameter can be used as a regulator.

An interesting question arising is whether  being a non-degenerate fused flag implies all the other relations between Q-functions. Using the dense set argument of Theorem~\ref{thm:uni}, we can give a positive answer if we can find an algorithm to generate all the $Q$-functions from $Q_{(a),1}$ using the fused flag properties only. We can show that the fused flag condition implies \eqref{eq:QQi} and hence we can reproduce the Weyl-orbit Q-system from the fused flag. Equation \eqref{eq:QQi} allows also computing, although using general position assumptions in contrast to relations used in Theorem~\ref{thm:unie}, the extended Q-system for all cases except for $E_8$. The case of $E_8$ is more problematic because it does not have an irrep with all components being on the Weyl orbit, and we are not aware how to derive the fusion property \eqref{eq:fijk} using only the fused flag properties. Hence, for the $E_8$ case, we are not certain whether each non-degenerate fused flag is an extended Q-system, however we conjecture that it is. 

\subsection{Opers}
\label{sec:opers}
In this subsection we establish a relation between (finite-difference) opers and fused flags for any simply-laced Lie algebra.  For the $\sl_{\rank+1}$ case, the concept of the fused flag was introduced in \cite{Kazakov:2015efa}, and its relation to the oper formalism became clear after \cite{Koroteev:2018jht}, see also \cite{Koroteev:2020mxs} for the $\sl_{\infty}$ case. Below, we shall use the explicit $\sl_3$ example as an illustration, it also serves as a good link to \cite{STS1,STS2} where finite-difference opers are introduced as a generalisation of differential opers.

Definitions of opers vary somewhat across the literature, we shall rely on the ones in \cite{Frenkel:2020iqq}. This paper works on the level of a principal $\groupG$-bundle over $\Sigma$, its reduction to $\groupB$-bundles \etc, but we shall simplify the exposition and write formulae locally and in a certain gauge (equivalently, in a local trivialisation), and so the definitions shall be adapted accordingly.

\paragraph{Generic oper} First we introduce a generic finite-difference oper dubbed $(\groupG,\q)$-oper in \cite{Frenkel:2020iqq}. Two objects  enter its definition: The first one is a finite-difference connection which can be thought of as a $\groupG$-valued function $U(z)$. Informally, it is the Wilson line $U(z)=Pe^{\int_{\spm}^{\q\,\spm}A(z')dz'}$ though $A$ itself need not be defined. The  second object is a $\spm$-dependent complete flag which we shall denote as $\CF(z)\in \groupG/\groupB$. 

An oper is a pair $(U(z),\CF(z))$  considered modulo gauge transformations which satisfies the following criterium: In a gauge where $\CF(z)$ is a standard flag (corresponding to the highest-weight vector in the sense of \Plucker coordinates) at each point $z$, the connection $U(z)$ should be an element of the Bruhat cell $B s B$~\footnote{For a definition and basic properties of the Bruhat decomposition, see \eg \textsection{23.4} of \cite{Fulton-Harris}.}, where $s$ is a representative of a Coxeter element of the Weyl group. Explicitly, we should be able to represent $U(z)$ as
\be
\label{eq:operc1}
U(z)= n(z)\,\prod_{\alpha\in \Delta}s_{\alpha}\, b(z)\,,
\ee
where $b(z)\in \groupB$, $n(z)\in\groupN$, $s_{\alpha}$ are representatives of Weyl reflections \wrt to simple roots, \eg the ones defined by \eqref{eq:sst}, and the order in which the product is taken corresponds to the choice of a Coxeter element.
\newline
\newline
As an example, consider  the $\sl_3$ case and recast the conjugate Baxter equation \eqref{eq:KLWZ2}, explicitly $Q^{[-3/2]}-T_{(1)}^+Q^{[1/2]}+T_{(2)}^{[2]} Q^{[5/2]}-Q^{[9/2]}=0$, as a matrix linear equation
\be\label{eq:LB}
\Phi^{++}=U\,\Phi\,,\quad U=
\begin{pmatrix}
T_{(2)}^{[3/2]} & -T_{(1)}^{[1/2]} & 1
\\
1 & 0 & 0
\\
0 & 1 & 0
\end{pmatrix}\,,
\ee
with the matrix of solutions being 
\be
\label{eq:solmat}
\Phi=\begin{pmatrix}
 Q^{1[2]} &  Q^{2[2]}  &  Q^{3[2]}
\\
 Q^{1} & Q^{2}  &  Q^{3}
\\
Q^{1[-2]} &  Q^{2[-2]}  & Q^{3[-2]}
\end{pmatrix}\,.
\ee

Elements of $\groupB$ can be represented as $\left(\begin{smallmatrix} * & * &  * \\ 0 & * & * \\ 0 & 0 & *\end{smallmatrix}
\right)$ with non-zero diagonal entries, elements of $\groupN$ are the same type matrices with $*=1$ on the diagonal. Elements of $\groupG/\groupB$ are then $3\times 3$ matrices modulo right multiplication of $\groupB$. We can interpret them as an ordered set of vectors $(V,V',V'')$  which must obey $V\wedge V'\wedge V''=1$ and store physical information in $V$ and $V\wedge V'$ that are considered up to a normalisation and define a complete flag $\emptyset\subset \mathbb{C}\subset \mathbb{C}^2\subset \mathbb{C}^3$.

The connection $U$ of \eqref{eq:LB} can be factorised as $U=
\left(\begin{smallmatrix}
1 & T_{(2)}^{[3/2]} & -T_{(1)}^{[1/2]} 
\\
0 & 1 & 0
\\
0 & 0 & 1
\end{smallmatrix}
\right)
\left(
\begin{smallmatrix}
0 & 0 & 1
\\
1 & 0 & 0
\\
0 & 1 & 0
\end{smallmatrix}
\right)
$. The second matrix in the product is a Coxeter element and so this factorisation is of the form \eqref{eq:operc1}.  Thus the pair $(U(z),\CF(z))$, where $\CF(z)$ is the identity matrix, forms an oper.

\paragraph{Relation to fused flag} Let us first understand a geometric interpretation of the condition \eqref{eq:operc1}. Let explicitly the product over simple roots be
\be
\label{eq:420}
\prod_{\alpha\in \Delta}s_{\alpha}=s_{\alpha_{a_1}}\ldots s_{\alpha_{a_{\rank}}}\,,
\ee
where $a_i$ is a permutation of $(1,2,\ldots,\rank)$. For a given $k$, define sets $A_k=(a_1,\ldots,a_k)$, $B_k=(a_{k+1},\ldots,a_{\rank})$. Define correspondingly $s_{A_k}=s_{\alpha_{a_1}}\ldots s_{\alpha_{a_k}}$ and $s_{B_k}=s_{\alpha_{a_{k+1}}}\ldots s_{\alpha_{a_\rank}}$. Then represent the standard complete flag as $(x_k,y_k)$, where $x_k$ is the standard partial flag of $\groupG/\groupP_{A_k}$ and $y_k$ is the standard partial flag of $\groupG/\groupP_{B_k}$. We note that $U(z)x_k=n(z)s_{A_k} x_k$ since $x_k$ is invariant under action of $\groupB$ and $s_{B_k}$. On the other hand, $n(z) s_{A_k}y_k=y_k$. Hence we conclude that
\be
\label{eq:operc2}
(U(z)x_k,y_k) \in \groupG/\groupB\,,\quad{\rm for\ } k=0,1,\ldots, \rank\,.
\ee
That is we can parallel-transport using $U$ only a special subset of \Plucker coordinates, $V_{(a_1)},\ldots, V_{(a_k)}$, and still remain in the maximal flag.

The argumentation to derive \eqref{eq:operc2} from \eqref{eq:operc1} can be reversed: If \eqref{eq:operc2} holds and all the points $(U(z)x_k,y_k)$ are distinct then $U(z)$ is of the form \eqref{eq:operc1}. 
\newline
\newline
For our $\sl_3$ example, let ${\bf e}_i$ be the standard basis vectors. Then $({\bf e}_1, {\bf e}_1\wedge {\bf e}_2)$ is the standard flag. We check that $(U{\bf e}_1, {\bf e}_1\wedge {\bf e}_2)$ is a flag because $U{\bf e}_1=T^{[3/2]}_{(2)}{\bf e}_1+{\bf e}_2$.  To see that this is a non-trivial property, we remark that $({\bf e}_1, U {\bf e}_1\wedge U{\bf e}_2)$ is not a flag since $U^{-1}{\bf e}_1={\bf e}_3$.
\newline
\newline
 In a gauge where the flag $\CF(z)$ is standard, all information about the oper is concentrated in the connection $U(z)$. Let us now perform a gauge transformation~\footnote{We can potentially spoil some nice analytic structure in this way but we do not loose information.} to make the connection trivial $U(z)=\Id$. In this gauge, all information is transferred to the flag $\CF(z)$. 

Remarkably, the oper condition in this gauge can be rewritten as the one of a fused flag. Indeed, let $V_{(a)}$ be the \Plucker coordinates of $\CF(z)$ in this new gauge. Parallel transport with respect to the trivial connection $U$ does not change them: $V_{(a)}^{\rm pt}(\q z)=V_{(a)}(z)$. On the other hand, $V_{(a)}(z)$ as functions of $z$ are non-trivial. Property \eqref{eq:operc2}, together with the obvious $(x_k,y_k)\in \groupG/\groupB$, becomes in the new gauge
\be
\label{eq:operc3}
(V_{(a_1)}^{-},\ldots,V_{(a_k)}^{-},V_{(a_{k+1})}^{\pm},\ldots,V_{(a_{\rank})}^{\pm}) \in \groupG/\groupB\,,\quad{\rm for\ } k=0,1,\ldots, \rank\,.
\ee
Now, recall that one can assign the Coxeter height function $p$ to the Coxeter element \eqref{eq:420}, see Section~\ref{sec:33}. Using this function, identify 
\be
\label{eq:1203}
V_{(a)}=Q_{(a)}^{[p_a]}\,.
\ee
Condition \eqref{eq:operc3} ensures that $Q_{(a)}$ satisfy conditions of Lemma~\ref{thm:lm} and hence define a fused flag.

We hence see that a non-degenerate fused flag is an oper realised in a particular gauge. A fused flag can be also gauged \cite{Kazakov:2015efa}. A gauged non-degenerate fused flag considered modulo gauge transformations is hence an equivalent of a finite-difference oper. There is however an interesting caveat. The definition of an oper involves a choice of the Coxeter element. In contrast, the fused flag does not require to make this choice. The choice is being made only when we link the fused flag and the oper.
\newline
\newline
In the $\sl_3$ example, the fused flag gauge is realised for $\CF=\Phi^{-1}$, where $\Phi$ is given by \eqref{eq:solmat}. To compute $\Phi^{-1}$ in a suggestive way, use the relation between Q-functions and their Hodge duals \eqref{eq:Hodge}, explicitly $Q^a=\epsilon^{abc}  Q_b^+ Q_c^-$, and then invoke the projection relation $Q^a  Q_{a}^\pm=0$, the quantisation condition $Q^a Q_{a}^{[\pm 3]}= 1$, and the bilinear formula for transfer matrices \eqref{eq:TQQ}. The result of inversion is 
\be
\CF=\Phi^{-1}=
\begin{pmatrix}
 Q_1^{-} &  Q_1^{[-3]}  & Q_1^{[-5]}
\\
 Q_2^{-} &  Q_2^{[-3]}  &  Q_2^{[-5]}
\\
 Q_3^{-} & Q_3^{[-3]}  &  Q_3^{[-5]}
\end{pmatrix}
\begin{pmatrix}
 1 & -T_{(2)}^{[-1/2]} & T_{(1)}^{[-3/2]} 
 \\
 0 & 1 & -T_{(2)}^{[-5/2]} 
 \\
 0 & 0 & 1
\end{pmatrix}\,.
\ee
We see that  \eqref{eq:1203}  becomes $V_{(1)}=Q_{(1)}^-$ and $V_{(2)}=Q_{(1)}^-\wedge Q_{(1)}^{[-3]}=Q_{(2)}^{[-2]}$ and we reproduce the construction discussed in the introduction, \cf \eqref{eq:fus}.

\paragraph{ODE/IM perspective} ODE/IM provides us with an interesting connection between the extended Q-systems and opers. Using the WKB asymptotics \eqref{eq:WKB0} we can expect that the following Wilson line in the $x$-plane~\footnote{Not to confuse with informal Wilson line in the interpretation of $U(z)$, these are entirely different objects!} connects $\Psi$-function at the origin and the infinity
\be
\label{eq:gauge1}
Q_{(a)}^{[p_a]}(\spm)=\spm^{-\frac{\rho^{\vee}}{h\,M}}\left(\lim_{x_0\to \infty} P e^{\int_{x_0}^0 A(x',\spm)dx'}x_0^{-M\rho^{\vee}}e^{-\Lambda\frac{x_{0}^{M+1}}{M+1}}\right)\mathsf{U}_{(a)}^{[p_a]}\,.
\ee
Here, we remind, $\mathsf{U}_{(a)}^{[p_a]}$ are the eigenvectors of $\Lambda_{L(\omega_a)}$ with the eigenvalue $\gamma^{p_a/2}\mu_a$, for $a=1,\ldots,\rank$. As explained in Section~\ref{sec:33}, they are the highest-weight vectors in a basis where $\Lambda$ belongs to a Cartan subalgebra $\mathfrak{h}'$ and for a specific choice of simple roots. Hence, in this basis, they are the \Plucker coordinates of the standard flag.

Furthermore, using \eqref{eq:Symanzik} and \eqref{eq:easy}, we can compute that
\be
\label{eq:gauge2}
Q_{(a)}^{[p_a]}(\q\,\spm)=\spm^{-\frac{\rho^{\vee}}{h\,M}}\left(\lim\limits_{x_0\to \infty} P e^{\int_{x_0}^0 A(x',\spm)dx'}x_0^{-M\rho^{\vee}}e^{-\Lambda\frac{x_{0}^{M+1}}{M+1}}\right)\gamma^{\rho^{\vee}}\mathsf{U}_{(a)}^{[p_a]}\,.
\ee
Recall that $\gamma^{\rho^{\vee}}$ is a Coxeter element in the same basis where $\mathsf{U}_{(a)}^{[p_a]}$ define the standard flag. We hence can view this Wilson line as a gauge transformation from a fused flag gauge (where connection $U(z)$ is trivial) to a standard flag gauge (where $U(z)$ is a Coxeter element).

This very plausible explanation has however a drawback. The limit $x_0\to\infty$ is ill-defined, in particular due to Stokes phenomena. At this moment, the best we can do is to declare $\left(\lim\limits_{x_0\to \infty} P e^{\int_{x_0}^0 A(x',\spm)dx'}x_0^{-M\rho^{\vee}}e^{-\Lambda\frac{x_{0}^{M+1}}{M+1}}\right)$ to be such a group element depending on $\spm$ that \eqref{eq:gauge1} and \eqref{eq:gauge2} hold. It would be interesting to provide an intrinsic self-consistent definition of this Wilson line. 

\paragraph{Miura oper}
A generic oper can be further dressed with additional structures or constrained by additional requirements giving rise to opers of specific type. Probably the most important example is a Miura oper which is an oper $(U(z),\CF(z))$ supplemented with a solution $\CF_-$ of the finite difference equation $\CF_-^{++}=U\CF_-$, where $\CF_-(z)\in \groupG/\groupB_-$ with $\groupB_-$ being the Borel subgroup opposite to $\groupB$.

To understand the significance of $\CF_-$, consider the fused flag gauge. Recall that the relations of the extended Q-system are invariant under action of $\groupG$~\footnote{More accurately, they are invariant under action of $\groupG$-functions that  periodically depend on $z$ but this is of no relevance for the argument because if a function enters a relation, only its even integer shifts can be featured in the same relation. Likewise, we ignore potential periodic dependence of $\CF_-$ on $z$.}, and hence there is no distinguished way to select components of vectors $Q_{(a)}$. However, $\CF_-(z)$ is constant in this gauge and it is not invariant under $\groupG$ breaking the symmetry to (a conjugate of) $\groupB_-$. In physical terminology, $\CF_-\in\groupG/\groupB_-$ is an order parameter (though the symmetry is not broken spontaneously but by our choice of $\CF_-$). Now, consider Bruhat decomposition $\CF_-=b_- s b_-'/\groupB_-$, where $b_-,b_-'\in\groupB_-$ and $s$ is a Weyl group representative, and perform the global symmetry transformation: $Q_{(a)}\to\tilde Q_{(a)}=s^{-1}b_-^{-1} Q_{(a)}$. After this transformation, $\CF_-$ becomes the standard flag which is preserved by action of $\groupB_-$. Whereas $\tilde Q_{(a)}$ overall transforms under this action, its highest-weight component $\tilde Q_{(a),1}$ is the only component that does not transform save for a constant normalisation and in this sense it is distinguished.  Hence, the extra data $\CF_-$ is a way to declare which components of vectors $Q_{(a)}$ should be called $Q_{(a),1}$ -- Q-functions on the Dynkin diagram -- and used, for instance, in conventional Bethe equations and computation of quantum eigenvalues. The flag manifold $\groupG/\groupB_-$ is the configuration space for all possible such choices, we can identify it with the one in \cite{Mukhin2003,MV05}.\label{pg:miura}

Let us also give a different perspective on the Miura oper. Consider a gauge in which $\CF(z)$ is the standard flag and perform a mixed Bruhat decomposition $\CF_-(z)=b(z)s b_-(z)/\groupB_-$, note that $b(z)\in\groupB$. We expect $\CF_-$ to be in generic position for almost all values of the spectral parameter, and then $s=\Id$ for this domain~\footnote{by contrast, the above-discussed Bruhat decomposition of $\CF_-$ in the fused flag gauge does not depend on $z$ and can have any $s$.}. Perform now the gauge transformation by $b^{-1}(z)$. Then, in this new Miura gauge, $\CF(z)$ still remains the standard flag but $\CF_-(z)$ is also the standard flag. Hence all the information about the Miura oper is concentrated in the connection $U$, this connection alone or rather a finite-difference operator $\mathcal{D}^2-U$ is also encountered as a definition of Miura oper in literature.
\newline
\newline
For the $\sl_3$ example, selecting concrete solutions of Baxter equation for the role of $Q^a$, $a=1,2,3$, in \eqref{eq:solmat}, modulo right action of lower-triangular matrices, is what defines $\CF_-$ in the gauge \eqref{eq:LB}. The mentioned right action can modify $Q^3$ only by an inessential overall normalisation but it introduces ambiguity in $Q^a$ for  $a=2,3$, by adding $Q^b$ with $b>a$. Correspondingly the Q-functions of the fused flag nailed by $\CF_-$ are $Q_{12}=Q^3$ and $Q_1=W(Q^2,Q^3)$.

If one takes $\CF_-=\Phi$ and choose $b(z)\in\groupN$ in the decomposition $\CF_-=b(z)b_-(z)$ then $U$ in the Miura gauge becomes
\be
U_{\rm Miura}=\left(\begin{matrix}\tilde\Lambda_1^{[-1/2]} & 0 & 0
\\
1 & \tilde\Lambda_2^{[-1/2]} & 0
\\
0 & 1 & \tilde\Lambda_3^{[-1/2]} \end{matrix}\right)\,,
\ee
where $\tilde\Lambda_a$ and $\Lambda_a$ of \eqref{eq:Miura20} are related as 
\be
(\tilde\Lambda_3-\CD^{-2})(\tilde\Lambda_2-\CD^{-2})(\tilde\Lambda_1-\CD^{-2})=(1-\Lambda_3\CD^{-2})(1-\Lambda_2\CD^{-2})(1-\Lambda_1\CD^{-2})\,;
\ee
They are natural quantities for the factorised form of the conjugate Baxter equation $Q^{a[-\frac{3}{2}]}(\tilde\Lambda_3-\overleftarrow\CD^{-2})(\tilde\Lambda_2-\overleftarrow\CD^{-2})(\tilde\Lambda_1-\overleftarrow\CD^{-2})=0$.
\section{Applications}\label{sec:Applications}

\subsection{Solving Hirota equation}\label{sec:SolvingHirota}
In this subsection we provide a solution to the so-called Y- and T-systems in terms of Q-functions. These systems were considered for all simple Lie algebras, see \cite{Kuniba:2010ir} for a review and references therein, we focus on the simply-laced cases only.

Y-systems appears in the context of thermodynamic Bethe Ansatz. It is a collection of $Y_{a,s}$, where $a$ run through the nodes of the Dynkin diagram and, depending on the model, $s\in\mathbb{Z}$ or $s\in\mathbb{Z}_{\geq 0}$. For the simply-laced case, these functions satisfy the following condition
\be
    \frac{Y^+_{a,s}Y^-_{a,s}}{Y_{a,s+1}Y_{a,s-1}} = \frac{\prod\limits_{b,C_{ab}=-1}(1+Y_{b,s})}{(1+Y_{a,s-1})(1+Y_{a,s+1})}\,.
\ee
Upon substitution $Y_{a,s}=\frac{1}{T_{a,s-1}T_{a,s+1}}\prod\limits_{b,C_{ab}=-1}T_{b,s}$, one obtains the Hirota equation (T-system)
\be\label{eq:TT}
T_{a,s}^+T_{a,s}^--T_{a,s+1}T_{a,s-1}=\prod_{b,C_{ab}=-1}T_{b,s}\,.
\ee
Apart from appearing in TBA, $T$-functions that satisfy \eqref{eq:TT} have also an interpretation as transfer matrices with auxiliary space being a Kirillov-Reshetikhin module.

Similarity in structure of \eqref{eq:TT} and \eqref{eq:QQi} is very suggestive. Using the $\sl_{\rank+1}$ solution \eqref{eq:TQQ} as a further insight, it is then not difficult to guess the following ansatz for T-functions
\be\label{eq:TfromQ}
T_{a,s}=\langle Q_{(a)}^{[s]},\tilde Q_{(a^*)}^{[-s]}\rangle\,,
\ee
where $Q$ and $\tilde Q$ are two \apriori different Q-systems.

Here is a proof that this ansatz indeed solves \eqref{eq:TT}:
\be
&&T_{a,s}^+T_{a,s}^--T_{a,s+1}T_{a,s-1} =
\nonumber
\\
&&\left\langle \left(Q_{(a)}^+\wedge Q_{(a)}^-\right)^{[s]},\left(\tilde Q_{(a^*)}^+\wedge \tilde Q_{(a^*)}^-\right)^{[-s]} \right\rangle 
=
\left\langle \left(Q_{(a)}^+\wedge Q_{(a)}^-\right)_{L(\omega_{\rm max})}^{[s]},\left(\tilde Q_{(a^*)}^+\wedge \tilde Q_{(a^*)}^-\right)_{L(\omega_{\rm max})}^{[-s]} \right\rangle
\nonumber\\
&&=
\nonumber\\
&&\left\langle \left(\bigotimes\limits_{b,C_{ab}=-1} \hspace{-.8em}Q_{(b)}\right)_{L(\omega_{\rm max})}^{[s]},\left(\bigotimes\limits_{b,C_{ab}=-1}  \hspace{-.8em}\tilde Q_{(b^*)}\right)_{L(\omega_{\rm max})}^{[-s]} \right\rangle
=\left\langle \left(\bigotimes\limits_{b,C_{ab}=-1} \!\!\Q_{(b)}\right)^{[s]},\left(\bigotimes\limits_{b,C_{ab}=-1} \!\!\tilde Q_{(b^*)}\right)^{[-s]} \right\rangle
\nonumber\\
&&= \prod\limits_{b,C_{ab}=-1} T_{b,s}\,,
\ee
where we used \eqref{eq:QQi} and, notably, the following projection relations of the Q-functions
\begin{align}
\label{eq:Q2}
\left(Q_{(a)}^+ \wedge Q_{(a)}^-\right)_{L(\omega)} &=0\,,
&
\left(\prod_{b,C_{ab}=-1}Q_{(b)}\right)_{L(\omega)} &=0\,,
& 
\text{for all }\ \ \omega &<\omega_{\rm max}=\sum_{b,C_{ab}=-1} \omega_b\,.
\end{align}

There are cases when the T-system has a boundary. For instance, one has $s\geq 0$ in the transfer matrix interpretation and moreover one fixes $T_{a,0}=1$ since these functions have meaning of the transfer matrices in the trivial representation. In addition, Hirota equation should work for $s=0$ if one sets $T_{a,-1}=0$. 

These features can be reproduced if we identify $Q$ and $\tilde Q$. After slight redefinitions, one sets
\be
\label{eq:Tas}
T_{a,s}=\langle Q_{(a)}^{[s+\frac \Cox2]}, Q_{(a^*)}^{[-s-\frac \Cox2]}\rangle\,.
\ee
This ansatz solves \eqref{eq:TT} and it has the following additional properties: $T_{a,0}=1$ which is the quantisation relation \eqref{eq:qrh}, and moreover $T_{a,s}=0$ for $s=-1,-2,\ldots 1-\Cox$, this is the projection relation \eqref{eq:prh}.

\subsection{Character solution}
Choose  an element of the Cartan algebra $H$ and consider the ansatz
\be\label{eq:charsol}
Q_{(a)}=\spm^{H}\Ac_{(a)}\,,
\ee
where $\Ac_{(a)}$ are vectors that do not depend on the spectral parameter. For this ansatz $Q^{[2]}_{(a)}=\q^{H}Q_{(a)}$ and hence all equations on Q-functions reduce to polynomial equations on $\Ac_{(a)}$. A good parameterisation for $H$ is $H=\sum_{i=1}^{\rank'} H_i \log_\q x_i$ in which case the coefficients of these polynomial equations are Laurent polynomials in $x_i$ with integer coefficients. It is pertinent to choose $H_i$ as generators in the orthogonal basis, see \eg \cite{Feger:2019tvk}. For $\sl_{\rank+1}$, this is actually the basis of $\gl_{\rank+1}$ meaning that $\rank'=\rank+1$ and that the $Q$-vector should not be sensible to the shifts $H_i\to H_i+C$ which is achieved by setting $\prod_{i}x_i=1$. For the $\so_{2\rank}$ case, the orthogonal basis is explicitly described in Section~\ref{sec:DAlgebras}.

A solution of equations for $\Ac_{(a)}$  exists always as we can conclude based on the following two facts: First, Theorems~\ref{thm:uni}~and~\ref{thm:unie} ensure that the extended Q-system exists for any $Q_{(a),1}$ and hence for $Q_{(a),1}=z^{\omega_a(H)}\Ac_{(a),1}$. Second, all the QQ-relations have multiplicative nature, cf.~\eqref{eq:fr}, implying that  the analytic dependence of other Q-functions on $z$ can be maintained in the form \eqref{eq:charsol} if we start from $Q_{(a),1}=z^{\omega_a(H)}\Ac_{(a),1}$. Furthermore, we see that the solution is unique if we fix the values of $\Ac_{(a),1}$. Indeed, while a solution is defined  modulo symmetry $Q_{(a)}\to g\, Q_{(a)}$, the only group elements commuting with $z^H$ are elements of the maximal torus whose action amounts in rescaling of $\Ac_{(a),1}$ assumed to be fixed.

For any $\groupG$-invariant combination $F(Q)$ of Q-functions one has $F^{[2]}=q^{H} F=F$. Hence all such combinations are independent of $z$. Furthermore, for any group element $g$, $F(Q)=F(gQ)$. On the other hand, $g Q = z^{g H g^{-1}} g \Ac_{(a)}$ can be interpreted as the solution \eqref{eq:charsol} of the Q-system with $H\to g H g^{-1}$. We hence conclude that combinations $F(Q)$ are class functions of $q^H$ considered as a group element. 

In particular, $T_{a,s}$ computed by \eqref{eq:Tas} are class functions. Since they are $z$-independent and satisfy \eqref{eq:TT} they should be characters in the corresponding representations. While these are irreps in the case of $\sl_{\rank+1}$, these representations are typically reducible for the case of other Lie algebras, the reason is that they are actually irreps of the relevant quantum algebra.

We build the explicit character solution of the extended Q-system for $\so_{2\rank}$ series in Section~\ref{sec:charT}. The explicit solution for $\sl_{\rank+1}$ is given for instance by (3.9) of \cite{Kazakov:2015efa}. 

\subsection{Analytic Bethe Ansatz}\label{sec:AnalyticBetheAnsatz}
Until now, we mostly avoided discussing the explicit analytic properties of  Q-functions. Specifying them is precisely what defines the physical model we are dealing with. In this section we propose the analytic structures that are supposed to describe rational, trigonometric, and elliptic spin chains. The story is a fairly straightforward generalisation of what was done for the $A_n$ case in \cite{Krichever:1996qd}~\footnote{Following the seminal work \cite{Reshetikhin1983}, the name Analytic Bethe Ansatz is often associated to ans{\"a}tze for T-functions. We extend this terminology imposing analytic requirements directly on Q-functions.}$^,$\footnote{For $A_n$ case, one can alternatively impose analytic requirements using supersymmetrisation and working with the distinguished Q-functions only, as was realised for rational \cite{Marboe:2016yyn} and trigonometric cases \cite{Nepomechie:2020ixi}. We do not attempt to generalise this approach here.}, our parameterisation builds upon (3.62) of \cite{Ryan:2018fyo} though shift conventions are not the same.

To give a uniform presentation, we shall use the additive spectral parameter $\spa$ and agree to relate it to $\spm$ by $e^{2\pi\,u}=z$, correspondingly $e^{2\pi\,\hbar}=\q$.

The spectrum of spin chains should be described by Bethe equations
\be
\label{eq:NBAE1}
\prod_{\ell=1}^L \frac{\phi(u_{a,k}-\theta_{\ell}+\frac{\hbar}{2}m_a^{\ell})}{\phi(u_{a,k}-\theta_{\ell}-\frac{\hbar}{2}m_a^{\ell})}=
-\prod_{b}e^{2\pi\,\hbar\,C_{ab}h_b}\prod_{k'=1}^{M_b} \frac{\phi(u_{a,k}-u_{b,k'}+\frac{\hbar}{2}A_{ab})}{\phi(u_{a,k}-u_{b,k'}-\frac{\hbar}{2}A_{ab})}\,,
\ee
where $\phi(u)=u$ for the rational case, $\phi(u)=\sinh(2\pi\,u)$ for the trigonometric case, and $\partial_{u}^2\ln\phi(u)=-\wp(u)$ for the elliptic case. For the trigonometric case, it is assumed that $\hbar$ is not rational and $\q$ is not a root of unity. For the elliptic case $\hbar$ is not commensurate with the periods $1,\tau$.

For the rational and the trigonometric cases, the physical meaning of the other parameters in \eqref{eq:NBAE1} are: $[m_1^{\ell},\ldots,m_{\rank}^{\ell}]$ are Dynkin labels of the representation assigned to the $\ell$'th node (it is a quantum algebra irrep but generically it might be reducible as a representation of the Lie algebra $\algg$);  $h_b:=\omega_b(H)$~\footnote{$H$ is the same as in \eqref{eq:charsol}} specify twisted boundary conditions of the spin chain; and $\theta_{\ell}$ are inhomogeneities. For the elliptic case, there are certain restrictions on admissible values of $M_b$ since $\phi$ are not periodic functions on the torus, also physical models in the elliptic case were not built to the same level of generality as in the case of rational and trigonometric systems.

From the character solution \eqref{eq:charsol}, Bethe equations \eqref{eq:NBAE}, and using experience with $A_r$ system, it is natural to guess the following ansatz for Q-functions
\be\label{eq:ABA}
\mbox{$\displaystyle Q_{(a),i}(\spa)=N_{(a),i}\times \Ac_{(a),i}\prod_{j=1}^{\rank'}x_j^{\frac{2\pi\,u}{\hbar}\gamma_i(H_j) }\times \dressing_{a}(\spa)\times q_{(a),i}(\spa)$\,.}
\ee
It is split into four factors. The first factor is a number whose sole purpose is to adjust the normalisation such that there is an equality sign in \eqref{eq:QQi}. It has no physical importance. The second factor is the character solution \eqref{eq:charsol} $e^{2\pi u H}\Ac_{(a)}$, we just wrote it in components. The aim of the third  factor is to reproduce the \lhs (the source term) in Bethe equations \eqref{eq:NBAE}. We shall study it in a moment. Finally, the last factor is  ``polynomial''
\be
q_{(a),i}=\prod_{k=1}^{M_{(a),i}}\phi(\spa-u_{(a),i,k})\,
\ee
which, in the trigonometric case, was featured in (5.2) of Baxter's original work \cite{Baxter:1971cs} about the six-vertex model (written there as a polynomial in the multiplicative spectral parameter). Zeros of $q_{(a),1}$, $u_{a,k}\equiv u_{(a),1,k}$ satisfy conventional Bethe equations \eqref{eq:NBAE} whose explicit form is \eqref{eq:NBAE1}, $M_a=M_{(a),1}$. Zeros of $q_{(a),\sigma(1)}$, where $\sigma$ is an element of the Weyl group, satisfy Weyl-dual Bethe equations.
\newline
\newline
The dressing factor $\sigma_a$ does not depend on $i$ which reflects that all the Weyl-dual Bethe equations have structurally the same source term. By recalling that Bethe equations in terms of Q-functions are written as \eqref{eq:BAEOW} and requiring that the \lhs of \eqref{eq:NBAE1} is reproduced from $\sigma_a$ when we substitute the Ansatz \eqref{eq:ABA} into \eqref{eq:BAEOW}, one gets the following equation on $\sigma$
\be\label{eq:dreq}
\sum_{b}[C_{ab}]_{\CD} \log \sigma_{b} = - \sum_{\ell=1}^{L} [m_{a}^{\ell}]_{\CD} \log \phi(u-\theta_{\ell})\,,
\ee
where we used  a notation for ``$\CD$-deformed''  numbers $[n]_\CD:=\CD^{n-1}+\CD^{n-3}+\ldots +\CD^{1-n}$ for $n>0$, $[n]_\CD:=-[-n]_\CD$ for $n<0$. Equation \eqref{eq:dreq} is formally solved by 
\be
\log\sigma_a=-\sum_b\sum_{\ell=1}^{L} \left([C]_\CD^{-1}\right)_{ab} [m_b^{\ell}]_\CD\log \phi(u-\theta_{\ell})\,,
\ee
one can also provide a precise meaning for this formal solution, \cf \cite{Volin:2010cq}.

We see that the dressing factors are, in a sense, the inverse deformed Cartan matrices describing interaction between Bethe roots and source terms which is reminiscent of integrable relativistic integrable models where the dressing factors are the same inverse deformed Cartan matrices describing interactions between particles \cite{1998solv.int.10007Z}.
\newline
\newline
The twist factor does not have nice periodicity properties to be considered as a function on a cylinder or a torus. To improve on this issue, we can perform the gauge transformation
\begin{align}
Q_{(a)}^{\rm gt} &= e^{-2\pi\,u\,H}Q_{(a)}\,, & U^{\rm gt}= e^{-2\pi\,(u+\hbar)\,H}e^{2\pi\,u\,H}=e^{-2\pi \hbar\, H}\,.
\end{align}
In the new gauge, Q-functions do not have the $u$-dependent part of the twist factor. We get a non-trivial constant finite-difference connection $U$ instead.

The dressing factor is also not a particularly pleasant function of the spectral parameter. For a torus, this would be an object with everywhere dense set of poles, and so one must work on a universal cover to be able of defining it. However we should recall that Q-functions are projective coordinates meaning that the dressing factor nearly cancels out from the equations and the piece remaining due to non-locality is well-defined without resorting to the universal cover. Consider for instance the QQ-relation \eqref{eq:QQi}. Explicitly in terms of $q$ and in the case $H=0$ it becomes
\be
\left( q_{(a)}^+ \wedge q_{(a)}^- \right)_{L(\omega_{\rm max})} &=& J_a\times  \left(\bigotimes_{b,C_{ab}=-1} q_{(b)}\right)_{L(\omega_{\rm max})}\,,
\nonumber\\
\label{eq:514}
J_a=\prod_b \sigma_b^{-[C_{ab}]_D} &=& \prod_{\ell=1}^{L}\prod_{k=-\frac{m_b-1}{2}}^{\frac{m_b-1}{2}}\phi(u-\theta_{\ell}+k\,\hbar)\,.
\ee
Note that the dressing factor determines the position of poles and zeros where the fused flag fails to be non-degenerate. Indeed, in the above example, by remembering the relation $\left( q_{(a)}^+ \wedge q_{(a)}^- \right)_{L(\omega<\omega_{\rm max})}=0$ which holds always, we see that vectors $q_{(a)}^+$ and $q_{(a)}^-$ are collinear at $u=\theta_{\ell}+k\,\hbar$ with $k$ being in the specified by \eqref{eq:514} range. These are regular singularities of a $(G,q)$-oper in \cite{Koroteev:2018jht, Frenkel:2020iqq}.

Curiously, while the extended Q-system must be a non-degenerate fused flag almost everywhere, the prescription of degeneration points is an essential ingredient for selecting physically relevant solutions.

\paragraph{Completeness and faithfulness conjectures} Based on the results established for the $A_r$ spin chains in the vector representation \cite{2013arXiv1303.1578M,Chernyak:2020lgw} we conjecture that the extended Q-system is always the right object to correctly encode the spectrum of the corresponding integrable model, at least for spin chain nodes in a fundamental representation of the quantum algebra. In contrast, we know already for the $A_r$ case that Bethe equations \eqref{eq:NBAE} have shortcomings, see for instance discussion in \cite{Chernyak:2020lgw}. 

The completeness conjecture is that the algebraic number of extended Q-systems that verify analytic structure \eqref{eq:ABA} is equal to the dimension of the corresponding weight subspace of the Hilbert space. The number of Q-systems should be computed modulo residual symmetries. For generic twist $H$, only action of Cartan subalgebra $\mathfrak{h}\in\algg$ is a symmetry (it only changes normalisations and hence is inessential), while for zero twist $H=0$, rotation by any element of $\algg$ is a symmetry. The weight subspace is defined as a space of highest-weight vectors with respect to action of the residual symmetry~\footnote{If the symmetry is only $\mathfrak{h}$, then it is simply a space of vectors of given weight.}, of weight given by the Dynkin labels $[d_1,\ldots, d_\rank]$ computed as
\be
d_a=\sum_{\ell=1}^L m_{a}^{\ell} - \sum_{b=1}^{\rank }C_{ab}M_b\,.
\ee
Beyond the $A_r$ case, we performed a verification of the completeness conjecture by an explicit computation for various examples of $\so_8$ and $\so_{10}$ rational spin chains, this result will be published separately \cite{EV-in-preparation}.

The faithfulness conjecture is that the Bethe algebra restricted to the weight subspace is isomorphic to the algebra of  Q-functions of analytic form \eqref{eq:ABA} and satisfying equations of the extended Q-system.

In the rational and trigonometric cases, the algebra of Q-functions can be considered as a polynomial quotient ring whose variables are coefficients of Baxter polynomials.  At least in these cases, isomorphism to the Bethe algebra would imply that the spectrum of the Bethe algebra restricted to the weight subspace is simple.

\section{$D_\rank$ series}\label{sec:DAlgebras}
In previous sections we discussed properties of extended Q-systems including their geometric interpretation as fused flags, how they are used to solve Hirota equations, and their character solution. In this section we demonstrate these properties on the specific example of $\so_{2\rank}$. 
\subsection{Notations}
We use the same notation as in Section~\ref{sec:sl4so6} for vectors, spinors and co-spinors
\begin{align}
Q_{(1),\vecf{i}} &= \qV_i\,, &
Q_{(\rank-1),\alpha} &= \qS_\alpha\,, &
Q_{(\rank),\dot \alpha} &= \qC_{\dot \alpha}\,.
\end{align}
The Q-vectors in the remaining fundamental representations are rank-$a$ antisymmetric tensors $V_{(a)}$, $a\leq \rank-2$, whose components shall be written using the multi-index notation $\qV_{I} = \qV_{i_1,\dots,i_a}$. It will be also convenient to consider $V_{(a)}$ for $a=\rank-1,\rank$, these tensors correspond to non-fundamental representations.

When labelling the components of vectors and spinors we shall have the orthogonal basis of $\so_{2\rank}$ in mind. It is spanned by $\orthb_{1},\ldots,\orthb_{r}$ with the inner product $\langle \orthb_{a},\orthb_{b}\rangle = \delta_{ab}$. Simple roots of $\so_{2\rank}$ are expressed in this basis as
\be
    \alpha_a = \orthb_a -\orthb_{a+1}\,, \quad a\leq \rank-1\,, \quad
    \quad \alpha_{r} = \orthb_{r-1}+\orthb_{r}\,,
\ee
and $\langle \alpha_a,\alpha_b\rangle = C_{ab}$ is the Cartan matrix of $\so_{2\rank}$. A weight expressed in terms of Dynkin labels can be converted into the orthogonal basis using
\begin{align}\label{eq:ConvertingBasis}
    \underbrace{[0\dots010\dots 0]}_{\text{Non zero at $a\leq r-2$}} &= \sum_{b=1}^{a}\orthb_{b}\,, &
    [0\dots0 1 0] &= \frac{1}{2}\left((\sum_{a=1}^{\rank-1} \orthb_{a})-\orthb_{r}\right)\,, &  [0\dots 0 1] &= \frac{1}{2}\sum_{a=1}^{\rank} \orthb_{a}\,.
\end{align}

The weights comprising the vector representation are $\pm\orthb_{a}$, $a=1,\ldots, r$, and so we will suggestively use the set $\{1,\ldots,\rank,-\rank,\ldots,-1\}$ for the possible values of the index $i$.   $\qV_{a}$ means that $i=a$ is positive, and $\qV_{-a}$ means that $i=-a$ is negative. The square norm of a vector in the explicit index notation is $\metric^{ij}V_{i}V_{j} = 2\sum\limits_{a=1}^{\rank}\qV_a \qV_{-a}$.

While indices $\alpha,\dot\alpha$ shall assume integer values from $1$ till $2^{\rank-1}$, it will be sometimes convenient to convert $\qS_{\alpha}$ and $\qC_{\dot\alpha}$ to the Cartan notation $\qZ_{A}$ \cite{Cartan:104700}, where $A=\{a_1,\dots,a_k\}$ is a multi-index. It is defined as follows: we note that all the weights in the spinor irreps  are sums $\frac 12\sum_{a=1}^{\rank}\pm\orthb_{a}$, and then $A=\{a_1,\dots,a_k\}$ are the positions of the minus signs in the corresponding sum, $\qZ_{A}= \qC_{\dot \alpha}$ when the number of entries in $A$ is even and $\qZ_{A}=\qS_{\alpha}$ when the number of entries in $A$ is odd. For an explicit example see \eqref{eq:example}. Note that $\qZ_{\emptyset} = \qC_{1}$ and $\qZ_{r} = \qS_{1}$ are the highest-weight components of the spinors $\qC$ and $\qS$, \cf \eqref{eq:ConvertingBasis}. 

The standard QQ-relations \eqref{eq:Plucker3} involving the Q-functions on the Dynkin diagram and their immediate Weyl transforms are 
\begin{align}
    W(V_{1\dots k},V_{1\dots k-1,k+1}) 
    &= 
    V_{1\dots k+1 }V_{1\dots k-1}\,, 
    \quad 
    \nonumber
    k< r-2\,, \\
    W(V_{1\dots r-2},V_{1\dots r-3,r-1})
    &=
    V_{1\dots r-3}\psi_1 \eta_1\,, 
    \nonumber
    \\
    W(\psi_1,\psi_2) &= V_{1\dots r-2}\,, 
    \quad
    \\
    \nonumber
    W(\eta_1,\eta_2) &= V_{1\dots r-2}\,. 
\end{align}

To relate all components of spinors and vectors, we need to introduce gamma-matrices. While our expressions will be often written covariantly, we also build an explicit basis that shall be used in the character solution to the Q-system. Such a basis was constructed in Section~\ref{sec:sl4so6}, and we generalize it now to an arbitrary rank. Once again, the building blocks are the three matrices $\sigma^z,\sigma^\pm$ defined in \eqref{eq:Basic22matrices}. We use them to define the $2^\rank\times 2^\rank$ gamma-matrices as
\be\label{eq:GammaBasis}
    \Gamma_{\pm a} = -\underbrace{\sigma^{z}\otimes\dots \otimes \sigma^{z}}_{a-1\text{ times}} \otimes \sigma^{\mp} \otimes \underbrace{\mathbf{1} \otimes \dots \mathbf{1}}_{\rank-a\text{ times}}\,.
\ee
The matrices $\Gamma_{\pm a}$ satisfy the anti-commutation relations
\be
    \{\Gamma_{i},\Gamma_{j}\} = \delta_{i+j,0}\,
\ee
consistent~\footnote{We use a definition of the Clifford algebra without $2$ in front of the metric} with our usage of $\metric_{ij} =\delta_{i+j,0}$. We will write $\Gamma_{I}$ with $I$ a multi-index to denote the weighted antisymmetric product, for example $\Gamma_{ij} = \frac{1}{2}(\Gamma_{i}\Gamma_{j}-\Gamma_{j}\Gamma_{i})$.

We pick
\be
    C= (-1)^{\frac{(\rank+1)(\rank+2)}{2}}\prod_{a=1}^{\rank}(\Gamma_{a}+ \Gamma_{-a})\,
\ee
as the charge conjugation matrix. The sign in front of $C$ is for future convenience. The charge conjugation matrix can be written in terms of tensor products of $\sigma^x$ and $\sigma^y$ as shown explicitly for $\so_6$ in section~\ref{sec:sl4so6}. The product of $\Gamma_{I}$ and $C$ is either a symmetric or an antisymmetric matrix: $(C\Gamma_{I})^T= (-1)^{\frac{1}{2}(r-|I|-1)(r-|I|)} C\Gamma_{I}$.

We emphasize the Weyl nature of our spinors by writing
\begin{subequations}
\begin{align}
    &\psi_{\alpha}(C\Gamma_{I})^{\alpha\beta}\psi_{\beta} = \psi_{\alpha}\gamma_{I}^{\alpha\beta}\psi_{\beta}\,,
    & 
    &\psi_{\alpha}(C\Gamma_{I})^{\alpha\dot\beta}\eta_{\dot\beta} = \psi_{ \alpha}\gamma_{I}^{ \alpha\dot\beta}\eta_{\dot\beta}\,, 
    \\
    &\eta_{\dot\alpha}(C\Gamma_{I})^{\dot\alpha\beta}\psi_{\beta} = \eta_{ \dot\alpha}\bar{\gamma}_{I}^{\dot\alpha \beta}\psi_{\beta}\,
    & 
    &\eta_{\dot \alpha}(C\Gamma_{I})^{\dot \alpha\dot \beta}\eta_{\dot \beta} = \eta_{\dot \alpha}\bar{\gamma}_{I}^{\dot\alpha\dot\beta}\eta_{\dot \beta}\,,     
\end{align}
\end{subequations}
and use $C$ to raise and lower indices so that for example
\begin{equation*}
    \psi^{\alpha}\psi_{\alpha}= \psi_{\beta}C^{\beta\alpha}\psi_{\alpha}\,.
\end{equation*}
One can contract spinors of the same type for even rank and spinors of the opposite type for odd rank.

\subsection{Isotropic spaces and fused flag}\label{sec:IsotropicSpFl}
Recall that $\gamma = e^{\frac{2\pi \ii}{\Cox}}$, where $\Cox=2\rank-2$ is the Coxeter number of $\so_{2\rank}$. To derive various relations between Q-functions, we need explicitly the Perron-Frobenius vector of the incidence matrix:
\begin{align}
\mu_a &=[a]_{\gamma^{1/2}}\,,\quad a \leq r-2\,;
&
\mu_{r-1}&=\mu_r=\frac 12 [r-1]_{\gamma^{1/2}}\,.
\end{align}

We focus first on the tensor product of two vectors, its decomposition into irreps is $\qV_{(1)}\otimes \qV_{(1)} = \text{Sym}^2(\qV_{(1)})+\wedge^2 \qV_{(1)} + \mathbf{1}$.  Recall that $\mu_{1}$ is the maximal eigenvalue of \eqref{eq:Ladef2} in the vector representation and $\Psi_{(1)}$ is the corresponding $S$-solution of \eqref{eq:dA}. The associated eigenvalues of $\Psi_{(1)}^{[m]}\otimes \Psi_{(1)}^{[-m]}$ are then $\gamma^{m/2}\mu_{1}+\gamma^{-m/2}\mu_1$. We shall not be interested in the projection to the symmetric part. For the antisymmetric part and $m=1$ we get a fusion relation between $V_{(1)}^+ \wedge V^{-}_{(1)}$ and $\qV_{(2)}$. In components it reads
\be
    V_{ij} = W(\qV_i,\qV_j)\,.
\ee
The trivial representation has eigenvalue zero. Since $\mu_1>0$ it follows that the projection of $\Psi_{(1)}^{[m]}\otimes \Psi_{(1)}^{[-m]}$ to the singlet must vanish until the factors of $\gamma^{m/2}$ and $\gamma^{-m/2}$ cancel with each other; this gives a set of projection relations and one quantisation condition:
\begin{align}
\label{eq:SO2NVProj}
    \qV^{[m]}_i (\qV^i)^{[-m]} &=0\,,\quad m=0,1\dots, \frac{\Cox}{2}-1\,;  & \qV^{[\frac{\Cox}{2}]}_i (\qV^i)^{[-\frac{\Cox}{2}]} &=1\,.
\end{align}
Much like $\qV_{(2)}$, other antisymmetric tensors can be all expressed using $\qV_{(1)}$ as
\be
\label{eq:VWdet}
\qV_{i_1\dots i_a} &=& W(\qV_{i_1},\dots,\qV_{i_a})\,,\quad a \leq \rank\,.
\ee
It follows from the explicit form of $\qV_I$ that they satisfy QQ-relations analogous to \eqref{eq:Plucker5}
\be
    W(\qV_{Ii},\qV_{Ij}) = \qV_{Iij}\qV_{I}\,.
\ee
Since $\qV_{(a)}$ is an antisymmetrisation of $a$ vectors, this form describes an $a$-dimensional hyperplane in $\mathbb{C}^{2\rank}$, however, in contrast to the Q-system in the $\sl_{\rank+1}$ case, the hyperplane is not arbitrary: For $a\leq \rank-1$, it is spanned by vectors that have vanishing inner product with each other. Such a hyperplane is, by definition, an isotropic space \cite{Cartan:104700}.

In the case of $\groupG=\SO_{2\rank}$ the complete flag $\groupG/\groupB$ has the structure
\be\label{eq:GB}
    \groupG/\groupB=\{W_1\subset W_{2}\subset \dots \subset W_{\rank-1}\subset \mathbb{C}^{2\rank}, \langle W_a,W_a\rangle=0\}\,,
\ee
where $W_a$ is an $a$-dimensional hyperplane and $\langle W_a,W_a\rangle=0$ means that it is isotropic. From the explicit expressions for $\qV_{(a)}$ \eqref{eq:VWdet} and from the projection properties \eqref{eq:SO2NVProj}, it is seen that $\{\qV^{[p_1]}_{(1)},\qV^{[p_2]}_{(2)},\dots,\qV^{[p_{\rank-1}]}_{(\rank-1)}\}\in \groupG/\groupB$ for all $p_{a}$ such that $p_{a}-p_{a+1}=\pm 1$, $a=1,\ldots,\rank-2$.  This is the structure of the fused flag for $\so_{2\rank}$.  We can also split $V_{(r-1)}$ into fermions ``$V_{(r-1)}=\psi \eta$'' as described in the next subsection, and then, for the Q-functions to form a complete flag, the  fermions should be shifted $\psi^{[p_{r-1}]}$, $\eta^{[p_r]}$ in a way that the conditions $p_r-p_{r-2}=\pm 1$ and $p_{r-1}-p_{r-2}=\pm 1$ are satisfied~\footnote{\label{ftnote} This includes a curious case with $p_{r-1}\neq p_{r}$. \label{laref} $\psi^+\otimes \eta^-$ projected to the irrep $L(\omega_r+\omega_{r-1})$ (the same one $V_{(r-1)}$ belongs to) is not equal to $V_{(r-1)}$, however it corresponds to an $S^*$-solution of the linear problem \eqref{eq:dA} associated to an eigenvalue of $\Lambda$ on the next to the maximal concentric circle of the Coxeter plane, and it is also uniquely defined. Let us call it $V_{(r-1)}^*$, it is a Q-function of new type such that $\{\qV^{[p_1]}_{(1)},\qV^{[p_2]}_{(2)},\dots,\qV^{[p_{\rank-2}]}_{(\rank-2)},\qV^{*[p_{\rank-2}]}_{(\rank-1)}\}\in \groupG/\groupB$.}.

Let us now consider \eqref{eq:VWdet} for $a=\rank$. $V_{(r)}$ belongs to the representation $\Lambda^r L(\omega_1)$ which is not irreducible. It is a direct sum of self- and anti-self dual irreps that are isomorphic, respectively, to $L(2\omega_{\rank})$ and $L(2\omega_{\rank-1})$. These irreps have isotropic planes in the $\groupG$-orbits of their highest-weight vectors, however specifying these planes is not necessary for defining a point in $\groupG/\groupB$, the information provided by \eqref{eq:GB} suffices. At the same time, the hyperplane $V_{(r)}$ is not isotropic and its projections $V_{(\pm),(r)}$ to the irreps are not hyperplanes meaning in particular that they are not on the mentioned $\groupG$-orbits. 

\subsection{Pure spinors and fused Fierz relations}
There is another way to explore the extended Q-system for $\so_{2\rank}$ by building it up from the spinors. In this case we study the linear problem \eqref{eq:dA} to explore tensor products between the spinors
\be\label{eq:tpr}
   \psi^{[m]} \otimes \psi^{[-m]}, \quad \psi^{[m]} \otimes \eta^{[-m]}, \quad
   \eta^{[m]} \otimes \eta^{[-m]}\,.
\ee
As exemplified in the case of $\mathfrak{su}(4)\simeq \so_6$, we compare the eigenvalues of this linear problem with those for antisymmetric powers of the vector representation. To this end one exploits the following relation
\be\label{eq:PFR}
 \mu_{r-1}(\gamma^{m/2}+\gamma^{-m/2})=\mu_{r-1-m}
\ee
which is immediate to verify using that $\gamma^{\pm\frac{r-1}{2}}=\pm\ii\,.$ 

We observe the following pattern: if one wants to project the tensor product of spinors to the fundamental representation $L(\omega_a)$, the shift in \eqref{eq:tpr} should be $m=r-1-a$. If the shift is smaller than this value, projection to $L(\omega_a)$ vanishes. Graphically on the Dynkin diagram, this means that we start at the bifurcation node for $m=1$ and walk away from the spinor nodes by increasing $m$.  For an even rank, we demonstrate this observation by the diagram in Figure~\ref{fig:Dn}. The case of an odd rank is obtained by swaps $\qS \otimes \qS \Longleftrightarrow \qS \otimes \qC$ starting from the vector node of the diagram, we have for example $V_{(1)}=\qS^{[-r+2]}\gamma_{(1)}\qS^{[r-2]}$.
\begin{figure}
\begin{center}
\begin{tikzpicture}[thick, scale=2.5]

\draw[black,thick](-4,0)--(-2,0);   
\draw[black,dashed](-2,0)--(-1,0);
\draw[black,thick](-1,0)--({1/sqrt(2)-1},{1/sqrt(2)});
\draw[black,thick](-1,0)--({1/sqrt(2)-1},{-1/sqrt(2)});

\filldraw[black,fill=white] (-4,0) circle (3pt) node[anchor=north] {};
\node[] at (-4,-0.25) {\scalebox{0.8}{$V_{(1)}$}};
\node[] at (-4,0.25) {\scalebox{0.8}{$\qS^{[-\rank+2]}\gamma_{(1)}\qC^{[\rank-2]}$}};

\filldraw[black,fill=white] (-3,0) circle (3pt) node[anchor=south] {};
\node[] at (-3,-0.25) {\scalebox{0.8}{$V_{(2)}$}};
\node[] at (-3,0.25) {\scalebox{0.8}{$\psi^{[-\rank+3]} \gamma_{(2)}\psi^{[\rank-3]}$}};

\filldraw[black,fill=white] (-2,0) circle (3pt) node[anchor=north]{}; \node[] at (-2,-0.25) {\scalebox{0.8}{$V_{(3)}$}};
\node[] at (-2,0.25) {\scalebox{0.8}{$\qS^{[-\rank+4]}\gamma_{(3)}\qC^{[\rank-4]}$}};

\filldraw[black,fill=white] (-1,0) circle (3pt) node[anchor=west] {};
\node[] at (-1-0.1,0.25) {\scalebox{0.8}{$\qS^- \gamma_{(r-2)} \qS^+$}};
\node[] at (-1,-0.25) {\scalebox{0.8}{$\qV_{(r-2)}$}};

\filldraw[black,fill=white] ({1/sqrt(2)-1},{1/sqrt(2)}) circle (3pt) node[anchor=south] {};
\node[] at ({1/sqrt(2)-1},{1/sqrt(2)+0.25}) {\scalebox{1}{$\qS$}};

\filldraw[black,fill=white] ({1/sqrt(2)-1},{-1/sqrt(2)}) circle (3pt) node[anchor=north] {};
\node[] at ({1/sqrt(2)-1},{-1/sqrt(2)-0.25}) {\scalebox{1}{$\qC$}};

\filldraw[black,dashed,fill=white] ({1/sqrt(2)-1},{0}) circle (3pt) node[anchor=north] {};
\node[] at ({1/sqrt(2)-1},{0-0.25}) {\scalebox{0.8}{$\qV_{(\rank-1)}$}};
\node[] at ({1/sqrt(2)-1},{0+0.25}) {\scalebox{0.8}{$\qS \gamma_{(\rank-1)} \qC$}};

\filldraw[black,dashed,fill=white] ({1/sqrt(2)},{0}) circle (3pt) node[anchor=north] {};
\node[] at ({1/sqrt(2)},{0-0.25}) {\scalebox{0.8}{$\qV_{(-),(\rank)}$}};
\node[] at ({1/sqrt(2)},{0+0.25}) {\scalebox{0.8}{$\qS^- \gamma_{(-),(\rank)} \qS^+$}};

\end{tikzpicture}
\end{center}
\caption{An illustration of fused Fierz identities showing a relation between tensor representations and fused products of spinors along the Dynkin diagram. The extra ``$A$-type'' nodes relate the non-fundamental representations $V_{(r-1)}$ and $V_{(r)}$ to the spinors.}
\label{fig:Dn}
\end{figure}
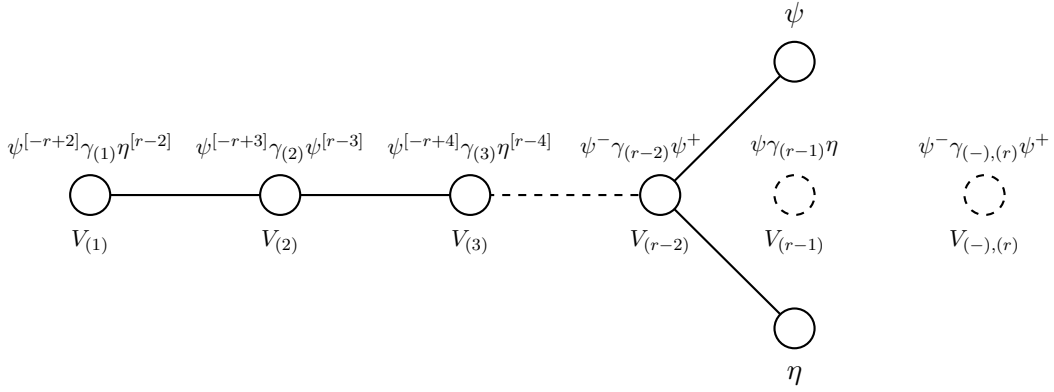

We summarize our findings for how to relate spinors to vectors using the index notation. First, we have projection relations for products of spinors of the same type
\be
    \gamma_{I}^{\alpha\beta}\qS^{[-m]}_{\alpha}  \qS^{[m]}_\beta= \bar\gamma_{I}^{\dot\alpha\dot\beta}\qC^{[-m]}_{\dot\alpha}  \qC^{[m]}_{\dot\beta} = 0\,, \quad m=0,1, \dots, \rank-2-|I| \,,
\ee
and of different type
\be
     \gamma_{I}^{\alpha\dot\beta}\qS^{[-m]}_{\alpha}  \qC^{[m]}_{\dot\beta} =\bar{\gamma}_{I}^{\dot\alpha\beta}\qC^{[-m]}_{\dot\alpha}  \qS^{[m]}_{\beta} = 0\,, \quad m=-\rank+2+|I|,\dots, \rank-2-|I|\,. 
\ee
A consequence of these projection properties is that $\qS_{\alpha}$ and $\qC_{\dot\alpha}$ are pure Weyl spinors. That is, they satisfy 
\be
    \psi_{\alpha} \gamma^{\alpha\beta}_{I} \psi_{\beta} =0\,, 
    \quad 
    \eta_{\dot \alpha} \bar\gamma^{\dot \alpha\dot \beta}_{I} \eta_{\dot \beta} =0\,, 
    \quad 
    |I|<\rank\,.
\ee
Geometrically these equations mean that $\psi_{\alpha}$ and $\eta_{\dot\alpha}$ describe two $\rank$-dimensional isotropic hyperplanes. We remark that these isotropic hyperplanes are \emph{not} equal to $V_{(\pm),(r)}$, see \eqref{eq:fused520}. Furthermore, the conditions
\be
    \psi_{\alpha}\gamma_{I}^{\alpha\dot\beta}\eta_{\dot \beta} = 0\,, \quad |I|<r-1\,
\ee
imply that the intersection of the two hyperplanes constructed from $\psi$ and $\eta$ is an $(\rank-1)$-dimensional isotropic hyperplane, this is exactly $V_{(r-1)}$. For a complete discussion regarding  pure spinors see \cite{Cartan:104700}. 

When $m=r-1-|I|$, we find fusion relations, which we call fused Fierz relations, relating the spinors with the other representations:
\begin{subequations}\label{eq:519}
\begin{align}
     V_{I}&=\gamma_{I}^{\alpha\beta}\qS^{[-r+1+|I|]}_{\alpha}  \qS^{[r-1-|I|]}_\beta = 
     \bar \gamma_{I}^{\dot\alpha\dot\beta}\qC^{[-r+1+|I|]}_{\dot\alpha}\qC^{[r-1-|I|]}_{\dot\beta}\,, & |I| &=r-2,r-4,\ldots\,,
     \\
     V_{I}&=\gamma_{I}^{\alpha\dot \alpha}\qS^{[-r+1+|I|]}_{\alpha}\qC^{[r-1-|I|]}_{\dot \alpha}
     =\bar{\gamma}_{I}^{\dot \alpha \alpha}\eta^{[-\rank + 1 +|I|]}_{\dot \alpha} \psi^{[\rank-1-|I|]}_{\alpha}
     \,, & |I| &=r-3,r-5,\ldots\,.
\end{align}
\end{subequations}
In Figure~\ref{fig:Dn}, we have also indicated further ``$A$-type'' nodes to show the effect of projecting to the non-fundamental representations and getting $V_{(r-1)}$ and $V_{(\pm),(r)}$  featured in the discussion of Section~\ref{sec:IsotropicSpFl}. The fusion relation between spinors and $\qV_{(r-1)}$ is
\begin{subequations}
\label{eq:520}
\be
     \gamma_{(r-1)}^{\alpha\dot\beta}\qS_{\alpha}\qC_{\dot\beta} = V_{(r-1)}
\ee 
and the relations for $\qV_{(+),(r)}$ and $\qV_{(-),(r)}$ are 
\be
\label{eq:fused520}
    (\bar\gamma_{(+),(r)})^{\dot\alpha\dot\beta}\eta^-_{\dot\alpha}\eta^+_{\dot\beta} = V_{(+),(r)}\,, \quad (\gamma_{(-),(r)})^{\alpha\beta}\qS^-_{\alpha}\qS^+_{\beta} = V_{(-),(r)}\,.
\ee
\end{subequations}
Finally there are quantisation conditions 
\begin{subequations}
\begin{align}
\label{eq:521}
    &(\psi^{\alpha})^{[-\frac{\Cox}{2}]}(\psi_{\alpha})^{[\frac{\Cox}{2}]} = 1\,, 
    \quad 
    (\eta^{\dot \alpha}) ^{[-\frac{\Cox}{2}]}(\eta_{\dot\alpha})^{[\frac{\Cox}{2}]} = 1\,, \quad
    \text{$r$ is even}\,, \\
    &(\psi^{\dot \alpha})^{[-\frac{\Cox}{2}]}(\eta_{\dot\alpha})^{[\frac{\Cox}{2}]} = 1\,, 
    \quad 
    (\eta^{ \alpha})^{[-\frac{\Cox}{2}]}(\psi_{\alpha})^{[\frac{\Cox}{2}]} = 1\,, \quad
    \text{$r$ is odd}\,.
\end{align}
\end{subequations}

\subsection{T-functions}\label{sec:SoTFunctions}
We proposed in Section~\ref{sec:SolvingHirota} that T-functions are to be constructed using the inner products between Q-vectors and their contra-gradient representations \eqref{eq:TfromQ}. We list here the explicit expressions for $\so_{2\rank}$. They can be used, for instance, in solution of TBA for $O(2\rank)$ sigma-model worked out in \cite{Balog:2005yz}.

For the vector representation and antisymmetric tensor representations, the T-functions are
\be\label{eq:AntiTensorCS}
    T_{a,s} = (-1)^{(\rank+1)a}\frac{1}{a!}\qV^{[\rank-1+s]}_{i_1\dots i_a} (\qV^{i_1\dots i_a})^{[-\rank+1-s]}\ , \quad 1\leq a\leq \rank-2 \ .
\ee
The off-set shift of $\frac{\Cox}{2}=\rank-1$ is determined from the projection properties \eqref{eq:SO2NVProj}. The conditions $T_{a,0}=1$ are the quantisation conditions. The extra sign appearing for even ranks can be eliminated upon redefining $g_{ij}= -\delta_{i+j,0}$, we will however stick with the original definition $g_{ij}=\delta_{i+j,0}$.

For the spinor representations, the statement is slightly rank dependent. For even $\rank$, the inner product is between spinors of the same type and the T-functions becomes 
\be
    T_{\rank-1,s} = (\qS^{\alpha})^{[-\rank+1-s]}(\qS_{\alpha})^{[\rank-1+s]}\,, \quad T_{\rank,s} = (\qC^{\dot\alpha})^{[-\rank+1-s]}(\qC_{\dot\alpha})^{[\rank-1+s]}\,,
\ee
while for odd $\rank$ we must contract the two different spinor representations with each other giving
\be
   T_{\rank-1,s} = (\qC^{\alpha})^{[-\rank+1-s]}(\qS_{\alpha})^{[\rank-1+s]}\,, \quad T_{\rank,s} = (\qS^{\dot \alpha})^{[-\rank+1-s]}(\qC_{\dot\alpha})^{[\rank-1+s]}\,. 
\ee
To check these expressions we solved the Q-system corresponding to a length-two rational spin chain with sites in the vector representation for $D_3,D_4$ and $D_5$ and computed $T_{1,1},T_{\rank-1,1}$ and $T_{\rank,1}$. We compared our result with the direct diagonalisation of the transfer matrices from \cite{Witten:1978bc} and found a perfect match up to a shift of the spectral parameter stemming from different conventions being used.

\subsection{Character solution for $\so_{2\rank}$}\label{sec:charT}
Consider first vectors, the character ansatz \eqref{eq:charsol} becomes 
\be
    \qV_i = \Ac_{i}\, x^{\frac{u}{\hbar}}_{i}\,,
\ee
where $x_{-a} = \frac{1}{x_a}$ and $\Ac_i$ are constants we wish to determine.

There exists a trick to quickly find the products $\Ac_{a}\Ac_{-a}$. The main observation is that we have $r$ conditions coming from the projection relations and  from the quantisation condition \eqref{eq:SO2NVProj}. By writing out the inner product explicitly
\be
    \qV^{[s]}_{i}(\qV^{i})^{[-s]} = \sum_{a=1}^{r}\Ac_a \Ac_{-a}(x_a^{s}+\frac{1}{x^{s}_{a}})\,,
\ee
we see that there are exactly $r$ coefficients, $\mathsf{A}_{a}\mathsf{A}_{-a}$, to fix. The solution can be found by taking a determinant ansatz, inspired by Weyl's character formula, for the sum
\be
    \qV^{[s]}_{i} (\qV^{[-s]})^{i} = \sum_{a=1}^{\rank} \mathsf{A}_{a}\mathsf{A}_{-a}(x_a^{s}+\frac{1}{x^{s}_a}) = 
    \begin{vmatrix}
    x^s_{1}+\frac{1}{x^s_1} &  &  & \hdots &  \\
    x^s_{2}+\frac{1}{x^s_2} &  &  & \hdots & \\
   \vdots & \Vec{\mathcal{W}}_{r-1} & \Vec{\mathcal{W}}_{r-2} & \hdots & \Vec{\mathcal{W}}_{1}
   \\
    x^{s}_{r}+\frac{1}{x^s_r} & & & \hdots & 
    \end{vmatrix}\,.
\ee
The projection properties fix the vectors $\Vec{\mathcal{W}}_{a}$ to be $\Vec{\mathcal{W}}_{a}\!\propto\! (x^{a-1}_{1}\!+\!\frac{1}{x^{a-1}_{1}},\dots,x^{a-1}_r\!+\!\frac{1}{x^{a-1}_r})^T$ while the quantisation condition fixes the overall normalisation. One gets then
\be\label{eq:STchar}
    T_{1,s} \equiv  (-1)^{\rank+1}\qV_i^{[s+\rank-1]}\,(\qV^i)^{[1-s-\rank]} = \frac{\det\limits_{1\leq a,b \leq \rank }(x^{\rank-b+\delta_{b,1}s}_a+x^{-(\rank-b+\delta_{b,1}s)}_a)}{\det\limits_{1\leq a,b\leq \rank}(x^{\rank-b}_a+x_a^{-(\rank-b)})} \,
\ee
which we recognize as the character of the symmetric traceless tensor representation.

From \eqref{eq:STchar} we find $\mathsf{A}_a\mathsf{A}_{-a}$ by expanding the determinant in minors:
\be
     \mathsf{A}_a\mathsf{A}_{-a} = \prod_{b\neq a}\frac{\sqrt{x_b}}{x_b-x_a}\prod_{b\neq a}\frac{\sqrt{x_b}}{x_b-\frac 1{x_a}}\,.
\ee
Now we recall that the Q-system is invariant under rescalings induced by the action of the maximal torus which allows us to set $\mathsf{A}_a$, $a=1,\ldots,\rank$, to any value. The expression for $\mathsf{A}_a\mathsf{A}_{-a}$ suggests  a natural normalisation to choose:
\be\label{eq:ParAVec}
    \Ac_{\pm a}=\prod_{b\neq a}\frac{\sqrt{x_b}}{x_b-x^{\pm 1}_a}\,.
\ee
The vector $V_i$ can be used to construct all other antisymmetric tensors. To see the pattern, start with $\qV_{ij}$ which is constructed from $W(\qV_i,\qV_j)$:
\begin{align}
    V_{ij} = \Ac_{i}\Ac_{j}\frac{1}{\sqrt{x_{i}}\sqrt{x_j}}(x_{i}-x_{j})x^{\frac{u}{\hbar}}_{i}x^{\frac{u}{\hbar}}_{j}\,. 
\end{align}
The result for $\qV_{(k)}$ is a generalisation where the last term is replaced by a Vandermonde-like determinant, 
\be\label{eq:antiCharQ}
    \qV_{i_1\dots i_k} = \left(\prod_{a=1}^{k}\mathsf{A}_{i_a}\sqrt{x_{i_a}}^{(1-k)}x^{\frac{u}{\hbar}}_{i_a}\right) \det_{1\leq a,b \leq k}(x_{i_a}^{k-b})\, .
\ee
This becomes particularly nice when all indices are positive
\be\label{eq:positiveInd}
V_A=\mathsf{A}_A x^{\frac{u}{\hbar}}_{A}\,,\quad \Ac_A=\frac{(-1)^{\frac{1}{2}|A|(|A|-1)}}{\Delta(x_A)} \times \prod\limits_{b\in \bar{A},a\in A}\frac{\sqrt{x_b}}{(x_b-x_a)}\,.
\ee
Here $\Delta(x_{A})$ is the Vandermonde determinant for $x_{A}$'s: $\Delta(x_{A}) = \prod_{a<b\in A}(x_a-x_b)$.

The expressions for T-functions $T_{a,s}$, $a\leq \rank-2$, follow from the inner product between antisymmetric tensors, see \eqref{eq:AntiTensorCS}. Just as for the vector case, we expect that these expressions are the actual characters of the corresponding representations. It is known that the character solution of $D$-type Hirota equations is \cite{Kirillov:1990,Kuniba:2010ir}
\begin{align}\label{eq:SO2nCS}
    T_{a,s} = \sum_{k_{a_0}+k_{a_0+2}+\dots +k_{a}=s}\chi(k_{a_0} L(\omega_{a_0})\oplus k_{a_0+2}L(\omega_{a_0+2}) \oplus \dots k_{a-2}L(\omega_{a-2}) \oplus k_{a}L(\omega_a)).
\end{align}
If $a$ is even, $\omega_{a_0}=\omega_{0}$ which corresponds to the trivial representation and, if  $a$ is odd, $\omega_{a_0}=\omega_{1}$. Using Mathematica we have checked numerically that for small representations and rank \eqref{eq:AntiTensorCS} and \eqref{eq:SO2nCS} agree but we have not proved in an explicit way that this is the case in general. However, since $T_{a,s}$ for $a=2,\ldots,s-2$ follow unambiguously from $T_{1,s}$ using Cherednik-Bazhanov-Reshetikhin formulae \cite{Cherednik1987sbi,Bazhanov:1989yk}, equality between \eqref{eq:AntiTensorCS} and \eqref{eq:SO2nCS} is guaranteed.

Having specified the character solution we can now use the analytic Bethe Ansatz \eqref{eq:ABA} to write down more explicitly the expressions for $T$-functions with non-trivial dependence on $\spa$. Consider the case of the vector representation and let $\qV_{i} = A_{i} x^{\frac{u}{\hbar}}_{i}{q}_{i}$, where $q_i$ is a polynomial in $\spa$ whose zeros are Bethe roots of Bethe equations and their Weyl transforms. Expanding out the inner product again allows us to write down the transfer matrices explicitly including twist as
\be
    T_{1,s} = (-1)^{(\rank+1)}\sum_{a=1}^{r}\frac{ x_a^{\rank-1+s}{{q}}^{[r-1+s]}_{a}{{q}}^{[-r+1-s]}_{-a}+x_a^{-\rank+1-s}{{q}}^{[-r+1-s]}_{a}{{q}}^{[r-1+s]}_{-a}}{\prod\limits_{b\neq a}\left(x_b+\frac 1{x_b}-x_a-\frac 1{x_a}\right)}\,.
\ee
$T_{a,s}$ for $a\leq r-2$ is obtained in the same way from \eqref{eq:AntiTensorCS}.
\newline
\newline
We now turn to spinors. The ansatz is
\be\label{eq:AnalyticalAnsatzSpinor}
    \qZ_{A} = \Bc_{A} \frac{\prod_{a=1}^{\rank}\sqrt{x_a}^{\frac{u}{\hbar}}}{\prod_{a\in A}x_a^{\frac{u}{\hbar}}}\,,
\ee
where $\Bc_A$ are the constants, with respect to $\spa$, that we want to fix. To relate spinors with vectors we will use the fusion relations \eqref{eq:fused520} between $\qV_{(\pm),(r)}$ and the fused symmetric square of $\qS_\alpha$ and $\qC_{\dot \alpha}$. In particular we have that $\eta^-_1\eta^+_1 = -\qV_{12\dots r}$ and $\psi^-_1\psi^+_1 =- \qV_{12\dots r-1,-r}$. Using the explicit form of $\qV_{I}$ this gives
\be\label{eq:qZ0Vr}
    \qC^-_1 \qC^+_1   =\frac{(-1)^{\frac{1}{2}r(r-1)+1}}{\Vandermonde}\prod_{a=1}^{\rank}x^{\frac{u}{\hbar}}_a\,, \quad \qS^-_{1}\qS^+_1  =\frac{(-1)^{\frac{1}{2}r(r-1)+1}}{\Vandermonde} x^{\rank-1}_{\rank}x^{-\frac{u}{\hbar}}_{\rank}\prod_{a=1}^{\rank-1} x^{\frac{u}{\hbar}}_a\,,
\ee
where $\Delta$ is the Vandermonde determinant. From \eqref{eq:AnalyticalAnsatzSpinor} we see that the shifts of the spectral parameter cancel each other and will not play a part. To get the other spinor components we act on \eqref{eq:qZ0Vr} with the elements of the Weyl group that flip signs of the vector indices in \eqref{eq:antiCharQ}~\footnote{Acting with the Weyl group can potentially induce the extra signs in the relations between Q-functions, \cf \eqref{eq:ss}, however this does not happen for this particular class of Weyl group elements.}.  After some calculations the resulting expression is
\be\label{eq:BSquared}
    \Bc^2_{A} = (-1)^{\frac{1}{2}r(r-1)+1}\bc_{A}^2, \quad \bc_{A} \equiv \frac{1}{\sqrt{\Vandermonde}}\prod_{a\in A}\sqrt{x}^{\rank-1}_{a}\prod_{a<b\in A}\frac{x_a-x_b}{x_ax_b-1}\,.
\ee
with $A$ an ordered set of indices. To take the square-root of this expression we use the equations
\be
   \gamma^{\alpha\dot\beta}_{(r-1)}\psi_{\alpha}\eta_{\dot \beta} = V_{(r-1)}\,.
\ee
These equations fix all spinors up to an overall factor, \cf discussion after \eqref{eq:Pluco}. In reality we do not actually need to use all the equations contained in the above expression, it is enough to focus on the Weyl orbit of the highest weight. 

The overall normalization is then finally fixed, up to a sign, from
\be
    W(\psi_1,\psi_2) = V_{12\dots r-2}\,.
\ee
We omit the algebra and state the final result
\be
    \Bc_A=\pm(-1)^{|A|}\ii^{\frac{1}{2}\rank(\rank-1)+1}\bc_A\,.
\ee

\section{Exceptional algebras}\label{sec:ExceptionalAlgebras}
Study of exceptional cases emphasises strongly that a fused flag is a non-trivially constrained system compared to an ordinary bundle with the flag manifold in the fiber. To define locally a section of an ordinary bundle, we need as many functions as $\dim\groupG/\groupB$, whereas local definition of a fused flag could use, in principle, only as many functions as the rank of the algebra. This $\rank$-dimentional functional freedom is implicit in a covariant description coming through a variety of relations  satisfied by the Q-vectors. Representation theory of exceptional algebras is very rich producing many remarkable such relations. Below we list some of them, however, without doubt, it is only a tip of an iceberg.

The notations for Q-vectors follow the enumeration for the nodes of Dynkin diagrams shown in Fig~\ref{fig:Egraphs}. To explore possible relations, we used LieArt 2.0 package \cite{Feger:2019tvk} and the explicit knowledge of $\Lambda$-eigenvalues following from the results of Section~\ref{sec:33}.
\begin{figure}[t]
\begin{minipage}[c]{0.3\hsize}
            \begin{tikzpicture}[scale=.3]
                \tikzset{node/.style={draw,circle,thick,inner sep=0pt,minimum size=7pt}}
                \draw (-1,0) node[anchor=east]  {$E_6:$~};
                \node[node, label=below:$1$] (n1) at (0, 0) {};
                \node[node, label=below:$2$] (n2) at (2, 0) {};
                \node[node, label=below:$3$] (n3) at (4, 0) {};
                \node[node, label=below:$4$] (n4) at (6, 0) {};
                \node[node, label=below:$5$] (n5) at (8, 0) {};
                \node[node, label=above:$6$] (n6) at (4, 2) {};
                \draw[thick] (n1) -- (n2) -- (n3) -- (n4) -- (n5);
                \draw[thick] (n3) -- (n6);
            \end{tikzpicture}
        \end{minipage}
        \begin{minipage}[c]{0.3\hsize}
            \begin{tikzpicture}[scale=.3]
                \tikzset{node/.style={draw,circle,thick,inner sep=0pt,minimum size=7pt}}
                \draw (-1,0) node[anchor=east]  {$E_7:$~};
                \node[node, label=below:$1$] (n1) at (0, 0) {};
                \node[node, label=below:$2$] (n2) at (2, 0) {};
                \node[node, label=below:$3$] (n3) at (4, 0) {};
                \node[node, label=below:$4$] (n4) at (6, 0) {};
                \node[node, label=below:$5$] (n5) at (8, 0) {};
                \node[node, label=below:$6$] (n6) at (10, 0) {};
                \node[node, label=above:$7$] (n7) at (4, 2) {};
                \draw[thick] (n1) -- (n2) -- (n3) -- (n4) -- (n5) -- (n6);
                \draw[thick] (n3) -- (n7);
            \end{tikzpicture}
        \end{minipage}~~~~
        \begin{minipage}[c]{0.3\hsize}
            \begin{tikzpicture}[scale=.3]
                \tikzset{node/.style={draw,circle,thick,inner sep=0pt,minimum size=7pt}}
                \draw (-1,0) node[anchor=east]  {$E_8:$~};
                \node[node, label=below:$1$] (n1) at (0, 0) {};
                \node[node, label=below:$2$] (n2) at (2, 0) {};
                \node[node, label=below:$3$] (n3) at (4, 0) {};
                \node[node, label=below:$4$] (n4) at (6, 0) {};
                \node[node, label=below:$5$] (n5) at (8, 0) {};
                \node[node, label=below:$6$] (n6) at (10, 0) {};
                \node[node, label=below:$7$] (n7) at (12, 0) {};
                \node[node, label=above:$8$] (n8) at (4, 2) {};
                \draw[thick] (n1) -- (n2) -- (n3) -- (n4) -- (n5) -- (n6) -- (n7);
                \draw[thick] (n3) -- (n8);
            \end{tikzpicture}
        \end{minipage}
            \caption{Enumration of Dynkin nodes for exceptional Lie algebras.}
    \label{fig:Egraphs}
\end{figure}
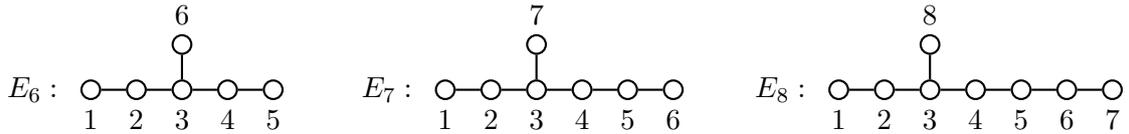

To keep the presentation short, we use the following convention. Expressions of type
\be
\label{eq:2051}
Q_{(a)}^{[m_0]}\otimes Q_{(b)}^{[-m_0]}\otimes \ldots \to {\rm \rhs}
\ee
for a fixed integer $m_0$ means that the \rhs is in an irrep of the Lie algebra and one gets an equality between \lhs and \rhs by restricting the \lhs to this irrep. If we write $=$ instead of $\to$ then this means that the \lhs is also an irrep. 

The fusion relations \eqref{eq:2051} come always with the associated projection relations: If we consider $Q_{a}^{[m]}\otimes Q_{b}^{[-m]}\otimes \ldots$ with $0\leq m<m_0$ then the restriction of this expression to the irrep of the \rhs is zero. We won't write the projection relations explicitly.

\subsection{$E_6$}
\label{sec:E6}
This is the Lie algebra of dimension $\bf{78}$, with Coxeter number of the associated Weyl group $\Cox=12$.

$E_6$ is the only exceptional algebra which has representations that are not the same as their contragradients. The contragradient representation is obtained by the reflection of the Dynkin diagram and so $L(\omega_1)^*=L(\omega_5)$, and $L(\omega_2)^*=L(\omega_4)$. Hence computation of transfer matrices involves pairing of different Q-functions, for instance
\begin{align}
T_{1,s}&=\langle Q_{(1)}^{[s+6]}, Q_{(5)}^{[-s-6]} \rangle\,,
&
T_{5,s}&=\langle Q_{(5)}^{[s+6]}, Q_{(1)}^{[-s-6]} \rangle\,.
\end{align}

The 27-dimensional fundamental representation $L(\omega_1)$ has all its components on the Weyl orbit of the highest weight. This representation and its conjugate are analogs of the vector representations for algebras from classical series, in particular in the sense that the Q-vectors at other nodes of the Dynkin diagram can be obtained using familiar Wronskian formulae (with no projections to irreps needed):

\begin{subequations}
\begin{align}
Q_{(1)}^+\wedge Q_{(1)}^- &= Q_{(2)} \,,
&
Q_{(1)}^{[2]}\wedge Q_{(1)}\wedge Q_{(1)}^{[-2]} &= Q_{(3)} \,,
\\
Q_{(5)}^+\wedge Q_{(5)}^- &= Q_{(4)} \,,
&
Q_{(5)}^{[2]}\wedge Q_{(5)}\wedge Q_{(5)}^{[-2]} &= Q_{(3)} \,.
\end{align}

\end{subequations}
Hence, we can use an embedding of lines into planes intuition to describe (at least partial) flags, however these lines are special: tensor powers of $L(\omega_1)$ have several irreps ~\footnote{Cube of $L(\omega_1)$ contains in total 10 irreps. We show only two for simplicity.} 
\begin{subequations}
\be
L(\omega_1)\otimes L(\omega_1) &=& L(\omega_{\rm max}=2\omega_1)+L(\omega_1)+L(\omega_2)\,,
\\
L(\omega_1)\otimes L(\omega_1)\otimes L(\omega_1)&=& L(\omega_{\rm max}=3\omega_1)+\ldots +L(0)+\ldots
\ee
\end{subequations}
and projection to all of them, except for the maximal ones,  of the corresponding tensor products of Q-functions is  zero. The projection relations are analogs of the null-vector/pure spinor conditions. In the fused flag, these projections are paired with the following fusion properties:
\begin{subequations}
\be
Q_{(1)}^{[\pm 6]} \otimes Q_{(5)}^{[\mp 6]} &\to& 1\,,
\ee
\begin{align}
Q_{(1)}^{[4]} \otimes Q_{(1)}^{[-4]} &\to Q_{(1)}\,, &
Q_{(1)}^{[8]} \otimes Q_{(1)}^{[-8]} \otimes Q_{(1)} &\to 1\,,
\\
Q_{(5)}^{[4]} \otimes Q_{(5)}^{[-4]} &\to Q_{(5)}\,, &
Q_{(5)}^{[8]} \otimes Q_{(5)}^{[-8]} \otimes Q_{(5)} &\to 1\,.
\end{align}
\end{subequations}

Furthermore, there is a Fierz-type relation to get $Q_{(6)}$
\be
    Q_{(1)}^{[\pm 3]} \otimes Q_{(5)}^{[\mp 3]} &\to Q_{(6)}\,.
\ee

Representation $L(\omega_6)$ is the adjoint representation of $E_6$. Hence $L(\omega_6)\wedge L(\omega_6)$ is definitely reducible. Indeed, for {\it any} simple Lie algebra with commutation relations $[J^i,J^j]=f^{ij}_k J^k$, $L_{\rm adj}\wedge L_{\rm adj}$ contains $L_{\rm adj}$ as an irrep spanned by $f_{ij}^k J^i\otimes J^j$, where raising/lowering of indices is done by the Killing form. If $h$ is the Coxeter number and $Q_{(\rm adj)}$ is the Q-vector built from the $S$-solution of \eqref{eq:dA} with the maximal positive eigenvalue of $\Lambda_{\rm adj}$ then
\be
\label{eq:adj}
Q_{(\rm adj)}^{[h/3]} \wedge Q_{(\rm adj)}^{[-h/3]} \to Q_{(\rm adj)}\,,
\ee
the equation is only meaningful in the sense of $S$-solutions of $\eqref{eq:dA}$ if $h/3$ is an even number.

The first example where the adjoint  is a fundamental representation and $h/3$ is not even is $D_5$ whose Coxeter number is $\Cox=8$. For this case, $Q_{\rm (adj)}\equiv Q_{(2)}$. This Q-function originates from $\Psi_{(2)}$, where the half-rotated $\Psi_{(2)}^\pm$  are the $S$-solutions of \eqref{eq:dA} corresponding to the complex eigenvalues $\gamma^{\pm 1/2}\mu_2$. There is also an $S^*$-solution $\Psi_{(2)}^*$ of \eqref{eq:dA}  with the real eigenvalue $\mu_2^*$ such that $\mu_2^*/\mu_2=\sqrt{2-\sqrt{2}}$. We found that 
\begin{align}
\label{eq:D5SS}
Q_{(2)}^{[3]}\wedge Q_{(2)}^{[-3]} &\to Q_{(2)}^* &(\text{example from }D_5)\,.
\end{align}

Returning back to $E_6$, one has $L(\omega_6)\wedge L(\omega_6)=L(\omega_3)\oplus L(\omega_6)$, and the corresponding fusion relations are
\begin{subequations}
\begin{align}
Q_{(6)}^+\wedge Q_{(6)}^- &\to Q_{(3)}\,,
\\
Q_{(6)}^{[4]}\wedge Q_{(6)}^{[-4]} &\to Q_{(6)}\,.
\end{align}
\end{subequations}

Symmetric power of the adjoint representation decomposes as $S^2(L(\omega_6))=L(0)\oplus L(\omega_1+\omega_5)$,  for the symmetric trace-less part one then derives
\be
Q_{(6)}^{[3]}\otimes Q_{(6)}^{[-3]} \to (Q_{(1)}\otimes Q_{(5)})_{L(\omega_1+\omega_5)}\,.
\ee

\subsection{$E_7$}
\label{sec:E7}
This is the Lie algebra of dimension ${\bf 133}$, with Coxeter number of the associated Weyl group $\Cox=18$.

Its ``vector'' representation ${\bf 56}\equiv L(\omega_6)$ has an interesting property. Alongside with the standard quadratic invariant existing because ${\bf 56}$ is its own contra-gradient, there exists also an independent symmetric quartic invariant. There is the only way to multiply four solutions of $\eqref{eq:dA}$ in the irrep ${\bf 56}$ such they have a non-trivial cone of applicability for the constant solution. We use this to conclude that, for the quartic invariant denoted as $\langle \cdot,\cdot,\cdot,\cdot \rangle$, it should be~\footnote{For this conclusion we assumed that this invariant is non-zero for combinatorial reasons if we take this particular combinations of $Q$-functions and normalised it accordingly to get $1$ on the \rhs of \eqref{eq:2114}.}
\be\label{eq:2114}
\langle Q_{(6)}^{[9]},Q_{(6)}^{[9]},Q_{(6)}^{[-9]},Q_{(6)}^{[-9]} \rangle = 1\,.
\ee
The associated projection relations are of the form $\langle Q_{(6)}^{[s_1]},Q_{(6)}^{[s_2]},Q_{(6)}^{[s_3]},Q_{(6)}^{[s_4]}\rangle =0$ if $-9\leq s_i\leq 9$ and $s_i$ are different from those featured in \eqref{eq:2114}. On the other hand, by leaving the cone of applicability, one  constructs an entirely novel family of ``quartic transfer matrices'':
\begin{align}
\label{eq:qtm}
T_{\{s_1,s_2,s_3,s_4\}}^{[s]} &=\langle Q_{(6)}^{[s_1]},Q_{(6)}^{[s_2]},Q_{(6)}^{[s_3]},Q_{(6)}^{[s_4]}\rangle\,, & s &=\sum\limits_{i=1}^4 s_i\,.
\end{align}
In the set $\{s_1,s_2,s_3,s_4\}$, the order of the entries $s_i$ is of no importance.

The quadratic invariant, similarly to the symplectic case, is antisymmetric in its entries,  meaning that wedging the vector representation to get the other fundamentals along the bottom line of the Dynkin diagram will require the subsequent projection to the corresponding irrep:
\begin{align}
Q_{(6)}^+\wedge Q_{(6)}^- &\to Q_{(5)}\,,
\\
Q_{(6)}^{[2]}\wedge Q_{(6)}\wedge Q_{(6)}^{[-2]} &\to Q_{(4)}\,,
\\
Q_{(6)}^{[3]}\wedge Q_{(6)}^{+}\wedge Q_{(6)}^{-}\wedge Q_{(6)}^{[-3]} &\to Q_{(3)}\,.
\end{align}
Adjoint ${\bf 133}=L(\omega_1)$ sits in the symmetric square of ${\bf 56}$:
\be
Q_{(6)}^{[5]}\otimes Q_{(6)}^{[-5]} \to  Q_{(1)}\,.
\ee
We can then use 
\be
Q_{(1)}^+\wedge Q_{(1)}^- \to Q_{(2)}\,
\ee
to generate the Q-vector in $L(\omega_2)$ and a Fierz-type relation
\be
Q_{(1)}^{[3]}\otimes Q_{(6)}^{[-4]} \to Q_{(7)}\,
\ee
to get the Q-vector in $L(\omega_7)$.

We list also several other fused relations which make more direct transitions between Q-vectors in fundamental representations:
\begin{subequations}
\label{eq:821}
\begin{align}
Q_{(5)}^{[5]}\otimes Q_{(5)}^{[-5]} &\to Q_{(2)}\,,
\\
Q_{(5)}^{[2]}\otimes Q_{(5)}^{[-2]} &\to Q_{(3)}\,,
\\
Q_{(4)}^{[5]}\otimes Q_{(4)}^{[-5]} &\to Q_{(3)}\,,
\\
Q_{(4)}^{[7]}\otimes Q_{(4)}^{[-7]} &\to Q_{(5)}\,,
\\
Q_{(1)}^{[11]}\otimes Q_{(7)}^{[-4]} &\to Q_{(6)}\,,
\\
Q_{(1)}^{[3]}\otimes Q_{(7)}^{[-2]} &\to Q_{(4)}\,,
\\
Q_{(6)}^{[11]}\otimes Q_{(7)}^{[-3]} &\to Q_{(1)}\,.
\end{align}
\end{subequations}
Besides, there are many fusion relations featuring $S^*$-solutions, like \eqref{eq:D5SS}. We do not present them here.

\subsection{$E_8$}
\label{sec:E8}
This is the Lie algebra of dimension ${\bf 248}$, with Coxeter number of the associated Weyl group $\Cox=30$. In addition to the quadratic invariant, this algebra has octic invariant \cite{Cederwall:2007qb}, and hence one can introduce ``octic transfer matrices'', similarly to \eqref{eq:qtm}.

A unique feature of $E_8$ is that it does not posses a vector representation. The minimal nontrivial representation is the adjoint ${\bf 248}=L(\omega_7)$. Its eight-dimensional zero-weight subspace is not on the Weyl orbit of the highest weight. To generate it from the Weyl-orbit components (which is important for the proof of Theorem~\ref{thm:unie}) we use \eqref{eq:adj} which explicitly becomes
\be
\label{eq:fijk}
f^{ij}_k Q_{(7),i}^{[10]} Q_{(7),j}^{[-10]} = Q_{(7),k}\,,
\ee
where $f_{ij}^k$ are the structure constants of $E_8$. Because Cartan generators commute between themselves, zero-weight components on the \rhs of \eqref{eq:fijk} are obtained from products of the Weyl-orbit components on the \lhs of \eqref{eq:fijk}.

All the other Q-vectors can be obtained from the adjoint representation using for instance the following fusion relations
\begin{subequations}
\begin{align}
Q_{(7)}^{[1]}\otimes Q_{(7)}^{[-1]} &\to Q_{(6)}\,,
\\
Q_{(7)}^{[6]}\otimes Q_{(7)}^{[-6]} &\to Q_{(1)}\,,
\\
Q_{(6)}^{[7]}\otimes Q_{(6)}^{[-7]} &\to Q_{(5)}\,,
\\
Q_{(6)}^{[6]}\otimes Q_{(6)}^{[-6]} &\to Q_{(2)}\,,
\\
Q_{(6)}^{[2]}\otimes Q_{(6)}^{[-2]} &\to Q_{(4)}\,,
\\
Q_{(1)}^{[7]}\otimes Q_{(1)}^{[-7]} &\to Q_{(8)}\,,
\\
Q_{(8)}^{[1]}\otimes Q_{(8)}^{[-1]} &\to Q_{(3)}\,.
\end{align}
\end{subequations}

\section{Conclusions and outlook}
In this work, we introduced a concept of the extended Q-system for simply-laced Lie algebras and studied its most essential properties. Quite remarkably, Q-functions of this system form a fused flag which can be defined as follows: if $Q_{(a)}$ is a vector of \Plucker coordinates of the minimal flag $\groupG/\groupP_a$ then $\{Q_{(a)}^\pm,Q_{(a')}\}$ are \Plucker coordinates of the flag $\groupG/\groupP_{aa'}$, simultaneously for {\it both} directions of the shift if $a,a'$ are any adjacent nodes of the Dynkin diagram. Then, by Lemma~\ref{thm:lm}, $\{Q_{(1)}^{[p_1]},\ldots, Q_{(\rank)}^{[p_r]}\}$ are \Plucker coordinates of the complete flag $\groupG/\groupB$ for {\it any} choice of the Coxeter height function $p_a$ which is the definition of a fused flag used in the main text. The fused flag is gauge equivalent to an oper, one should choose and fix an arbitrary Coxeter height function $p_a$ to define the equivalence.

It is instructive to compare the extended Q-system to other collections of Baxter Q-functions. The simplest option is to choose $Q_{(a),1}$ -- the Q-functions along a Dynkin diagram. Their zeros satisfy nested Bethe equations \eqref{eq:BAEOW} which become explicitly \eqref{eq:NBAE1} for the case of spin chains. They, in general position, contain in principle all information about the spectrum of an integrable model. However, Bethe equations are not always the best system to solve in practice, and working with $Q_{(a),1}$ lacks covariance which brings bogus complexity to various computations. The functions $Q_{(a),1}$ are often supplemented with their first descendants $Q_{(a),2}$. The obtained pairs of Q-functions form the QQ-system. The advantage over Bethe equations is a polynomial-type formulation of equations on the spectrum, however QQ-relations \eqref{eq:Plucker3} include typically non-physical solutions, exception is the $\sl_2$ case. The next addition is the Q-system on the Weyl orbit, where all $Q_{(a),\sigma(1)}$ are considered. It is likely that this system already features completeness and faithfulness, \ie solutions of \eqref{eq:Plucker4} with right analytic properties of $Q_{(a),\sigma(1)}$ are in a precise bijection with eigenspaces of a maximal commutative subalgebra acting on the Hilbert space. The extended Q-system is a special further extension of the Weyl-orbit Q-system which we demonstrated to be unique. The added value of this extension is covariance of the obtained Q-vectors under the action of (the Langlands dual of) the symmetry algebra which enables concise derivations of various remarkable relations. For instance, transfer matrices can be represented by simple bilinear combinations of Q-functions \eqref{eq:Tas} which should be compared with the expansion over Young tableaux \eqref{eq:Tchar}. Of course, both expressions are eventually equivalent, the point is that \eqref{eq:Tas} is a recipe to resum \eqref{eq:Tchar} in a particular universal way applicable at once for all Kirillov-Reshetikhin representations of the auxiliary space.

Whereas we derived relations of the extended Q-system using a particular linear problem \eqref{eq:dA} and the machinery of ODE/IM correspondence, we demonstrated that the extended Q-system is a universal concept. For any given $Q_{(a),1}$ there is a unique, up to symmetries, extension to the extended Q-system. Hence formal functional freedom one can enjoy has as many independent functions as the rank of the algebra. However, demanding that Q-functions belong to a certain analyticity class significantly restricts this freedom. In Section~\ref{sec:AnalyticBetheAnsatz} we proposed an explicit ansatz for analytic structure of Q-functions describing rational, trigonometric, and elliptic spin chains, and we conjecture that all Q-functions obeying this ansatz provide a complete and faithful description of the commuting charges spectra. ``Complete'' means that the number of solutions of QQ-relations is the right one, ``faithful'' means that the algebra of Q-functions is isomorphic to the Bethe algebra. The key point of this conjecture is that {\it all} Q-functions should satisfy the ansatz and then statements are true always and not only in general position~\footnote{In the proved case of rational $\gl_N$ spin chains with nodes in the vector representation \cite{2013arXiv1303.1578M,Chernyak:2020lgw}, the only requirement is that a spin chain is a cyclic representation of Yangian.
We have a numeric evidence supporting the conjecture  for other fundamental representations. For the case of symmetric powers however, it was observed \cite{RLV} that an additional constraint on the analytic class of certain T-functions or equivalently on the supersymmetric extension \cite{Tsuboi:2009ud,Marboe:2016yyn} of the Q-system is needed.}, demanding analyticity only for $Q_{(a),1}$, or only for $Q_{(a),1}$ and $Q_{(a),2}$ would  be not always enough.
\newline
\newline
The message of our work can also be re-stated from the point of view of representation theory. In the case of quantum algebras, the representation theory is not developed to the same level as in the case of Lie algebras. One of the problems is not sufficient understanding of relations among quantum characters, for instance not all transfer matrices have known explicit expressions in terms of Q-functions. We believe that the development is hindered by attempts to express all structures through prefundamental representations corresponding to functions $Q_{(a),1}$ and suggest that simultaneous usage of all members of the extended Q-system could be beneficial for better understanding of the character ring. Curiously enough, the extended Q-system might be still not the final object. We gave some examples of functions $Q^*$, \cf \eqref{eq:D5SS} and the footnote on page~\pageref{ftnote}, which correspond to $S^*$-solutions of the linear problem \eqref{eq:dA}. They are neither members of the extended Q-system nor transfer matrices but are nevertheless well-defined. It would be really interesting to understand how these objects are interpreted in terms of quantum characters.
\newline
\newline
This work was only focused on simply-laced cases, we discussed explicitly $D_\rank$ series in detail and also exceptional algebras though more schematically. Generalisation of the extended Q-system to the non-simply laced case may be done in a ``naive'' way, simply by choosing the corresponding Cartan matrix in the described constructions \cite{MV05,Frenkel:2020iqq}. However, this approach will lead to ``wrong'' Bethe equations (from the point of view of the most common integrable models). The called-for Bethe equations arise when we consider a twisted affine Lie algebra \cite{Masoero:2015rcz}, we plan to address the question of how the extended Q-system and fused flag look like in this case in a future publication. 
\newline
\newline
Among potential applications, let us mention integrable systems which arise in the context of AdS/CFT correspondence and which are naturally based on various high-rank symmetries. To study the archetypal examples, we must generalise our findings to supersymmetric algebras. In this way we should get a new insight in description of AdS$_5$/CFT$_4$ \cite{Gromov:2013pga} and AdS$_4$/CFT$_3$ \cite{Cavaglia:2014exa} quantum spectral curves as well as new means to attempt deriving  quantum spectral curve for AdS$_3$/CFT$_2$ which one expects for various backgrounds \cite{Zarembo:2010yz}. For the AdS$_5$/CFT$_4$ case, a substantial progress was done in \cite{Kazakov:2015efa} but it is yet unclear whether we reached a complete understanding of $\glmn$ extended Q-systems, in particular one still needs to explore a supersymmetric version of Langlands duality. 

There is also a special twisting limit of AdS/CFT integrability that breaks supersymmetry and leads to fishnet \cite{Gurdogan:2015csr} and, via holography, to fishcain \cite{Gromov:2019aku} models based on (a real form of) $\so_M$ symmetry. Extended Q-systems should be directly helpful in solution of the corresponding TBA equations and separation of variables for arbitrary $M$ \cite{Basso:2019xay,Derkachov:2019tzo,Derkachov:2020zvv}.

\paragraph{Note added} When the results of this work were ready and we were preparing the paper for publication, the paper by Ferrando, Frassek and Kazakov appeared \cite{Ferrando:2020vzk}. Their results  intersect considerably with our results applied to the case of $D_\rank$ algebras. However, the methods and some of key messages of their and our work are different.

\acknowledgments
We would like to thank Luca Cassia, Dmitry Chernyak,  Katsushi Ito, Peter Koroteev, S\'ebastien Leurent, Stefano Negro,  Paul Ryan,  Konstantin Zarembo, and Anton M. Zeitlin for useful discussions. D.V. is grateful to his former students Danny Bennet, Luke Corcoran, Bl\'{a}ith\'{i}n Power, and Jessica Weitbrecht with whom he had pleasure to learn the subject of ODE/IM correspondence.

The work of the authors was supported by the Knut and Alice Wallenberg Foundation under grant ``Exact Results in Gauge and String Theories''  Dnr KAW 2015.0083.

\bibliographystyle{halpha}
\biblio
\end{document}